\documentclass[aps,pra,twocolumn,notitlepage,superscriptaddress]{revtex4-2} 

\usepackage{qcircuit}
 
\usepackage{algorithm}
\usepackage{algpseudocode}

\floatname{algorithm}{Protocol}

\usepackage{epsfig,amssymb,amsmath,mathrsfs,amsfonts,amsthm,graphicx,float,color}

\usepackage[colorlinks=true]{hyperref}
\usepackage[active]{srcltx}
\usepackage{qcircuit} 
\def\<{\langle}\def\>{\rangle}
\def\bb{\langle\!\langle}\def\kk{\rangle\!\rangle}
\def\ket#1{|#1\rangle}\def\bra#1{\langle #1|}
\def\Tr{\operatorname{Tr}}\def\:{\hbox{\bf
    :}}

\def\R{\mathbb R}

\def\C{\mathbb C}
\def\Z{\mathbb Z}

\def\Var{\mathsf{Var}}

\def\spc#1{\mathcal{#1}}

\newcommand{\ee}{\mathrm{e}}
\newcommand{\ii}{\mathrm{i}}
\newcommand{\dd}{\mathrm{d}}
\newcommand{\HH}{\mathcal{H}} 

\newcommand{\Comb}{\operatorname{Comb}}

\newcommand{\JJ}{\mathcal{J}}

\newcommand{\In}{\text{in}}
\newcommand{\out}{\text{out}}

\newtheorem{lemma}{Lemma}
\newtheorem{theorem}{Theorem}

\newtheorem{definition}{Definition}

\begin{document}
\title{Optimal quantum metrology under energy constraints}
 
\author{Longyun Chen}
\author{Yuxiang Yang}
\affiliation{ QICI Quantum Information and Computation Initiative, School of Computing and Data Science, The University of Hong Kong, Pokfulam Road, Hong Kong SAR, China}

\begin{abstract} 
The traditional framework of quantum metrology commonly assumes unlimited access to resources, overlooking resource constraints in realistic scenarios. As such, the optimal strategies therein can be infeasible in practice. Here, we investigate quantum metrology where the total energy consumption of the probe state preparation, intermediate control operations, and the final measurement is subject to a constraint. 
We establish a comprehensive theoretical framework for characterizing energy-constrained multi-step quantum processes, based on which we develop a general optimization method for energy-constrained quantum metrology that determines both the optimal precision and the corresponding strategy. 
Using the method, we determine the ultimate precision limit of energy-constrained phase estimation and identify a novel advantage of quantum superpositions of causal orders
in enhancing the energy efficiency of adaptive quantum estimation.


\end{abstract}
 
\maketitle

\noindent{\it Introduction.}  
Quantum metrology exploits quantum effects to achieve the ultimate precision in parameter estimation. The field has demonstrated the potential of sensing in the regime of the Heisenberg limit, where quantum-enhanced protocols attain a quadratic advantage in precision over classical counterparts \cite{giovannettiQuantumEnhancedMeasurementsBeating2004a,giovannettiAdvancesQuantumMetrology2011,helstromQuantumDetectionEstimation1969,holevoProbabilisticStatisticalAspects2011,giovannettiQuantumMetrology2006}.

Traditionally, quantum metrology has assumed unlimited access to resources, allowing for the preparation of arbitrary quantum states and the implementation of any unitary evolution. This idealization, while mathematically convenient, fails to capture the physical constraints present in realistic experimental scenarios. 
A more realistic consideration is to investigate the task with constrained resources, like energy. Energy constraints in quantum information processing have attracted significant attention across various research areas, including quantum computation~\cite{chiribellaFundamentalEnergyRequirement2021,meierEnergyConsumptionAdvantageQuantum2025} and quantum channel capacities~\cite{giovannettiBroadbandChannelCapacities2003,guhaCapacityBosonicWiretap2008,wildeEnergyconstrainedPrivateQuantum2018}.  For quantum metrology, researchers have examined the impact of energy uncertainty on quantum clock resolution~\cite{buzekOptimalQuantumClocks1999,faistTimeenergyUncertaintyRelation2023}, photon number constraints in optical systems~\cite{demkowicz-dobrzanskiQuantumLimitsOptical2015,knottPracticalQuantumMetrology2016,wolfgrammEntanglementenhancedProbingDelicate2013}, and energy-constrained frequency estimation~\cite{liuzzo-scorpoEnergyefficientQuantumFrequency2018,goreckiInterplayTimeEnergy2025}.

Despite these advances, existing studies have several limitations. Most investigations~\cite{giovannettiBroadbandChannelCapacities2003,guhaCapacityBosonicWiretap2008,wildeEnergyconstrainedPrivateQuantum2018,buzekOptimalQuantumClocks1999,demkowicz-dobrzanskiQuantumLimitsOptical2015,knottPracticalQuantumMetrology2016,wolfgrammEntanglementenhancedProbingDelicate2013,goreckiInterplayTimeEnergy2025,faistTimeenergyUncertaintyRelation2023} focus primarily on the energy cost of state preparation. However, choosing suitable measurements and intermediate control operations, including error correction~\cite{kesslerQuantumErrorCorrection2014,arradIncreasingSensingResolution2014,durImprovedQuantumMetrology2014,demkowicz-dobrzanskiAdaptiveQuantumMetrology2017,zhouAchievingHeisenbergLimit2018}, unitary inversion~\cite{yuanSequentialFeedbackScheme2016}, and dynamical decoupling~\cite{sekatskiDynamicalDecouplingLeads2016}, are essential for achieving the ultimate precision, whose energy cost cannot be overlooked. 
Crucially, to determine the ultimate precision achievable by an experimenter with limited energy, one has to regard the state preparation, control operations, and the measurement as a unified \emph{multi-step quantum process} and optimize over all such processes with energy constraints.
Such energy consumption of the whole process was discussed recently in Ref.~\cite{meierEnergyConsumptionAdvantageQuantum2025}, however, its framework does not allow for determining an explicit value of the energy consumed by a generic operation. Therefore, a comprehensive framework for characterizing energy-constrained quantum processes remains to be established. 

Here, we establish a framework to characterize general multi-step processes with energy constraints, built upon the quantum comb approach for quantum metrology~\cite{chiribellaQuantumCircuitsArchitecture2008,chiribellaTheoreticalFrameworkQuantum2009,chiribellaOptimalNetworksQuantum2012,yangMemoryEffectsQuantum2019,altherrQuantumMetrologyNonMarkovian2021,liuOptimalStrategiesQuantum2023,liuFullyOptimizedQuantumMetrology2024,kurdzialekUniversalBoundsQuantum2025}. We then use it to develop a general optimization method for energy-constrained quantum metrology. 
For phase estimation, while unlimited energy allows for arbitrarily high precision by increasing dimension~\cite{kitaevQuantumMeasurementsAbelian1995,nielsenQuantumComputationQuantum2012}, our framework reveals an ultimate precision scaling of $1/E^2$ with $E$ being the energy constraint.
Extending the framework beyond fixed causality, we show an advantage of quantum superposition of causal orders \cite{hardyQuantumGravityComputers2009,chiribellaQuantumComputationsDefinite2013} in energy-constrained quantum metrology, as well as a strict hierarchy differentiating strategies and battery models, i.e., different means of energy supply.


\smallskip

\noindent{\it Energy-constrained multi-step processes.} 
To begin with, let us consider one-step processes, i.e., quantum channels. For convenience, the ground state energy is assumed to be zero. To realize the energy change throughout the process, the operation needs to interact with an external energy source, referred to as a ``battery''.
For instance, preparing a quantum state $\rho$ requires energy from the battery equal to $\Tr[H\rho]$, with $H$ being the Hamiltonian. 
For a quantum channel $\mathcal{E}$, the energy needed from the battery is at least the energy gap $\Delta E = \Tr[H_1\mathcal{E}(\rho)] - \Tr[H_0\rho]$ between the output and input states, where $H_0$ ($H_1$) is the Hamiltonian of the input (output). The energy consumption of $\mathcal{E}$ is $E(\mathcal{E}) := \max_{\rho} (\Tr[H_1\mathcal{E}(\rho)] - \Tr[H_0\rho])$.

Next, we show how to generalize 
to multi-step processes and efficiently determine the energy consumption, by employing the framework of quantum combs~\cite{chiribellaQuantumCircuitsArchitecture2008,chiribellaTheoreticalFrameworkQuantum2009}.
As shown in Fig.~\ref{fig:battery_setting}, an $N$-step quantum process with memory interacts with the system for $N$ times, taking in an input state in $\HH_{2n-2}$ and yielding an output state in $\HH_{2n-1}$ at step $n$. Crucially, the process may utilize an underlying ancilla $\rho_C$, enabling control, information preservation, and correlation between different steps.
Mathematically, it corresponds to a completely positive trace-preserving (CPTP) map from the joint input space $\bigotimes_{n=1}^N \HH_{2n-2}$ to the output space $\bigotimes_{n=1}^N \HH_{2n-1}$. Consequently, its Choi operator, obtained by sending half of a maximally entangled state through the joint input, is a positive semidefinite operator on $\bigotimes_{n=0}^{2N-1} \HH_n$.
Let $\Comb[(\HH_0,\HH_1),\dots,(\HH_{2N-2},\HH_{2N-1})]$ denote the set of all such Choi operators representing valid $N$-step processes with memory. A positive semidefinite operator $C$ on $\bigotimes_{n=0}^{2N-1} \HH_n$ belongs to this set if and only if it satisfies the following recursive conditions \cite{chiribellaQuantumCircuitsArchitecture2008,chiribellaTheoreticalFrameworkQuantum2009}:
\begin{equation}\label{eq:comb_def}
    \begin{cases}
        \Tr_{2n-1}[C^{(n)}] = C^{(n-1)}\otimes I_{2n-2}, & n=2,\dots,N, \\
        \Tr_{1}[C^{(1)}] = I_{0},
    \end{cases}
\end{equation}
where $\Tr_n[\cdot]$ and $I_n$ denote the partial trace and the identity operator on $\HH_n$, respectively, and $C^{(n)}$ is the Choi operator of the sub-process comprising the first $n$ steps.
Eq.~(\ref{eq:comb_def}) characterizes the \emph{causality} of the process. Namely, the $N$ steps occur in a fixed order. Later, we will also discuss the consequences of removing these causal constraints.
Interlacing two quantum combs will result in a new quantum comb that represents their sequential composition, which can be formally obtained by the link product~\cite{chiribellaQuantumCircuitsArchitecture2008,chiribellaTheoreticalFrameworkQuantum2009}. For two quantum combs $C_1$ on $\bigotimes_{a \in A} \HH_a$ and $C_2$ on $\bigotimes_{b \in B} \HH_b$, their link product 
\begin{equation} \label{eq:link_product}
    C_1 * C_2 := \Tr_{A \cap B} [(C_1^{T_{A \cap B}} \otimes I_{B \backslash A}) (I_{A \backslash B} \otimes C_2)]
\end{equation}
yields a quantum comb on $\bigotimes_{c \in (A \cup B) \backslash (A \cap B)} \HH_c$. Here, $C^{T_{A \cap B}}$ means the partial transpose of the operator $C$ on the space $\bigotimes_{c \in A\cap B} \HH_c$. See Appendix.~\ref{sec:quantum_comb} for more details about quantum combs.

For a generic quantum comb $C$, we determine its energy consumption in a manner analogous to the energy consumption $E(\mathcal{E})$ for a channel $\mathcal{E}$. 
Notably, there are two possible models as shown in Fig.~\ref{fig:battery_setting}: a \emph{global battery} shared across all steps, or separate \emph{local batteries} for individual steps. Although these schemes are equivalent for single-step channels, they lead to different energy consumption in multi-step processes. 
\begin{figure}[t]
    \centering
    \includegraphics[width=0.9\linewidth]{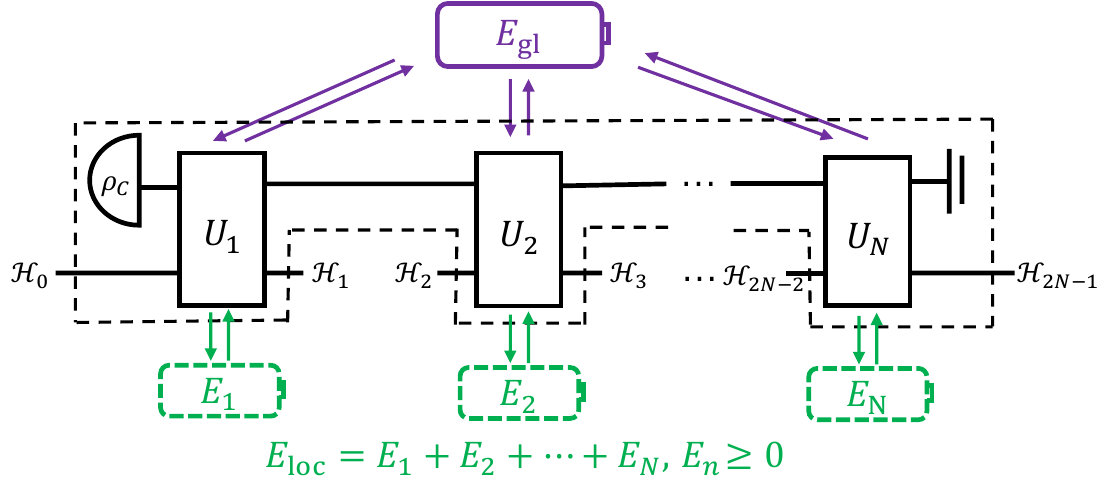}
    \caption{An $N$-step energy-constrained process, with a global battery with initial energy $E_{\text{gl}}$ shared across all steps (the purple solid lines) or $N$ isolated local batteries with initial energy $E_n \geq 0$ respectively, where the $n$-th battery exclusively interacts with the $n$-th step (the green dashed lines). Energy can be transferred bidirectionally. }
    \label{fig:battery_setting}
\end{figure}

\begin{definition}[Battery models]
    We say an $N$-step comb   
    is $E$-globally ($E$-locally) energy-constrained if there exists a global battery ($N$ local batteries) with initial (total) energy $E$: for any joint input state $\rho$ on $\bigotimes_{m=1}^N \HH_{2m-2}$ and any $n=1,\dots,N$, the total energy increase from the inputs to the outputs in the first $n$ steps does not exceed $E$ for the $E$-globally energy-constrained case, and the energy increase from the input to the output in the $n$-th step does not exceed $E_n \geq 0$ with $\sum_{n=1}^N E_n=E$ for the $E$-locally energy-constrained case.  
\end{definition}

In the above definition, we require $E_n\ge0$ for the local battery model to prevent energy released in one step from being transferred to other steps.
Note that we only consider the comb input and the comb output, while in Fig.~\ref{fig:battery_setting} the unitary operations $U_n$ act on both the input/output spaces and the control ancilla system. This simplification is justified because we can always find a unitary implementation that achieves the energy changes that are the same as the changes between the comb inputs and outputs, as will be shown in Thm.~\ref{thm:comb_implementation}.

For an $N$-step process, we first examine how energy flows within these two battery models. For a quantum comb $C$, denote the Hamiltonian of $\HH_{n}$ by $H_{n}$ and the Hamiltonian of $\HH_{n_1} \otimes \cdots \otimes \HH_{n_k}$ by $H_{n_1,\dots,n_k} = \sum_{j=1}^k H_{n_j} \otimes I$, with $I$ being the identity on a proper space. When we input a state $\rho$ from $\bigotimes_{n=1}^{N} \HH_{2n-2}$ to the comb, it yields an output state $\sigma = C * \rho$ determined by the link product in (\ref{eq:link_product}). The energy change in the $n$-th step can then be calculated as:
\begin{align}
    & \Tr[(H_{2n-1} \otimes I) \sigma] - \Tr[(H_{2n-2} \otimes I) \rho] \nonumber\\
    =& \Tr[(H_{2n-1} \otimes I) (C*\rho)] - \Tr[(H_{2n-2} \otimes I) \rho]  \label{eq:energy_change}\\
    =& \Tr\left[ \left( \Tr_{1,3,\dots,2N-1}[(H_{2n-1} \otimes I) C] - (H_{2n-2})^T \otimes I \right) \rho^T \right],\nonumber
    \end{align} 
where $I$ denotes the identity operator on a proper space. 

We denote by $E_{\text{gl/loc}}(C)$ the minimal energy consumption, i.e., the minimum of $E$ such that the $N$-step quantum comb $C$ is $E$-globally/locally energy-constrained, and by $\Comb^{\leq E}_{\text{gl/loc}}[(\HH_0,\HH_1),\dots,(\HH_{2N-2},\HH_{2N-1})]$ the set of $E$-globally/locally energy-constrained combs. When the spaces are clear from context, we simply write $\Comb^{\leq E}_{\text{gl/loc}}$. 
With (\ref{eq:energy_change}), we characterize the sets $\Comb^{\leq E}_{\text{gl}}$ and $\Comb^{\leq E}_{\text{loc}}$ for energy-constrained processes in Appendix~\ref{sec:proof_energy_constrained_comb}:
\begin{theorem}[Energy-constrained combs] \label{thm:energy_constrained_comb}
    For an $N$-step quantum comb $C$, $C \in \Comb^{\leq E}_{\text{gl/loc}}$ if and only if
        $E_{\text{gl/loc}}(C) \leq E$. Here, the energy consumption can be determined by $E_{\text{gl}}(C)=\max_{n\in[N]}\{\lambda_{\max}(O_n)\}$ and $E_{\text{loc}}(C)=\sum_{n=1}^N\max\{\lambda_{\max}(O'_n),0\}$, where  $O_n:=\Tr_{1,3,\dots,2N-1}[(H_{1,3,\dots,2n-1} \otimes I) C] - (H_{0,2,\dots,2n-2})^T \otimes I$, $O_n':=\Tr_{1,3,\dots,2N-1}[(H_{2n-1} \otimes I) C] - (H_{2n-2})^T \otimes I$, and $\lambda_{\max}(\cdot)$ denotes the maximal eigenvalue.
\end{theorem}

\smallskip

\noindent{\it Implementation of energy-constrained combs.}
A key question that follows is whether every energy-constrained quantum comb in $\Comb^{\leq E}_{\text{gl/loc}}$ can be physically realized.
Here, we give a positive answer, by showing how to build these processes using unitary operations and a control ancilla system while maintaining the prescribed energy constraints, as shown in Fig.~\ref{fig:battery_setting}.

Explicitly, we choose a control ancilla $\HH_C \simeq \HH_A^{\otimes 3}$, consisting of three identical systems $\HH_A$ satisfying $\dim \HH_A \geq \prod_{n=0}^{2N-1} \dim \HH_n $. We assume $\HH_{2n-2} \simeq \HH_{2n-1}$ (otherwise one can be extended to the same dimension as another).
With the above choice, we construct an implementation where the energy of the control ancilla remains constant throughout the entire process. Specifically, for each unitary operation $U_n$, the total energy change (encompassing both the comb input/output spaces and the control ancilla) precisely matches the energy gap between the comb input and output in the $n$-th step.

\begin{theorem}[Implementation of energy-constrained combs]
    \label{thm:comb_implementation}
    For any $C \in \Comb^{\leq E}_{\text{gl/loc}}$, there exists an implementation, as shown in Fig.~\ref{fig:battery_setting}, consisting of a sequence of unitary operations $\{U_n\}_{n=1}^N$ with $U_n$ being the operation of the $n$-th step, a control ancilla system initialized in a pure state $\rho_C$ whose energy remains constant throughout the whole process, and a global battery ($N$ local batteries) charged with energy $E$.
\end{theorem}

The proof can be found in Appendix~\ref{sec:proof_comb_implementation}. 
Note that we exclude the energy cost of initializing the control ancilla system, since the control ancilla stays ``hidden'' throughout the process and never becomes part of the output, allowing its energy to be recycled at the end.
We also reason, in the End Matter, that $\{U_n\}$ can always be implemented with energy-preserving operations and a sufficiently large quantum battery system. This aligns with the quantum comb framework, where arbitrarily large ancillary systems are inherently considered in the theoretical construction, making our battery implementation naturally compatible with the existing comb formalism.

\smallskip

\noindent{\it Metrology with limited energy.}
We now apply the formalism to quantum metrology.
As depicted in Fig.~\ref{fig:estimation_process} and Fig.~\ref{fig:causal_superposition_circuit}, given a parameterized multi-step process $C_\theta$ with an unknown parameter $\theta$ (which may represent, e.g., a quantum channel, a multi-step Markovian process, or a non-Markovian process with an inaccessible environment), the goal is to estimate $\theta\in [0,2\pi]$ by interacting with it via a process $T$ that interleaves with $C_\theta$ \cite{yangMemoryEffectsQuantum2019}, which incorporates all actions (state preparation, control, and measurement) taken by the experimenter. $T$ generalizes the probe state preparation and is thus referred to as the probe (process). Since both $C_\theta$ and $T$ can be formulated as Choi operators, we use $C_\theta$ and $T$ to denote these Choi operators for simplicity. An estimate $\hat{\theta}(x)$ of $\theta$ can be obtained by applying $T$ to $C_{\theta}$, followed by a final measurement $\mathcal{M}$ and an estimator function $\hat{\theta}$ that post-processes the measurement outcome $x$. To properly account for the energy consumption of realizing energy-non-preserving measurements, we restrict $\mathcal{M}$ to be a projection in the energy eigenbasis, thereby delegating the measurement energy consumption to  $T$.  

\begin{figure}[t]
    \centering
    \includegraphics[width=0.9\linewidth]{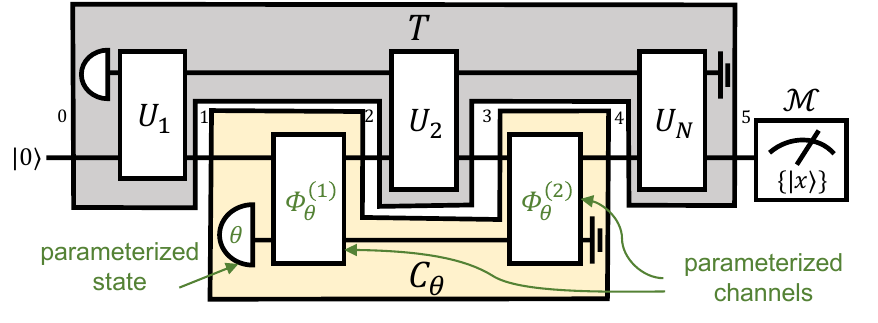}
    \caption{Our goal is to estimate $\theta$ from $C_\theta$ with an energy-constrained probe process $T$, consisting of state preparation and multi-step control, and a final measurement $\mathcal{M}$ in the energy eigenbasis $\{\ket{x}\}$.  }
    \label{fig:estimation_process}
\end{figure}

The performance of estimation is characterized by a cost function $c(\hat{\theta},\theta)$, e.g., the squared error $(\hat{\theta} - \theta)^2$ or the Holevo cost function $4\sin^2((\hat{\theta}-\theta)/2)$~\cite{holevoProbabilisticStatisticalAspects2011}. In this work, we mainly consider the Holevo cost function with a uniform prior $p(\theta) = \frac{1}{2\pi}$. Specifically, with an energy consumption of no more than $E \geq 0$, our goal is to minimize the average cost
\begin{equation}\label{eq:avg_cost}
    \begin{aligned}
        \bar{c}(T, \hat{\theta})&= \int_\theta \sum_x  4\sin^2((\hat{\theta}(x)-\theta)/2) p(x|\theta) p(\theta)\, \dd\theta 
    \end{aligned}
\end{equation}
by choosing an appropriate probe $T \in \Comb^{\leq E}_{\text{gl/loc}}$, where $p(x|\theta)$ denotes the probability of obtaining the final measurement outcome $x$.

Next, we simplify (\ref{eq:avg_cost}) to find the optimal choice of the estimator $\hat{\theta}$ for a given $T$. We combine the initial input $\ket{0}$ into $T$ and regard the space $\HH_0$ as the trivial space $\emptyset$. With quantum instrument $\{T_x\}$, the probe process $T = \sum_x T_x \otimes \ket{x}\bra{x}$ with $\{\ket{x}\}$ being energy eigenbasis. The probability can determined as $p(x|\theta) = T_x * C_\theta = \Tr[T_x C_\theta^T]$ [cf.~Eq.~(\ref{eq:link_product})]. Meanwhile, note that $4\sin^2((\hat{\theta}(x)-\theta)/2) = 2 - 2 \Re[\ee^{\ii(\hat{\theta}(x)-\theta)}]$.
Defining a task operator $\JJ := \int_{\theta} \ee^{-\ii\theta} C_{\theta}^T p(\theta) \dd\theta$,
we can rewrite \eqref{eq:avg_cost} as
\begin{equation}\label{eq:avg_cost_simplified}
        \bar{c}(T, \hat{\theta}) = 2 - 2\sum_x \Re[ \ee^{\ii \hat{\theta}(x)} \Tr[T_x \JJ]].
\end{equation}
Since our goal is to minimize $\bar{c}(T, \hat{\theta})$, the optimal estimator can be immediately obtained as $\hat{\theta}_{\rm opt}(x) = -\arg(\Tr[T_x \JJ])$. Therefore, the average cost can be defined in a form only dependent on the probe $T$ as $\bar{c}(T) := \min_{\hat{\theta}}\bar{c}(T,\hat{\theta})$.

For qubit combs, i.e., $\spc{H}_i\simeq \mathbb{C}^2$ for every $i = 1,\dots,2N-1$, the form of the optimal estimator can be made more explicit. We show that if the qubit combs $\{C_\theta\}$ satisfy $E_{\text{gl/loc}}(C_\theta) \leq 0$ and covariance, i.e., taking two arbitrary $C_\theta$ and $C_{\theta'}$, concatenating each output $\spc{H}_{2n}$ of $C_\theta$ with the input $\spc{H}_{2n-1}$ of $C_{\theta'}$ yields a new comb isomorphic to $C_{\theta+\theta'}$, then there is an optimal estimator that yields either $0$ or $\pi$. The condition $E_{\text{gl/loc}}(C_\theta) \leq 0$ holds for non-energy-generating processes, such as phase channels and amplitude-damping channels defined on the energy eigenbasis. This result is formally obtained as Thm.~\ref{thm:qubit_optimal_estimator} in the End Matter. It simplifies the problem for qubit systems, which will be helpful for our later discussion. This result does not generalize to the qudit case (see the End Matter), in contrast to the energy-unconstrained case, where the range of the optimal estimator is always $\{2\pi k/d\}_{k=0}^{d-1}$ for any $d$ when $\{C_\theta\}$ is covariant~\cite{damOptimalQuantumCircuits2007a,chiribellaOptimalEstimationGroup2005,damOptimalPhaseEstimation2007}. 

While the above analysis primarily focuses on the Bayesian approach (with a uniform prior), this framework also accommodates non-uniform prior cases as well as the frequentist approach with the Fisher information being the figure of merit. A detailed discussion can be found in the End Matter.

\smallskip

\noindent{\it Phase estimation.}
As a practical example, we apply the general results to phase estimation. We mainly focus on the global battery scenario. A similar analysis can be performed for the local battery scenario.

For phase estimation, given a $d$-dimensional system with equal-spaced Hamiltonian $H_d = \sum_{n=0}^{d-1} n \ket{n}\bra{n}$, our task is to estimate the parameter $\theta$ encoded in a phase channel diagonal in the energy eigenbasis defined as $\ket{n} \mapsto \ee^{\ii n \theta}\ket{n}$. 
In this case, the energy consumption of the probe $T$ incorporates both the cost of preparing the probe state and the cost of implementing the measurement. 
We denote by $\bar{c}_{\text{phase}}(E, d)$ the average cost $\bar{c}(T)$ optimized over $E$-globally energy-constrained probe combs $T \in \Comb^{\leq E}_{\text{gl}}$ for the $d$-dimensional phase channel.

For $d=2$, by the previous result about the optimal estimator, we can assume the range of the optimal estimator to be $\{0,\pi\}$ and exactly determine not only the minimum cost but also the optimal strategy (see Appendix~\ref{sec:proof_phase_estimation} for details). 
As $d$ increases, the minimum average cost decreases as the corresponding energy consumption increases.  
For $d\gg 1$, we consider the scaling of $\bar{c}_{\text{phase}}(E,\infty) := \lim_{d\rightarrow\infty} \bar{c}_{\text{phase}}(E,d)$ with respect to $E$. 
Using our technique, we obtain the scaling $\bar{c}_{\text{phase}}(E,\infty) \sim 1/E^2$ (see Appendix~\ref{sec:proof_high_dimension} for the proof). 
This result reveals a fundamental limitation in energy-constrained quantum metrology. Compared to the energy-unconstrained scenarios where the average cost can vanish by increasing system dimension \cite{nielsenQuantumComputationQuantum2012,kitaevQuantumMeasurementsAbelian1995}, our result demonstrates that limited energy resources inherently hinder such dimensional advantages. This finding highlights the crucial role of energy in quantum sensing strategies. It suggests that practical quantum sensors must carefully balance energy consumption against the ideal precision gained without energy constraints. It is possible that a suboptimal strategy in the ideal case could be optimal in the energy-constrained case due to higher energy efficiency.

\smallskip

\noindent{\it Energy-constrained causal superposition strategies.}  
Consider estimating $\theta$ from an $(N-1)$-step parameterized Markovian process, each step being a parameterized quantum channel $\mathcal{E}^{(n)}_\theta$. We upgrade the probe process $T$ by lifting the causal constraint (\ref{eq:comb_def}), allowing it to query $\{\mathcal{E}^{(n)}_\theta\}_{n=1}^{N-1}$ in a quantum superposition of different orders~\cite{hardyQuantumGravityComputers2009,chiribellaQuantumComputationsDefinite2013}. We assume the spaces $\{\HH_{n}\}_{n=1}^{2N-2}$ to be isomorphic (otherwise the smaller ones can be extended).

\begin{figure}[t]
    \centering
    \includegraphics[width=0.45\textwidth]{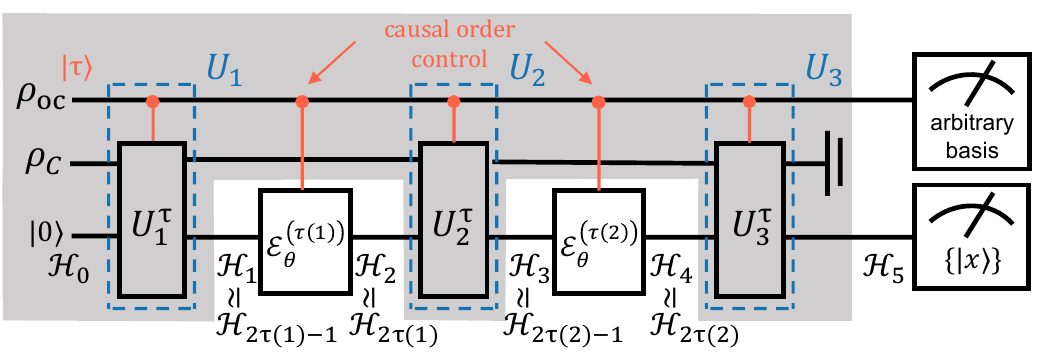}
    \caption{The quantum circuit to implement a causal superposition strategy with $N=3$.  A projective measurement in an arbitrary basis on the degenerate system $\HH_{\text{oc}}$ is allowed. The black (blue) part represents a definite causal order strategy without (with) an accessible ancilla $\HH_{\text{oc}}$. The orange part represents the strategy with quantum control of the causal order, which is a superposition of different causal orders $\tau \in S_{N-1}$.
    }
    \label{fig:causal_superposition_circuit}
\end{figure}


A causal superposition strategy is implemented by a quantum control of the causal order with an order control system $\HH_{\text{oc}}$~\cite{wechsQuantumCircuitsClassical2021}, which labels the current spacetime by giving the causal order $\tau \in S_{N-1}$ with the symmetric group $S_{N-1}$ of degree $N-1$, as Fig.~\ref{fig:causal_superposition_circuit}.
Notably, a measurement on $\HH_{\text{oc}}$ is allowed. Otherwise, the causal superposition can be simulated by a classical probabilistic mixture of the causal orders.

In the absence of energy constraints, the set of causal superposition processes, $\mathrm{Sup}$, is the collection of Choi operators on $(\bigotimes_{n=0}^{2N-1} \HH_{n}) \otimes \HH_{\text{oc}}$ satisfying 
\begin{equation} \label{eq:characterize_causal_superposition}
    \Tr_{\HH_{\text{oc}}}[T] = \sum_{\tau \in S_{N-1}} p_\tau T^\tau
\end{equation}
for an ensemble $\{p_\tau,T^\tau\}$, where $T^\tau \in \Comb^\tau:=\Comb[(\HH_0, \HH_{2\tau(1)-1}), \dots, (\HH_{2\tau(N-1)},\HH_{2N-1})]$ is a comb following the fixed causal order specified by $\tau$.

Now, we add the energy constraints. We assume the energy levels of $\HH_{\text{oc}}$ are degenerate, which means the changes in the order control system will not account for the energy consumption.
In the initialization phase of a causal superposition process, a battery is charged with energy $E$. After the spacetime branches, the battery likewise splits into multiple instances, one for each spacetime. 
If the energy flow between spacetime is not allowed, each spacetime possesses its own individual global (or local) battery, charged with energy $E$. The probe process in each spacetime must be achievable using only the energy $E$ in its respective battery. Regarding each spacetime as an individual definite causal order process, we can constrain the \emph{spacetime-individual battery} as:
\begin{equation}\label{eq:spacetime_individual_battery}
E_{\text{ind,gl/loc}}(T) := \max_{\tau\in S_{N-1}} E_{\text{gl/loc}}\left( T^\tau \right) \leq E.
\end{equation}
We denote by $\mathrm{Sup}^{\leq E}_{\text{ind,gl/loc}}$ the set of the energy-constrained causal superposition processes satisfying \eqref{eq:spacetime_individual_battery}. It is also possible to consider a model allowing the energy flow between spacetime. A detailed discussion on the distinct models can be found in Appendix~\ref {sec:separation_battery}.

\smallskip

\noindent{\it Energy efficiency advantage of causal superpositions.}
To determine the ultimate limit of energy-constrained causal superposition strategies, we replace the constraint $T \in \mathrm{Comb}^{\leq E}_{\text{gl/loc}}$ with $T \in \mathrm{Sup}^{\leq E}_{\text{ind,gl/loc}}$, and use the same method as before to minimize the average cost [cf.~\eqref{eq:avg_cost} and \eqref{eq:avg_cost_simplified}].
The preceding discussion introduces different types of strategies (definite causal order vs. causal superposition) and the battery models (global vs. local). We establish a hierarchy between them and show that the hierarchy can be made strict.

Let $\bar{c}_{\text{ind,gl}}^{\text{sup}}(E)$, $\bar{c}_{\text{gl}}^{\text{def}}(E)$, and $\bar{c}_{\text{loc}}^{\text{def}}(E)$ denote the optimal average cost for causal superposition strategies with spacetime-sharing global battery, and definite causal order strategies with global and local batteries, respectively, under an energy constraint $E$. The following hierarchy holds in general: 
\begin{equation}\label{eq:hierarchy}
\bar{c}_{\text{ind,gl}}^{\text{sup}}(E) \leq \bar{c}_{\text{gl}}^{\text{def}}(E) \leq \bar{c}_{\text{loc}}^{\text{def}}(E).
\end{equation} 
\begin{theorem}
[A strict hierarchy of energy-constrained strategies]
    There exist tasks of quantum metrology where all inequalities in (\ref{eq:hierarchy}) are strict.
\end{theorem}
\begin{figure}[t]
\centering
\includegraphics[width=0.95\linewidth]{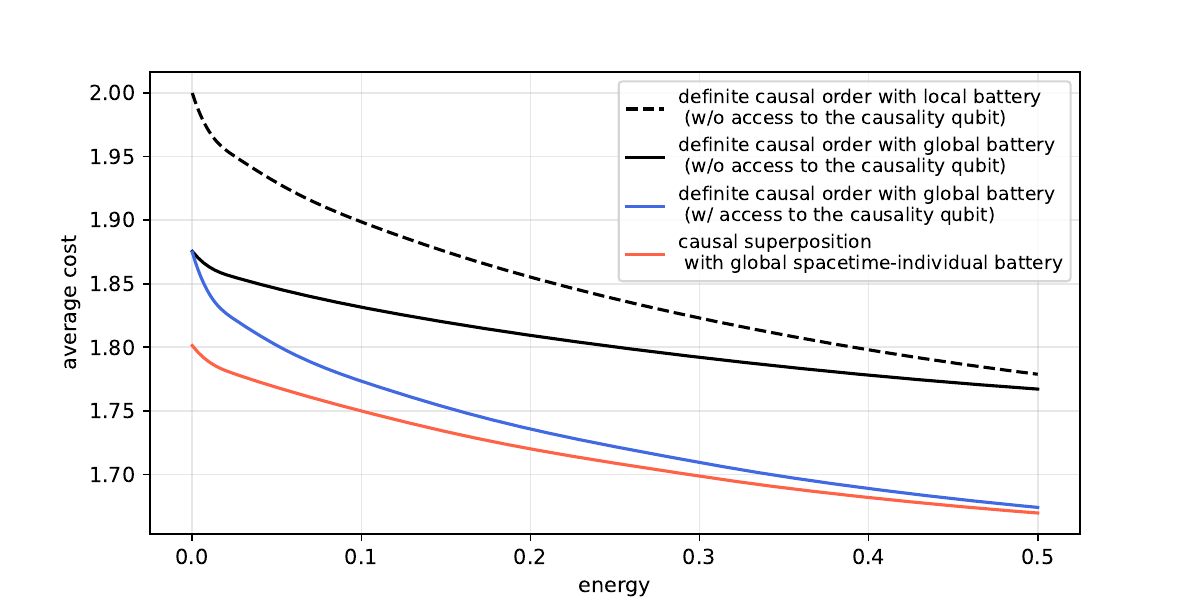}
\caption{The strict hierarchy of the optimal average cost $\bar{c}(E)$ for the different strategies and battery models.}
\label{fig:hierarchy}
\end{figure}
Recent works \cite{liuOptimalStrategiesQuantum2023,motheReassessingAdvantageIndefinite2024} discussed the advantage of causal superpositions but without energy constraints.
We prove this more general hierarchy by providing a concrete example: for the scheme in Fig.~\ref{fig:causal_superposition_circuit}, take $\mathcal{E}^{(1)}_\theta$ ($\mathcal{E}^{(2)}_\theta$) to be the phase channel $e^{-i\theta Z/2}$ ($e^{i\theta Z/2}$) preceded by a bit-flip channel $\rho \mapsto \frac{1}{2}\rho + \frac{1}{2}X\rho X$. The Hamiltonian $H = \ket{1}\bra{1}$. The minimum costs as functions of $E$ are shown in Fig. \ref{fig:hierarchy}.
To ensure a fair comparison, we also compare causal superposition strategies with definite causal order strategies with full and coherent access to a causality qubit — a resource not available in the original definite causal order setting. The gap persists, indicating that the higher energy efficiency stems from the causal superposition rather than additional degrees of freedom.

\smallskip

\noindent{\it Conclusion.}
We developed a comprehensive theoretical framework for characterizing energy-constrained quantum processes, using which we discovered several new phenomena in quantum metrology. 
For future directions, it is meaningful to investigate more  complex tasks, such as multi-parameter estimation~\cite{humphreysQuantumEnhancedMultiple2013,szczykulskaMultiparameterQuantumMetrology2016,hayashiFindingOptimalProbe2024,hayashiTightCramRRao2023,holevoProbabilisticStatisticalAspects2011,demkowicz-dobrzanskiMultiparameterEstimationQuantum2020,sidhuGeometricPerspectiveQuantum2020a} and sensing under non-Markovian noise~\cite{altherrQuantumMetrologyNonMarkovian2021,yangMemoryEffectsQuantum2019,chinQuantumMetrologyNonMarkovian2012}.
Our framework is not tailored for metrology or energy. Instead, it is a versatile general tool that can be used to investigate the ultimate tradeoff between performance and resources in generic quantum information processing tasks.

\smallskip

\noindent{\bf Note added.} After completion of this work, an independent work \cite{ahnefeldCoherenceResourcePhase2025} appeared on arXiv, which considered quantum metrology with limited coherence using the quantum comb technique. 

\smallskip

\noindent{\bf Acknowledgement.}  
We thank Yukuan Tao, Gerardo Adesso and Felix Ahnefeld for insightful discussions. 
This work is supported by the National Natural Science Foundation of China via the Excellent Young Scientists Fund (Hong Kong and Macau) Project 12322516, 
by the National Natural Science Foundation of China (NSFC)/Research Grants Council (RGC) Joint Research Scheme via Project N\underline{ }HKU7107/24,
by the Guangdong Provincial Quantum Science Strategic Initiative (GDZX2303007 and GDZX2403008),  
and by the Hong Kong Research Grant Council (RGC) through the General Research Fund (GRF) grant 17302724.

\smallskip 

\noindent{\bf End Matter on the physical implementation with quantum batteries.} 
From the perspective of quantum resource theory~\cite{chitambarQuantumResourceTheories2019,gourResourcesQuantumWorld2024}, the process given in Thm.~\ref{thm:comb_implementation} is still not yet guaranteed to be physically implementable: the unitary $U_n$ may not be free operations~\cite{chitambarQuantumResourceTheories2019,gourResourcesQuantumWorld2024} when energy is treated as a resource, since they are not necessarily strictly energy-preserving. What remains is to show that each $U_n$ acting on $\HH_{2n-2} \otimes \HH_C \simeq \HH_{2n-1} \otimes \HH_C$ can be implemented using an energy-preserving unitary $V_n$ that acts on both the system $\HH_{2n-2} \otimes \HH_C$ and a battery system $\HH_B$. This is to seek a free operation $V_n$ with a battery state $\beta$ in $\HH_B$, such that $\mathcal{U}_n(\cdot) = \Tr_{\HH_B} [\mathcal{V}_n(\cdot \otimes \beta)]$,
where $\mathcal{U}(\cdot) := U \cdot U^\dagger$ represents the unitary channel of $U$.
Such a physical implementation can be obtained by the construction presented in~\cite{abergCatalyticCoherence2014,chiribellaFundamentalEnergyRequirement2021}. For any unitary $U$ on a space $\HH$ with Hamiltonian $H$, if the energy levels are equally spaced, we can find a battery state $\beta$ with energy at most $\frac{\pi \|H\|}{2\sqrt{\varepsilon}}$ and an energy-preserving unitary $V$ on $\HH \otimes \HH_B$ that approximate $U$ with gate fidelity $1-\varepsilon$. A large battery then leads to $\epsilon\to 0$, enabling us to implement $U$ with arbitrarily high fidelity. This will not result in larger energy consumption, as the extra energy in the battery will not be dissipated and thus can be recycled at the end.

\smallskip

\noindent{\bf End Matter on qubit estimation.}
Here we provide a specific discussion for the optimal estimator for covariant qubit combs and the optimal strategy for qubit phase channel.

\smallskip

\noindent{\it The optimal estimator for covariant qubit combs.}
We obtain the explicit value of the optimal estimator for a covariant parameterized comb $C_\theta$ with non-positive energy consumption. We make a natural assumption that input spaces and the output spaces of $C_\theta$ are isomorphic and have the same Hamiltonian.

To clarify the meaning of the covariant combs, a parameterized comb $C_\theta$ is defined to be covariant if the concatenation $C_\theta \circ C_{\theta'}$ depicted in Fig.~\ref{fig:partial_link} is isomorphic to $C_{\theta+\theta'}$. Specifically, by the link product, for each $n \in \{1,\dots,N-1\}$, we can formally link the output space $\HH_{2n}$ of $C_\theta$ to the corresponding input space $\HH_{2n-1}$ of $C_{\theta'}$ to construct a new quantum comb $C_\theta \circ C_{\theta'}$. Here, for clarity, we use the notation $C_{\theta} \circ C_{\theta'}$ to denote this partially linked comb instead of the usual notation adopted for the link product. Since the Hilbert spaces considered here are isomorphic, the constructed comb $C_\theta \circ C_{\theta'}$ can be regarded as a comb in the set $\Comb[(\HH_1,\HH_2),\dots,(\HH_{2N-3},\HH_{2N-2})]$. Therefore, we can say if $C_\theta \circ C_{\theta'}$ coincides with $C_{\theta+\theta'}$ under the isomorphism.

\begin{figure}[htbp]
    \centering
    \includegraphics[width=0.2\textwidth]{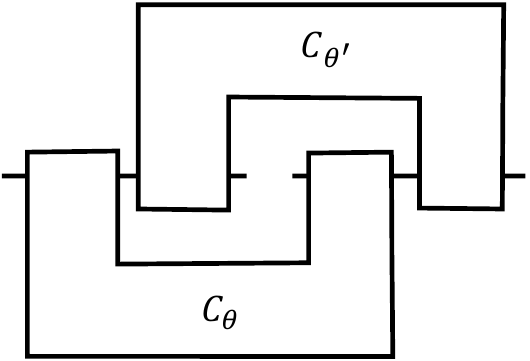}
    \caption{Given two quantum combs $C_\theta$ and $C_{\theta'}$, we can concatenate each output of $C_\theta$ with each input of $C_{\theta'}$ to form a new comb denoted by $C_\theta \circ C_{\theta'}$.}
    \label{fig:partial_link}
\end{figure}

\begin{definition}
    A parameterized comb $C_\theta$ is covariant if and only if the concatenation $C_\theta \circ C_{\theta'}$ is isomorphic to $C_{\theta+\theta'}$ for all $\theta,\theta'$.
\end{definition}

If a parameterized comb $C_\theta$ on qubit is covariant and has non-positive energy consumption, we prove that there always exists an optimal estimator with range in $\{0,\pi\}$, based on $\hat{\theta}_{\rm opt}(x) = -\arg(\Tr[T_x \JJ])$.
\begin{theorem}[Optimal estimator for covariant comb with non-positive energy consumption on qubit]\label{thm:qubit_optimal_estimator}
    Suppose the input spaces and the output spaces are isomorphic and have the same Hamiltonian. Let $C_\theta$ be a parameterized comb on a qubit. If $C_\theta$ is covariant and its energy consumption $E_{\text{gl/loc}}(C_\theta) \leq 0$, there always exists an optimal estimator with range in $\{0,\pi\}$ for the uniform-prior average Holevo cost as a figure of merit.
\end{theorem}

Its proof can be found in Appendix~\ref{sec:proof_qubit_optimal_estimator}. 
It is natural to ask whether this result can be extended to a general $d$-dimensional case for $d > 2$. Unfortunately, if we still assign the estimators with values in the set $\{\frac{2\pi n}{d}\}_{n=0}^{d-1}$, this could not be optimal. As Fig.~\ref{fig:fixed_estimator} shows, such a fixed estimator can even invalidate the advantage brought by higher dimensionality.

\begin{figure}[t]
    \centering
    \includegraphics[width=0.9\linewidth]{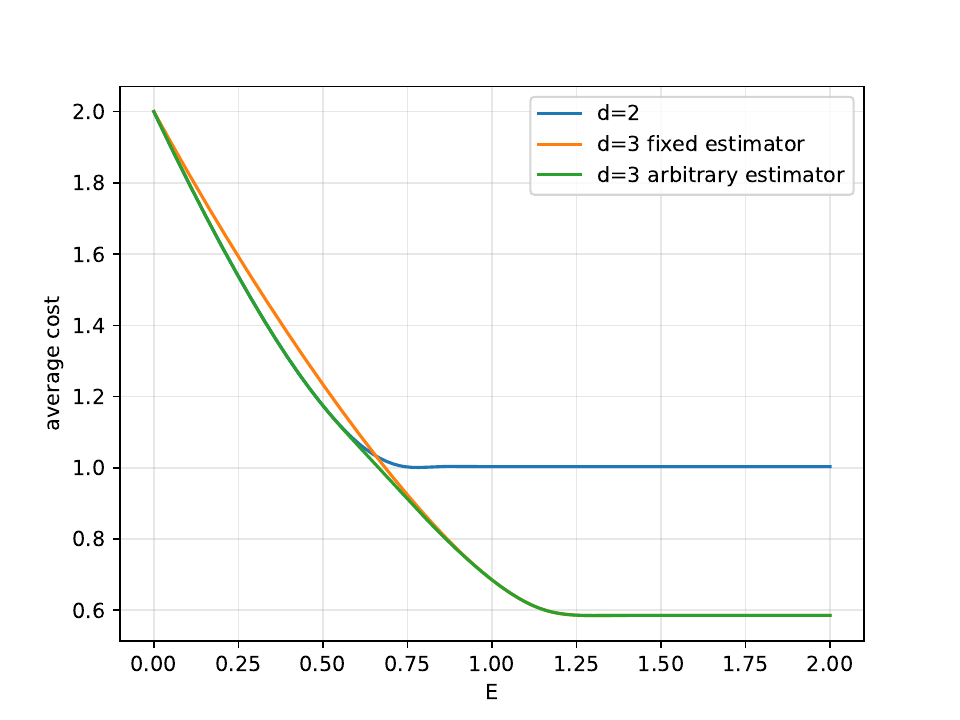}
    \caption{The optimal average cost for $E$-globally energy-constrained probe $T \in \Comb^{\leq E}_{\text{gl}}$ with different choices of estimator for a $d$-dimensional phase channel diagonal in the energy eigenbasis. The orange line is the optimal average cost with an estimator chosen in $\{0,\frac{2\pi}{3},\frac{4\pi}{3}\}$. The green line is the optimal average cost for an arbitrarily chosen estimator, which is obtained by numerically searching for the optimal one.}
    \label{fig:fixed_estimator}
\end{figure}





\smallskip

\noindent{\bf End Matter on local estimation.}
Here we show how to apply our energy-constrained framework to the local estimation. 

\smallskip

\noindent{\it Fisher information.}
In the standard quantum local estimation, the variance of any \emph{unbiased} estimator $\hat{\theta}$ for a single parameter $\theta$ obeys the quantum Cramér–Rao inequality:
\begin{equation}
    \operatorname{Var}(\hat{\theta} - \theta) \geq \frac{1}{\nu_{\text{rep}} \mathcal{I}(\rho_\theta)},
\end{equation}
where $\mathcal{I}(\rho_\theta)$ denotes the quantum Fisher information of the state $\rho_\theta = T*C_\theta$ and $\nu_{\text{rep}}$ is the repetition number of the estimation process~\cite{helstromQuantumDetectionEstimation1969,holevoProbabilisticStatisticalAspects2011,braunsteinStatisticalDistanceGeometry1994}. This bound can be asymptotically saturated for $\nu_{\text{rep}} \gg 1$.

Therefore, like the derivation for the Bayesian estimation, our goal is to maximize $\mathcal{I}(\rho_\theta)$ with $\rho_\theta = T*C_\theta$ for an energy-constrained probe $T$ satisfying $T = \sum_x T_x \otimes \ket{x}\bra{x}$, where $\{\ket{x}\}$ is the energy eigenbasis.

However, the requirement for an unbiased estimator may be too restrictive in energy-insufficient scenarios. Notably, when energy is insufficient, the optimal variance deteriorates rapidly. For instance, for the qubit phase channel, when the provided energy $E=0$, the only allowable strategy is to input a ground state into the phase channel, which captures no phase information and cannot yield an unbiased estimator. In this scenario, the variance tends to infinity when $E\to 0$ (see Appendix~\ref{sec:Fisher_information} for numerical verification). 

\smallskip

\noindent{\it Narrow prior distribution.}
An alternative approach is to treat the local estimation as a Bayesian estimation problem with a parameter drawn from a narrow prior distribution $p(\theta)$. It is worth noting that, in contrast to the Cram\`{e}r–Rao bound determined by the Fisher information, here we accommodate a broader class of estimators, including the \emph{biased} estimators. 

A possible way is to employ the optimal probe process $T$ that maximizes the Fisher information at the expected parameter value, while utilizing the prior knowledge to adapt the classical estimator $\hat{\theta}(x)$. The optimal average variance is bounded by the van Trees inequality~\cite{vantreesDetectionEstimationModulation2001,gillApplicationsVanTrees1995,elisabethVanTreesInequality2024}:
\begin{equation}
    \mathbb{E}_\theta [\mathrm{Var}(\hat{\theta} - \theta)] \geq \frac{1}{\mathcal{I}_q + \mathbb{E}_\theta [\mathcal{I}_c(\theta)]},
\end{equation}
where $\mathcal{I}_c(\theta)$ denotes the classical Fisher information of the final measurement outcome with respect to the parameter $\theta$, and $\mathcal{I}_q = \int_{p(\theta)>0} \frac{\dot{p}(\theta)^2}{p(\theta)}  d\theta$.

\begin{figure}[t]
    \centering
    \includegraphics[width=0.85\linewidth]{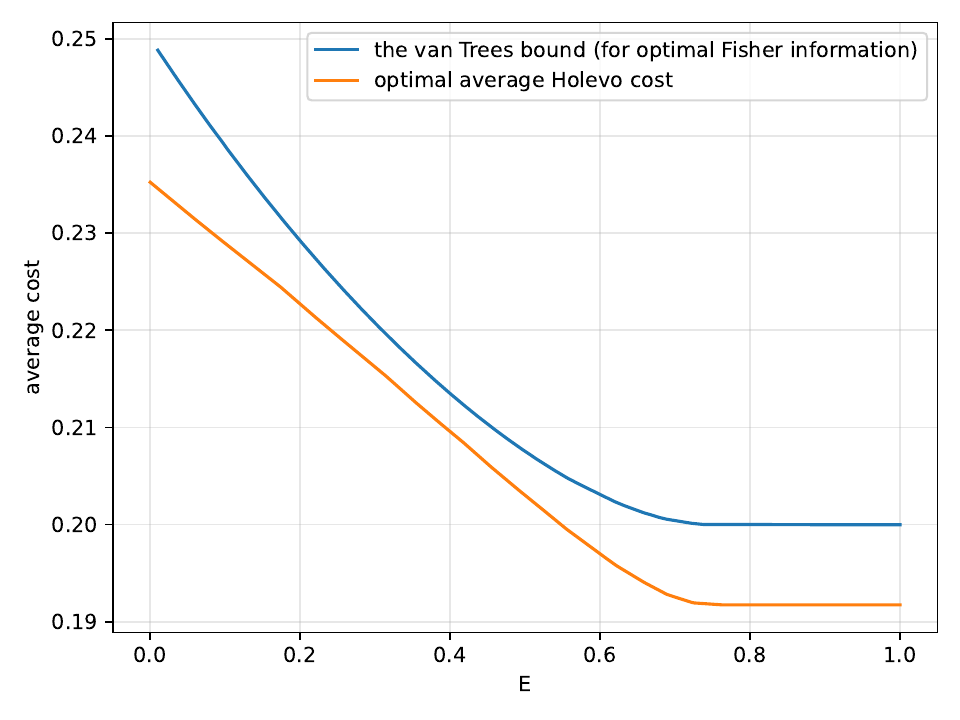}
    \caption{The optimal average cost for the energy-constrained qubit phase estimation with a narrow Gaussian prior $\theta \sim \mathcal{N}(\pi, 0.5)$. The probe process $T$ is in the $E$-globally constrained comb set $\Comb^{\leq E}_{\text{gl}}$.}
    \label{fig:narrow_distribution}
\end{figure}

Another way is to approximate the variance by the Holevo cost function. When the parameter $\theta$ lies within a small region, the variance can be approximated via the Holevo cost function as $(\theta - \theta')^2 \approx 4\sin^2\left(\frac{\theta - \theta'}{2}\right)$ for $\theta \approx \theta'$. By \eqref{eq:avg_cost_simplified}, our goal is to minimize
\begin{equation}
    2 - 2\sum_x \Re\left[\mathrm{e}^{\mathrm{i}\hat{\theta}(x)} \Tr\left[T_x \JJ\right]\right]
\end{equation}
with the task operator
$\JJ := \int_\theta \mathrm{e}^{-\mathrm{i}\theta} C_\theta^T \, p(\theta) \, \mathrm{d}\theta$
for an energy-constrained probe $T$.  

A numerical simulation is performed for the qubit phase estimation with Hamiltonian given by $H=\ket{1}\bra{1}$. Here, we assume the parameter $\theta$ is drawn from a Gaussian prior $\mathcal{N}(\pi,0.5)$ with the expectation value $\pi$ and the standard deviation $0.5$. The result is in Fig.~\ref{fig:narrow_distribution}. Although the unbiased estimator may vanish or perform poorly in the limited energy regime, it is still possible to utilize the prior knowledge to achieve an effective biased estimation.

\bibliographystyle{apsrev4-2} 
\bibliography{references}

\begin{thebibliography}{61}%
\makeatletter
\providecommand \@ifxundefined [1]{%
 \@ifx{#1\undefined}
}%
\providecommand \@ifnum [1]{%
 \ifnum #1\expandafter \@firstoftwo
 \else \expandafter \@secondoftwo
 \fi
}%
\providecommand \@ifx [1]{%
 \ifx #1\expandafter \@firstoftwo
 \else \expandafter \@secondoftwo
 \fi
}%
\providecommand \natexlab [1]{#1}%
\providecommand \enquote  [1]{``#1''}%
\providecommand \bibnamefont  [1]{#1}%
\providecommand \bibfnamefont [1]{#1}%
\providecommand \citenamefont [1]{#1}%
\providecommand \href@noop [0]{\@secondoftwo}%
\providecommand \href [0]{\begingroup \@sanitize@url \@href}%
\providecommand \@href[1]{\@@startlink{#1}\@@href}%
\providecommand \@@href[1]{\endgroup#1\@@endlink}%
\providecommand \@sanitize@url [0]{\catcode `\\12\catcode `\$12\catcode `\&12\catcode `\#12\catcode `\^12\catcode `\_12\catcode `\%12\relax}%
\providecommand \@@startlink[1]{}%
\providecommand \@@endlink[0]{}%
\providecommand \url  [0]{\begingroup\@sanitize@url \@url }%
\providecommand \@url [1]{\endgroup\@href {#1}{\urlprefix }}%
\providecommand \urlprefix  [0]{URL }%
\providecommand \Eprint [0]{\href }%
\providecommand \doibase [0]{https://doi.org/}%
\providecommand \selectlanguage [0]{\@gobble}%
\providecommand \bibinfo  [0]{\@secondoftwo}%
\providecommand \bibfield  [0]{\@secondoftwo}%
\providecommand \translation [1]{[#1]}%
\providecommand \BibitemOpen [0]{}%
\providecommand \bibitemStop [0]{}%
\providecommand \bibitemNoStop [0]{.\EOS\space}%
\providecommand \EOS [0]{\spacefactor3000\relax}%
\providecommand \BibitemShut  [1]{\csname bibitem#1\endcsname}%
\let\auto@bib@innerbib\@empty
\bibitem [{\citenamefont {Giovannetti}\ \emph {et~al.}(2004)\citenamefont {Giovannetti}, \citenamefont {Lloyd},\ and\ \citenamefont {Maccone}}]{giovannettiQuantumEnhancedMeasurementsBeating2004a}%
  \BibitemOpen
  \bibfield  {author} {\bibinfo {author} {\bibfnamefont {V.}~\bibnamefont {Giovannetti}}, \bibinfo {author} {\bibfnamefont {S.}~\bibnamefont {Lloyd}},\ and\ \bibinfo {author} {\bibfnamefont {L.}~\bibnamefont {Maccone}},\ }\href {https://doi.org/10.1126/science.1104149} {\bibfield  {journal} {\bibinfo  {journal} {Science}\ }\textbf {\bibinfo {volume} {306}},\ \bibinfo {pages} {1330} (\bibinfo {year} {2004})}\BibitemShut {NoStop}%
\bibitem [{\citenamefont {Giovannetti}\ \emph {et~al.}(2011)\citenamefont {Giovannetti}, \citenamefont {Lloyd},\ and\ \citenamefont {Maccone}}]{giovannettiAdvancesQuantumMetrology2011}%
  \BibitemOpen
  \bibfield  {author} {\bibinfo {author} {\bibfnamefont {V.}~\bibnamefont {Giovannetti}}, \bibinfo {author} {\bibfnamefont {S.}~\bibnamefont {Lloyd}},\ and\ \bibinfo {author} {\bibfnamefont {L.}~\bibnamefont {Maccone}},\ }\href {https://doi.org/10.1038/nphoton.2011.35} {\bibfield  {journal} {\bibinfo  {journal} {Nature Photonics}\ }\textbf {\bibinfo {volume} {5}},\ \bibinfo {pages} {222} (\bibinfo {year} {2011})}\BibitemShut {NoStop}%
\bibitem [{\citenamefont {Helstrom}(1969)}]{helstromQuantumDetectionEstimation1969}%
  \BibitemOpen
  \bibfield  {author} {\bibinfo {author} {\bibfnamefont {C.~W.}\ \bibnamefont {Helstrom}},\ }\href {https://doi.org/10.1007/BF01007479} {\bibfield  {journal} {\bibinfo  {journal} {Journal of Statistical Physics}\ }\textbf {\bibinfo {volume} {1}},\ \bibinfo {pages} {231} (\bibinfo {year} {1969})}\BibitemShut {NoStop}%
\bibitem [{\citenamefont {Holevo}(2011)}]{holevoProbabilisticStatisticalAspects2011}%
  \BibitemOpen
  \bibfield  {author} {\bibinfo {author} {\bibfnamefont {A.}~\bibnamefont {Holevo}},\ }\href {https://doi.org/10.1007/978-88-7642-378-9} {\emph {\bibinfo {title} {Probabilistic and {{Statistical Aspects}} of {{Quantum Theory}}}}}\ (\bibinfo  {publisher} {Edizioni della Normale},\ \bibinfo {address} {Pisa},\ \bibinfo {year} {2011})\BibitemShut {NoStop}%
\bibitem [{\citenamefont {Giovannetti}\ \emph {et~al.}(2006)\citenamefont {Giovannetti}, \citenamefont {Lloyd},\ and\ \citenamefont {Maccone}}]{giovannettiQuantumMetrology2006}%
  \BibitemOpen
  \bibfield  {author} {\bibinfo {author} {\bibfnamefont {V.}~\bibnamefont {Giovannetti}}, \bibinfo {author} {\bibfnamefont {S.}~\bibnamefont {Lloyd}},\ and\ \bibinfo {author} {\bibfnamefont {L.}~\bibnamefont {Maccone}},\ }\href {https://doi.org/10.1103/PhysRevLett.96.010401} {\bibfield  {journal} {\bibinfo  {journal} {Physical Review Letters}\ }\textbf {\bibinfo {volume} {96}},\ \bibinfo {pages} {010401} (\bibinfo {year} {2006})}\BibitemShut {NoStop}%
\bibitem [{\citenamefont {Chiribella}\ \emph {et~al.}(2021)\citenamefont {Chiribella}, \citenamefont {Yang},\ and\ \citenamefont {Renner}}]{chiribellaFundamentalEnergyRequirement2021}%
  \BibitemOpen
  \bibfield  {author} {\bibinfo {author} {\bibfnamefont {G.}~\bibnamefont {Chiribella}}, \bibinfo {author} {\bibfnamefont {Y.}~\bibnamefont {Yang}},\ and\ \bibinfo {author} {\bibfnamefont {R.}~\bibnamefont {Renner}},\ }\href {https://doi.org/10.1103/PhysRevX.11.021014} {\bibfield  {journal} {\bibinfo  {journal} {Physical Review X}\ }\textbf {\bibinfo {volume} {11}},\ \bibinfo {pages} {021014} (\bibinfo {year} {2021})}\BibitemShut {NoStop}%
\bibitem [{\citenamefont {Meier}\ and\ \citenamefont {Yamasaki}(2025)}]{meierEnergyConsumptionAdvantageQuantum2025}%
  \BibitemOpen
  \bibfield  {author} {\bibinfo {author} {\bibfnamefont {F.}~\bibnamefont {Meier}}\ and\ \bibinfo {author} {\bibfnamefont {H.}~\bibnamefont {Yamasaki}},\ }\href {https://doi.org/10.1103/PRXEnergy.4.023008} {\bibfield  {journal} {\bibinfo  {journal} {PRX Energy}\ }\textbf {\bibinfo {volume} {4}},\ \bibinfo {pages} {023008} (\bibinfo {year} {2025})}\BibitemShut {NoStop}%
\bibitem [{\citenamefont {Giovannetti}\ \emph {et~al.}(2003)\citenamefont {Giovannetti}, \citenamefont {Lloyd}, \citenamefont {Maccone},\ and\ \citenamefont {Shor}}]{giovannettiBroadbandChannelCapacities2003}%
  \BibitemOpen
  \bibfield  {author} {\bibinfo {author} {\bibfnamefont {V.}~\bibnamefont {Giovannetti}}, \bibinfo {author} {\bibfnamefont {S.}~\bibnamefont {Lloyd}}, \bibinfo {author} {\bibfnamefont {L.}~\bibnamefont {Maccone}},\ and\ \bibinfo {author} {\bibfnamefont {P.~W.}\ \bibnamefont {Shor}},\ }\href {https://doi.org/10.1103/PhysRevA.68.062323} {\bibfield  {journal} {\bibinfo  {journal} {Physical Review A}\ }\textbf {\bibinfo {volume} {68}},\ \bibinfo {pages} {062323} (\bibinfo {year} {2003})},\ \Eprint {https://arxiv.org/abs/quant-ph/0307098} {arXiv:quant-ph/0307098} \BibitemShut {NoStop}%
\bibitem [{\citenamefont {Guha}\ \emph {et~al.}(2008)\citenamefont {Guha}, \citenamefont {Shapiro},\ and\ \citenamefont {Erkmen}}]{guhaCapacityBosonicWiretap2008}%
  \BibitemOpen
  \bibfield  {author} {\bibinfo {author} {\bibfnamefont {S.}~\bibnamefont {Guha}}, \bibinfo {author} {\bibfnamefont {J.~H.}\ \bibnamefont {Shapiro}},\ and\ \bibinfo {author} {\bibfnamefont {B.~I.}\ \bibnamefont {Erkmen}},\ }\href {https://doi.org/10.48550/arXiv.0801.0841} {\bibinfo {title} {Capacity of the {{Bosonic Wiretap Channel}} and the {{Entropy Photon-Number Inequality}}}} (\bibinfo {year} {2008}),\ \Eprint {https://arxiv.org/abs/0801.0841} {arXiv:0801.0841 [quant-ph]} \BibitemShut {NoStop}%
\bibitem [{\citenamefont {Wilde}\ and\ \citenamefont {Qi}(2018)}]{wildeEnergyconstrainedPrivateQuantum2018}%
  \BibitemOpen
  \bibfield  {author} {\bibinfo {author} {\bibfnamefont {M.~M.}\ \bibnamefont {Wilde}}\ and\ \bibinfo {author} {\bibfnamefont {H.}~\bibnamefont {Qi}},\ }\href {https://doi.org/10.1109/TIT.2018.2854766} {\bibfield  {journal} {\bibinfo  {journal} {IEEE Transactions on Information Theory}\ }\textbf {\bibinfo {volume} {64}},\ \bibinfo {pages} {7802} (\bibinfo {year} {2018})},\ \Eprint {https://arxiv.org/abs/1609.01997} {arXiv:1609.01997 [quant-ph]} \BibitemShut {NoStop}%
\bibitem [{\citenamefont {Buzek}\ \emph {et~al.}(1999)\citenamefont {Buzek}, \citenamefont {Derka},\ and\ \citenamefont {Massar}}]{buzekOptimalQuantumClocks1999}%
  \BibitemOpen
  \bibfield  {author} {\bibinfo {author} {\bibfnamefont {V.}~\bibnamefont {Buzek}}, \bibinfo {author} {\bibfnamefont {R.}~\bibnamefont {Derka}},\ and\ \bibinfo {author} {\bibfnamefont {S.}~\bibnamefont {Massar}},\ }\href {https://doi.org/10.1103/PhysRevLett.82.2207} {\bibfield  {journal} {\bibinfo  {journal} {Physical Review Letters}\ }\textbf {\bibinfo {volume} {82}},\ \bibinfo {pages} {2207} (\bibinfo {year} {1999})},\ \Eprint {https://arxiv.org/abs/quant-ph/9808042} {arXiv:quant-ph/9808042} \BibitemShut {NoStop}%
\bibitem [{\citenamefont {Faist}\ \emph {et~al.}(2023)\citenamefont {Faist}, \citenamefont {Woods}, \citenamefont {Albert}, \citenamefont {Renes}, \citenamefont {Eisert},\ and\ \citenamefont {Preskill}}]{faistTimeenergyUncertaintyRelation2023}%
  \BibitemOpen
  \bibfield  {author} {\bibinfo {author} {\bibfnamefont {P.}~\bibnamefont {Faist}}, \bibinfo {author} {\bibfnamefont {M.~P.}\ \bibnamefont {Woods}}, \bibinfo {author} {\bibfnamefont {V.~V.}\ \bibnamefont {Albert}}, \bibinfo {author} {\bibfnamefont {J.~M.}\ \bibnamefont {Renes}}, \bibinfo {author} {\bibfnamefont {J.}~\bibnamefont {Eisert}},\ and\ \bibinfo {author} {\bibfnamefont {J.}~\bibnamefont {Preskill}},\ }\href {https://doi.org/10.1103/PRXQuantum.4.040336} {\bibfield  {journal} {\bibinfo  {journal} {PRX Quantum}\ }\textbf {\bibinfo {volume} {4}},\ \bibinfo {pages} {040336} (\bibinfo {year} {2023})},\ \Eprint {https://arxiv.org/abs/2207.13707} {arXiv:2207.13707 [quant-ph]} \BibitemShut {NoStop}%
\bibitem [{\citenamefont {{Demkowicz-Dobrzanski}}\ \emph {et~al.}(2015)\citenamefont {{Demkowicz-Dobrzanski}}, \citenamefont {Jarzyna},\ and\ \citenamefont {Kolodynski}}]{demkowicz-dobrzanskiQuantumLimitsOptical2015}%
  \BibitemOpen
  \bibfield  {author} {\bibinfo {author} {\bibfnamefont {R.}~\bibnamefont {{Demkowicz-Dobrzanski}}}, \bibinfo {author} {\bibfnamefont {M.}~\bibnamefont {Jarzyna}},\ and\ \bibinfo {author} {\bibfnamefont {J.}~\bibnamefont {Kolodynski}}\ }(\bibinfo {year} {2015})\ pp.\ \bibinfo {pages} {345--435},\ \Eprint {https://arxiv.org/abs/1405.7703} {arXiv:1405.7703 [quant-ph]} \BibitemShut {NoStop}%
\bibitem [{\citenamefont {Knott}\ \emph {et~al.}(2016)\citenamefont {Knott}, \citenamefont {Proctor}, \citenamefont {Hayes}, \citenamefont {Cooling},\ and\ \citenamefont {Dunningham}}]{knottPracticalQuantumMetrology2016}%
  \BibitemOpen
  \bibfield  {author} {\bibinfo {author} {\bibfnamefont {P.~A.}\ \bibnamefont {Knott}}, \bibinfo {author} {\bibfnamefont {T.~J.}\ \bibnamefont {Proctor}}, \bibinfo {author} {\bibfnamefont {A.~J.}\ \bibnamefont {Hayes}}, \bibinfo {author} {\bibfnamefont {J.~P.}\ \bibnamefont {Cooling}},\ and\ \bibinfo {author} {\bibfnamefont {J.~A.}\ \bibnamefont {Dunningham}},\ }\href {https://doi.org/10.1103/PhysRevA.93.033859} {\bibfield  {journal} {\bibinfo  {journal} {Physical Review A}\ }\textbf {\bibinfo {volume} {93}},\ \bibinfo {pages} {033859} (\bibinfo {year} {2016})}\BibitemShut {NoStop}%
\bibitem [{\citenamefont {Wolfgramm}\ \emph {et~al.}(2013)\citenamefont {Wolfgramm}, \citenamefont {Vitelli}, \citenamefont {Beduini}, \citenamefont {Godbout},\ and\ \citenamefont {Mitchell}}]{wolfgrammEntanglementenhancedProbingDelicate2013}%
  \BibitemOpen
  \bibfield  {author} {\bibinfo {author} {\bibfnamefont {F.}~\bibnamefont {Wolfgramm}}, \bibinfo {author} {\bibfnamefont {C.}~\bibnamefont {Vitelli}}, \bibinfo {author} {\bibfnamefont {F.~A.}\ \bibnamefont {Beduini}}, \bibinfo {author} {\bibfnamefont {N.}~\bibnamefont {Godbout}},\ and\ \bibinfo {author} {\bibfnamefont {M.~W.}\ \bibnamefont {Mitchell}},\ }\href {https://doi.org/10.1038/nphoton.2012.300} {\bibfield  {journal} {\bibinfo  {journal} {Nature Photonics}\ }\textbf {\bibinfo {volume} {7}},\ \bibinfo {pages} {28} (\bibinfo {year} {2013})},\ \Eprint {https://arxiv.org/abs/1212.4729} {arXiv:1212.4729 [quant-ph]} \BibitemShut {NoStop}%
\bibitem [{\citenamefont {{Liuzzo-Scorpo}}\ \emph {et~al.}(2018)\citenamefont {{Liuzzo-Scorpo}}, \citenamefont {Correa}, \citenamefont {Pollock}, \citenamefont {G{\'o}recka}, \citenamefont {Modi},\ and\ \citenamefont {Adesso}}]{liuzzo-scorpoEnergyefficientQuantumFrequency2018}%
  \BibitemOpen
  \bibfield  {author} {\bibinfo {author} {\bibfnamefont {P.}~\bibnamefont {{Liuzzo-Scorpo}}}, \bibinfo {author} {\bibfnamefont {L.~A.}\ \bibnamefont {Correa}}, \bibinfo {author} {\bibfnamefont {F.~A.}\ \bibnamefont {Pollock}}, \bibinfo {author} {\bibfnamefont {A.}~\bibnamefont {G{\'o}recka}}, \bibinfo {author} {\bibfnamefont {K.}~\bibnamefont {Modi}},\ and\ \bibinfo {author} {\bibfnamefont {G.}~\bibnamefont {Adesso}},\ }\href {https://doi.org/10.1088/1367-2630/aac5b6} {\bibfield  {journal} {\bibinfo  {journal} {New Journal of Physics}\ }\textbf {\bibinfo {volume} {20}},\ \bibinfo {pages} {063009} (\bibinfo {year} {2018})},\ \Eprint {https://arxiv.org/abs/1712.08143} {arXiv:1712.08143 [quant-ph]} \BibitemShut {NoStop}%
\bibitem [{\citenamefont {G{\'o}recki}\ \emph {et~al.}(2025)\citenamefont {G{\'o}recki}, \citenamefont {Albarelli}, \citenamefont {Felicetti}, \citenamefont {Candia},\ and\ \citenamefont {Maccone}}]{goreckiInterplayTimeEnergy2025}%
  \BibitemOpen
  \bibfield  {author} {\bibinfo {author} {\bibfnamefont {W.}~\bibnamefont {G{\'o}recki}}, \bibinfo {author} {\bibfnamefont {F.}~\bibnamefont {Albarelli}}, \bibinfo {author} {\bibfnamefont {S.}~\bibnamefont {Felicetti}}, \bibinfo {author} {\bibfnamefont {R.~D.}\ \bibnamefont {Candia}},\ and\ \bibinfo {author} {\bibfnamefont {L.}~\bibnamefont {Maccone}},\ }\href {https://doi.org/10.48550/arXiv.2409.18791} {\bibinfo {title} {Interplay between time and energy in bosonic noisy quantum metrology}} (\bibinfo {year} {2025}),\ \Eprint {https://arxiv.org/abs/2409.18791} {arXiv:2409.18791 [quant-ph]} \BibitemShut {NoStop}%
\bibitem [{\citenamefont {Kessler}\ \emph {et~al.}(2014)\citenamefont {Kessler}, \citenamefont {Lovchinsky}, \citenamefont {Sushkov},\ and\ \citenamefont {Lukin}}]{kesslerQuantumErrorCorrection2014}%
  \BibitemOpen
  \bibfield  {author} {\bibinfo {author} {\bibfnamefont {E.~M.}\ \bibnamefont {Kessler}}, \bibinfo {author} {\bibfnamefont {I.}~\bibnamefont {Lovchinsky}}, \bibinfo {author} {\bibfnamefont {A.~O.}\ \bibnamefont {Sushkov}},\ and\ \bibinfo {author} {\bibfnamefont {M.~D.}\ \bibnamefont {Lukin}},\ }\href {https://doi.org/10.1103/PhysRevLett.112.150802} {\bibfield  {journal} {\bibinfo  {journal} {Physical Review Letters}\ }\textbf {\bibinfo {volume} {112}},\ \bibinfo {pages} {150802} (\bibinfo {year} {2014})}\BibitemShut {NoStop}%
\bibitem [{\citenamefont {Arrad}\ \emph {et~al.}(2014)\citenamefont {Arrad}, \citenamefont {Vinkler}, \citenamefont {Aharonov},\ and\ \citenamefont {Retzker}}]{arradIncreasingSensingResolution2014}%
  \BibitemOpen
  \bibfield  {author} {\bibinfo {author} {\bibfnamefont {G.}~\bibnamefont {Arrad}}, \bibinfo {author} {\bibfnamefont {Y.}~\bibnamefont {Vinkler}}, \bibinfo {author} {\bibfnamefont {D.}~\bibnamefont {Aharonov}},\ and\ \bibinfo {author} {\bibfnamefont {A.}~\bibnamefont {Retzker}},\ }\href {https://doi.org/10.1103/PhysRevLett.112.150801} {\bibfield  {journal} {\bibinfo  {journal} {Physical Review Letters}\ }\textbf {\bibinfo {volume} {112}},\ \bibinfo {pages} {150801} (\bibinfo {year} {2014})}\BibitemShut {NoStop}%
\bibitem [{\citenamefont {D{\"u}r}\ \emph {et~al.}(2014)\citenamefont {D{\"u}r}, \citenamefont {Skotiniotis}, \citenamefont {Fr{\"o}wis},\ and\ \citenamefont {Kraus}}]{durImprovedQuantumMetrology2014}%
  \BibitemOpen
  \bibfield  {author} {\bibinfo {author} {\bibfnamefont {W.}~\bibnamefont {D{\"u}r}}, \bibinfo {author} {\bibfnamefont {M.}~\bibnamefont {Skotiniotis}}, \bibinfo {author} {\bibfnamefont {F.}~\bibnamefont {Fr{\"o}wis}},\ and\ \bibinfo {author} {\bibfnamefont {B.}~\bibnamefont {Kraus}},\ }\href {https://doi.org/10.1103/PhysRevLett.112.080801} {\bibfield  {journal} {\bibinfo  {journal} {Physical Review Letters}\ }\textbf {\bibinfo {volume} {112}},\ \bibinfo {pages} {080801} (\bibinfo {year} {2014})}\BibitemShut {NoStop}%
\bibitem [{\citenamefont {{Demkowicz-Dobrzanski}}\ \emph {et~al.}(2017)\citenamefont {{Demkowicz-Dobrzanski}}, \citenamefont {Czajkowski},\ and\ \citenamefont {Sekatski}}]{demkowicz-dobrzanskiAdaptiveQuantumMetrology2017}%
  \BibitemOpen
  \bibfield  {author} {\bibinfo {author} {\bibfnamefont {R.}~\bibnamefont {{Demkowicz-Dobrzanski}}}, \bibinfo {author} {\bibfnamefont {J.}~\bibnamefont {Czajkowski}},\ and\ \bibinfo {author} {\bibfnamefont {P.}~\bibnamefont {Sekatski}},\ }\href {https://doi.org/10.1103/PhysRevX.7.041009} {\bibfield  {journal} {\bibinfo  {journal} {Physical Review X}\ }\textbf {\bibinfo {volume} {7}},\ \bibinfo {pages} {041009} (\bibinfo {year} {2017})}\BibitemShut {NoStop}%
\bibitem [{\citenamefont {Zhou}\ \emph {et~al.}(2018)\citenamefont {Zhou}, \citenamefont {Zhang}, \citenamefont {Preskill},\ and\ \citenamefont {Jiang}}]{zhouAchievingHeisenbergLimit2018}%
  \BibitemOpen
  \bibfield  {author} {\bibinfo {author} {\bibfnamefont {S.}~\bibnamefont {Zhou}}, \bibinfo {author} {\bibfnamefont {M.}~\bibnamefont {Zhang}}, \bibinfo {author} {\bibfnamefont {J.}~\bibnamefont {Preskill}},\ and\ \bibinfo {author} {\bibfnamefont {L.}~\bibnamefont {Jiang}},\ }\href {https://doi.org/10.1038/s41467-017-02510-3} {\bibfield  {journal} {\bibinfo  {journal} {Nature Communications}\ }\textbf {\bibinfo {volume} {9}},\ \bibinfo {pages} {78} (\bibinfo {year} {2018})}\BibitemShut {NoStop}%
\bibitem [{\citenamefont {Yuan}(2016)}]{yuanSequentialFeedbackScheme2016}%
  \BibitemOpen
  \bibfield  {author} {\bibinfo {author} {\bibfnamefont {H.}~\bibnamefont {Yuan}},\ }\href {https://doi.org/10.1103/PhysRevLett.117.160801} {\bibfield  {journal} {\bibinfo  {journal} {Physical Review Letters}\ }\textbf {\bibinfo {volume} {117}},\ \bibinfo {pages} {160801} (\bibinfo {year} {2016})}\BibitemShut {NoStop}%
\bibitem [{\citenamefont {Sekatski}\ \emph {et~al.}(2016)\citenamefont {Sekatski}, \citenamefont {Skotiniotis},\ and\ \citenamefont {D{\"u}r}}]{sekatskiDynamicalDecouplingLeads2016}%
  \BibitemOpen
  \bibfield  {author} {\bibinfo {author} {\bibfnamefont {P.}~\bibnamefont {Sekatski}}, \bibinfo {author} {\bibfnamefont {M.}~\bibnamefont {Skotiniotis}},\ and\ \bibinfo {author} {\bibfnamefont {W.}~\bibnamefont {D{\"u}r}},\ }\href {https://doi.org/10.1088/1367-2630/18/7/073034} {\bibfield  {journal} {\bibinfo  {journal} {New Journal of Physics}\ }\textbf {\bibinfo {volume} {18}},\ \bibinfo {pages} {073034} (\bibinfo {year} {2016})}\BibitemShut {NoStop}%
\bibitem [{\citenamefont {Chiribella}\ \emph {et~al.}(2008)\citenamefont {Chiribella}, \citenamefont {D'Ariano},\ and\ \citenamefont {Perinotti}}]{chiribellaQuantumCircuitsArchitecture2008}%
  \BibitemOpen
  \bibfield  {author} {\bibinfo {author} {\bibfnamefont {G.}~\bibnamefont {Chiribella}}, \bibinfo {author} {\bibfnamefont {G.~M.}\ \bibnamefont {D'Ariano}},\ and\ \bibinfo {author} {\bibfnamefont {P.}~\bibnamefont {Perinotti}},\ }\href {https://doi.org/10.1103/PhysRevLett.101.060401} {\bibfield  {journal} {\bibinfo  {journal} {Physical Review Letters}\ }\textbf {\bibinfo {volume} {101}},\ \bibinfo {pages} {060401} (\bibinfo {year} {2008})},\ \Eprint {https://arxiv.org/abs/0712.1325} {arXiv:0712.1325 [quant-ph]} \BibitemShut {NoStop}%
\bibitem [{\citenamefont {Chiribella}\ \emph {et~al.}(2009)\citenamefont {Chiribella}, \citenamefont {D'Ariano},\ and\ \citenamefont {Perinotti}}]{chiribellaTheoreticalFrameworkQuantum2009}%
  \BibitemOpen
  \bibfield  {author} {\bibinfo {author} {\bibfnamefont {G.}~\bibnamefont {Chiribella}}, \bibinfo {author} {\bibfnamefont {G.~M.}\ \bibnamefont {D'Ariano}},\ and\ \bibinfo {author} {\bibfnamefont {P.}~\bibnamefont {Perinotti}},\ }\href {https://doi.org/10.1103/PhysRevA.80.022339} {\bibfield  {journal} {\bibinfo  {journal} {Physical Review A}\ }\textbf {\bibinfo {volume} {80}},\ \bibinfo {pages} {022339} (\bibinfo {year} {2009})}\BibitemShut {NoStop}%
\bibitem [{\citenamefont {Chiribella}(2012)}]{chiribellaOptimalNetworksQuantum2012}%
  \BibitemOpen
  \bibfield  {author} {\bibinfo {author} {\bibfnamefont {G.}~\bibnamefont {Chiribella}},\ }\href {https://doi.org/10.1088/1367-2630/14/12/125008} {\bibfield  {journal} {\bibinfo  {journal} {New Journal of Physics}\ }\textbf {\bibinfo {volume} {14}},\ \bibinfo {pages} {125008} (\bibinfo {year} {2012})},\ \Eprint {https://arxiv.org/abs/1207.6172} {arXiv:1207.6172 [quant-ph]} \BibitemShut {NoStop}%
\bibitem [{\citenamefont {Yang}(2019)}]{yangMemoryEffectsQuantum2019}%
  \BibitemOpen
  \bibfield  {author} {\bibinfo {author} {\bibfnamefont {Y.}~\bibnamefont {Yang}},\ }\href {https://doi.org/10.1103/PhysRevLett.123.110501} {\bibfield  {journal} {\bibinfo  {journal} {Physical Review Letters}\ }\textbf {\bibinfo {volume} {123}},\ \bibinfo {pages} {110501} (\bibinfo {year} {2019})}\BibitemShut {NoStop}%
\bibitem [{\citenamefont {Altherr}\ and\ \citenamefont {Yang}(2021)}]{altherrQuantumMetrologyNonMarkovian2021}%
  \BibitemOpen
  \bibfield  {author} {\bibinfo {author} {\bibfnamefont {A.}~\bibnamefont {Altherr}}\ and\ \bibinfo {author} {\bibfnamefont {Y.}~\bibnamefont {Yang}},\ }\href {https://doi.org/10.1103/PhysRevLett.127.060501} {\bibfield  {journal} {\bibinfo  {journal} {Physical Review Letters}\ }\textbf {\bibinfo {volume} {127}},\ \bibinfo {pages} {060501} (\bibinfo {year} {2021})}\BibitemShut {NoStop}%
\bibitem [{\citenamefont {Liu}\ \emph {et~al.}(2023)\citenamefont {Liu}, \citenamefont {Hu}, \citenamefont {Yuan},\ and\ \citenamefont {Yang}}]{liuOptimalStrategiesQuantum2023}%
  \BibitemOpen
  \bibfield  {author} {\bibinfo {author} {\bibfnamefont {Q.}~\bibnamefont {Liu}}, \bibinfo {author} {\bibfnamefont {Z.}~\bibnamefont {Hu}}, \bibinfo {author} {\bibfnamefont {H.}~\bibnamefont {Yuan}},\ and\ \bibinfo {author} {\bibfnamefont {Y.}~\bibnamefont {Yang}},\ }\href {https://doi.org/10.1103/PhysRevLett.130.070803} {\bibfield  {journal} {\bibinfo  {journal} {Physical Review Letters}\ }\textbf {\bibinfo {volume} {130}},\ \bibinfo {pages} {070803} (\bibinfo {year} {2023})},\ \Eprint {https://arxiv.org/abs/2203.09758} {arXiv:2203.09758 [quant-ph]} \BibitemShut {NoStop}%
\bibitem [{\citenamefont {Liu}\ \emph {et~al.}(2024)\citenamefont {Liu}, \citenamefont {Hu}, \citenamefont {Yuan},\ and\ \citenamefont {Yang}}]{liuFullyOptimizedQuantumMetrology2024}%
  \BibitemOpen
  \bibfield  {author} {\bibinfo {author} {\bibfnamefont {Q.}~\bibnamefont {Liu}}, \bibinfo {author} {\bibfnamefont {Z.}~\bibnamefont {Hu}}, \bibinfo {author} {\bibfnamefont {H.}~\bibnamefont {Yuan}},\ and\ \bibinfo {author} {\bibfnamefont {Y.}~\bibnamefont {Yang}},\ }\href {https://doi.org/10.1002/qute.202400094} {\bibfield  {journal} {\bibinfo  {journal} {Advanced Quantum Technologies}\ }\textbf {\bibinfo {volume} {7}},\ \bibinfo {pages} {2400094} (\bibinfo {year} {2024})},\ \Eprint {https://arxiv.org/abs/2409.07068} {arXiv:2409.07068 [quant-ph]} \BibitemShut {NoStop}%
\bibitem [{\citenamefont {Kurdzialek}\ \emph {et~al.}(2025)\citenamefont {Kurdzialek}, \citenamefont {Albarelli},\ and\ \citenamefont {{Demkowicz-Dobrzanski}}}]{kurdzialekUniversalBoundsQuantum2025}%
  \BibitemOpen
  \bibfield  {author} {\bibinfo {author} {\bibfnamefont {S.}~\bibnamefont {Kurdzialek}}, \bibinfo {author} {\bibfnamefont {F.}~\bibnamefont {Albarelli}},\ and\ \bibinfo {author} {\bibfnamefont {R.}~\bibnamefont {{Demkowicz-Dobrzanski}}},\ }\href {https://doi.org/10.1103/jy3v-wkcb} {\bibfield  {journal} {\bibinfo  {journal} {Physical Review Letters}\ }\textbf {\bibinfo {volume} {135}},\ \bibinfo {pages} {130801} (\bibinfo {year} {2025})},\ \Eprint {https://arxiv.org/abs/2410.01881} {arXiv:2410.01881 [quant-ph]} \BibitemShut {NoStop}%
\bibitem [{\citenamefont {Kitaev}(1995)}]{kitaevQuantumMeasurementsAbelian1995}%
  \BibitemOpen
  \bibfield  {author} {\bibinfo {author} {\bibfnamefont {A.~Y.}\ \bibnamefont {Kitaev}},\ }\href {https://doi.org/10.48550/arXiv.quant-ph/9511026} {\bibinfo {title} {Quantum measurements and the {{Abelian Stabilizer Problem}}}} (\bibinfo {year} {1995}),\ \Eprint {https://arxiv.org/abs/quant-ph/9511026} {arXiv:quant-ph/9511026} \BibitemShut {NoStop}%
\bibitem [{\citenamefont {Nielsen}\ and\ \citenamefont {Chuang}(2012)}]{nielsenQuantumComputationQuantum2012}%
  \BibitemOpen
  \bibfield  {author} {\bibinfo {author} {\bibfnamefont {M.~A.}\ \bibnamefont {Nielsen}}\ and\ \bibinfo {author} {\bibfnamefont {I.~L.}\ \bibnamefont {Chuang}},\ }\href {https://doi.org/10.1017/CBO9780511976667} {\emph {\bibinfo {title} {Quantum {{Computation}} and {{Quantum Information}}: 10th {{Anniversary Edition}}}}},\ \bibinfo {edition} {1st}\ ed.\ (\bibinfo  {publisher} {Cambridge University Press},\ \bibinfo {year} {2012})\BibitemShut {NoStop}%
\bibitem [{\citenamefont {Hardy}(2009)}]{hardyQuantumGravityComputers2009}%
  \BibitemOpen
  \bibfield  {author} {\bibinfo {author} {\bibfnamefont {L.}~\bibnamefont {Hardy}}\ }(\bibinfo {year} {2009})\ pp.\ \bibinfo {pages} {379--401},\ \Eprint {https://arxiv.org/abs/quant-ph/0701019} {arXiv:quant-ph/0701019} \BibitemShut {NoStop}%
\bibitem [{\citenamefont {Chiribella}\ \emph {et~al.}(2013)\citenamefont {Chiribella}, \citenamefont {D'Ariano}, \citenamefont {Perinotti},\ and\ \citenamefont {Valiron}}]{chiribellaQuantumComputationsDefinite2013}%
  \BibitemOpen
  \bibfield  {author} {\bibinfo {author} {\bibfnamefont {G.}~\bibnamefont {Chiribella}}, \bibinfo {author} {\bibfnamefont {G.~M.}\ \bibnamefont {D'Ariano}}, \bibinfo {author} {\bibfnamefont {P.}~\bibnamefont {Perinotti}},\ and\ \bibinfo {author} {\bibfnamefont {B.}~\bibnamefont {Valiron}},\ }\href {https://doi.org/10.1103/PhysRevA.88.022318} {\bibfield  {journal} {\bibinfo  {journal} {Physical Review A}\ }\textbf {\bibinfo {volume} {88}},\ \bibinfo {pages} {022318} (\bibinfo {year} {2013})},\ \Eprint {https://arxiv.org/abs/0912.0195} {arXiv:0912.0195 [quant-ph]} \BibitemShut {NoStop}%
\bibitem [{\citenamefont {van Dam}\ \emph {et~al.}(2007{\natexlab{a}})\citenamefont {van Dam}, \citenamefont {D'Ariano}, \citenamefont {Ekert}, \citenamefont {Macchiavello},\ and\ \citenamefont {Mosca}}]{damOptimalQuantumCircuits2007a}%
  \BibitemOpen
  \bibfield  {author} {\bibinfo {author} {\bibfnamefont {W.}~\bibnamefont {van Dam}}, \bibinfo {author} {\bibfnamefont {G.~M.}\ \bibnamefont {D'Ariano}}, \bibinfo {author} {\bibfnamefont {A.}~\bibnamefont {Ekert}}, \bibinfo {author} {\bibfnamefont {C.}~\bibnamefont {Macchiavello}},\ and\ \bibinfo {author} {\bibfnamefont {M.}~\bibnamefont {Mosca}},\ }\href {https://doi.org/10.1103/PhysRevLett.98.090501} {\bibfield  {journal} {\bibinfo  {journal} {Physical Review Letters}\ }\textbf {\bibinfo {volume} {98}},\ \bibinfo {pages} {090501} (\bibinfo {year} {2007}{\natexlab{a}})},\ \Eprint {https://arxiv.org/abs/quant-ph/0609160} {arXiv:quant-ph/0609160} \BibitemShut {NoStop}%
\bibitem [{\citenamefont {Chiribella}\ \emph {et~al.}(2005)\citenamefont {Chiribella}, \citenamefont {D'Ariano},\ and\ \citenamefont {Sacchi}}]{chiribellaOptimalEstimationGroup2005}%
  \BibitemOpen
  \bibfield  {author} {\bibinfo {author} {\bibfnamefont {G.}~\bibnamefont {Chiribella}}, \bibinfo {author} {\bibfnamefont {G.~M.}\ \bibnamefont {D'Ariano}},\ and\ \bibinfo {author} {\bibfnamefont {M.~F.}\ \bibnamefont {Sacchi}},\ }\href {https://doi.org/10.1103/PhysRevA.72.042338} {\bibfield  {journal} {\bibinfo  {journal} {Physical Review A}\ }\textbf {\bibinfo {volume} {72}},\ \bibinfo {pages} {042338} (\bibinfo {year} {2005})},\ \Eprint {https://arxiv.org/abs/quant-ph/0506267} {arXiv:quant-ph/0506267} \BibitemShut {NoStop}%
\bibitem [{\citenamefont {van Dam}\ \emph {et~al.}(2007{\natexlab{b}})\citenamefont {van Dam}, \citenamefont {D'Ariano}, \citenamefont {Ekert}, \citenamefont {Macchiavello},\ and\ \citenamefont {Mosca}}]{damOptimalPhaseEstimation2007}%
  \BibitemOpen
  \bibfield  {author} {\bibinfo {author} {\bibfnamefont {W.}~\bibnamefont {van Dam}}, \bibinfo {author} {\bibfnamefont {G.~M.}\ \bibnamefont {D'Ariano}}, \bibinfo {author} {\bibfnamefont {A.}~\bibnamefont {Ekert}}, \bibinfo {author} {\bibfnamefont {C.}~\bibnamefont {Macchiavello}},\ and\ \bibinfo {author} {\bibfnamefont {M.}~\bibnamefont {Mosca}},\ }\href {https://doi.org/10.1088/1751-8113/40/28/S07} {\bibfield  {journal} {\bibinfo  {journal} {Journal of Physics A: Mathematical and Theoretical}\ }\textbf {\bibinfo {volume} {40}},\ \bibinfo {pages} {7971} (\bibinfo {year} {2007}{\natexlab{b}})},\ \Eprint {https://arxiv.org/abs/0706.4412} {arXiv:0706.4412 [quant-ph]} \BibitemShut {NoStop}%
\bibitem [{\citenamefont {Wechs}\ \emph {et~al.}(2021)\citenamefont {Wechs}, \citenamefont {Dourdent}, \citenamefont {Abbott},\ and\ \citenamefont {Branciard}}]{wechsQuantumCircuitsClassical2021}%
  \BibitemOpen
  \bibfield  {author} {\bibinfo {author} {\bibfnamefont {J.}~\bibnamefont {Wechs}}, \bibinfo {author} {\bibfnamefont {H.}~\bibnamefont {Dourdent}}, \bibinfo {author} {\bibfnamefont {A.~A.}\ \bibnamefont {Abbott}},\ and\ \bibinfo {author} {\bibfnamefont {C.}~\bibnamefont {Branciard}},\ }\href {https://doi.org/10.1103/PRXQuantum.2.030335} {\bibfield  {journal} {\bibinfo  {journal} {PRX Quantum}\ }\textbf {\bibinfo {volume} {2}},\ \bibinfo {pages} {030335} (\bibinfo {year} {2021})},\ \Eprint {https://arxiv.org/abs/2101.08796} {arXiv:2101.08796 [quant-ph]} \BibitemShut {NoStop}%
\bibitem [{\citenamefont {Mothe}\ \emph {et~al.}(2024)\citenamefont {Mothe}, \citenamefont {Branciard},\ and\ \citenamefont {Abbott}}]{motheReassessingAdvantageIndefinite2024}%
  \BibitemOpen
  \bibfield  {author} {\bibinfo {author} {\bibfnamefont {R.}~\bibnamefont {Mothe}}, \bibinfo {author} {\bibfnamefont {C.}~\bibnamefont {Branciard}},\ and\ \bibinfo {author} {\bibfnamefont {A.~A.}\ \bibnamefont {Abbott}},\ }\href {https://doi.org/10.1103/PhysRevA.109.062435} {\bibfield  {journal} {\bibinfo  {journal} {Physical Review A}\ }\textbf {\bibinfo {volume} {109}},\ \bibinfo {pages} {062435} (\bibinfo {year} {2024})},\ \Eprint {https://arxiv.org/abs/2312.12172} {arXiv:2312.12172 [quant-ph]} \BibitemShut {NoStop}%
\bibitem [{\citenamefont {Humphreys}\ \emph {et~al.}(2013)\citenamefont {Humphreys}, \citenamefont {Barbieri}, \citenamefont {Datta},\ and\ \citenamefont {Walmsley}}]{humphreysQuantumEnhancedMultiple2013}%
  \BibitemOpen
  \bibfield  {author} {\bibinfo {author} {\bibfnamefont {P.~C.}\ \bibnamefont {Humphreys}}, \bibinfo {author} {\bibfnamefont {M.}~\bibnamefont {Barbieri}}, \bibinfo {author} {\bibfnamefont {A.}~\bibnamefont {Datta}},\ and\ \bibinfo {author} {\bibfnamefont {I.~A.}\ \bibnamefont {Walmsley}},\ }\href {https://doi.org/10.1103/PhysRevLett.111.070403} {\bibfield  {journal} {\bibinfo  {journal} {Physical Review Letters}\ }\textbf {\bibinfo {volume} {111}},\ \bibinfo {pages} {070403} (\bibinfo {year} {2013})}\BibitemShut {NoStop}%
\bibitem [{\citenamefont {Szczykulska}\ \emph {et~al.}(2016)\citenamefont {Szczykulska}, \citenamefont {Baumgratz},\ and\ \citenamefont {Datta}}]{szczykulskaMultiparameterQuantumMetrology2016}%
  \BibitemOpen
  \bibfield  {author} {\bibinfo {author} {\bibfnamefont {M.}~\bibnamefont {Szczykulska}}, \bibinfo {author} {\bibfnamefont {T.}~\bibnamefont {Baumgratz}},\ and\ \bibinfo {author} {\bibfnamefont {A.}~\bibnamefont {Datta}},\ }\href {https://doi.org/10.1080/23746149.2016.1230476} {\bibfield  {journal} {\bibinfo  {journal} {Advances in Physics: X}\ }\textbf {\bibinfo {volume} {1}},\ \bibinfo {pages} {621} (\bibinfo {year} {2016})}\BibitemShut {NoStop}%
\bibitem [{\citenamefont {Hayashi}\ and\ \citenamefont {Ouyang}(2024)}]{hayashiFindingOptimalProbe2024}%
  \BibitemOpen
  \bibfield  {author} {\bibinfo {author} {\bibfnamefont {M.}~\bibnamefont {Hayashi}}\ and\ \bibinfo {author} {\bibfnamefont {Y.}~\bibnamefont {Ouyang}},\ }\href {https://doi.org/10.1038/s41534-024-00905-x} {\bibfield  {journal} {\bibinfo  {journal} {npj Quantum Information}\ }\textbf {\bibinfo {volume} {10}},\ \bibinfo {pages} {111} (\bibinfo {year} {2024})},\ \Eprint {https://arxiv.org/abs/2401.05886} {arXiv:2401.05886 [quant-ph]} \BibitemShut {NoStop}%
\bibitem [{\citenamefont {Hayashi}\ and\ \citenamefont {Ouyang}(2023)}]{hayashiTightCramRRao2023}%
  \BibitemOpen
  \bibfield  {author} {\bibinfo {author} {\bibfnamefont {M.}~\bibnamefont {Hayashi}}\ and\ \bibinfo {author} {\bibfnamefont {Y.}~\bibnamefont {Ouyang}},\ }\href {https://doi.org/10.22331/q-2023-08-29-1094} {\bibfield  {journal} {\bibinfo  {journal} {Quantum}\ }\textbf {\bibinfo {volume} {7}},\ \bibinfo {pages} {1094} (\bibinfo {year} {2023})},\ \Eprint {https://arxiv.org/abs/2209.05218} {arXiv:2209.05218 [quant-ph]} \BibitemShut {NoStop}%
\bibitem [{\citenamefont {{Demkowicz-Dobrzanski}}\ \emph {et~al.}(2020)\citenamefont {{Demkowicz-Dobrzanski}}, \citenamefont {Gorecki},\ and\ \citenamefont {Guta}}]{demkowicz-dobrzanskiMultiparameterEstimationQuantum2020}%
  \BibitemOpen
  \bibfield  {author} {\bibinfo {author} {\bibfnamefont {R.}~\bibnamefont {{Demkowicz-Dobrzanski}}}, \bibinfo {author} {\bibfnamefont {W.}~\bibnamefont {Gorecki}},\ and\ \bibinfo {author} {\bibfnamefont {M.}~\bibnamefont {Guta}},\ }\href {https://doi.org/10.1088/1751-8121/ab8ef3} {\bibfield  {journal} {\bibinfo  {journal} {Journal of Physics A: Mathematical and Theoretical}\ }\textbf {\bibinfo {volume} {53}},\ \bibinfo {pages} {363001} (\bibinfo {year} {2020})},\ \Eprint {https://arxiv.org/abs/2001.11742} {arXiv:2001.11742 [quant-ph]} \BibitemShut {NoStop}%
\bibitem [{\citenamefont {Sidhu}\ and\ \citenamefont {Kok}(2020)}]{sidhuGeometricPerspectiveQuantum2020a}%
  \BibitemOpen
  \bibfield  {author} {\bibinfo {author} {\bibfnamefont {J.~S.}\ \bibnamefont {Sidhu}}\ and\ \bibinfo {author} {\bibfnamefont {P.}~\bibnamefont {Kok}},\ }\href {https://doi.org/10.1116/1.5119961} {\bibfield  {journal} {\bibinfo  {journal} {AVS Quantum Science}\ }\textbf {\bibinfo {volume} {2}},\ \bibinfo {pages} {014701} (\bibinfo {year} {2020})},\ \Eprint {https://arxiv.org/abs/1907.06628} {arXiv:1907.06628 [quant-ph]} \BibitemShut {NoStop}%
\bibitem [{\citenamefont {Chin}\ \emph {et~al.}(2012)\citenamefont {Chin}, \citenamefont {Huelga},\ and\ \citenamefont {Plenio}}]{chinQuantumMetrologyNonMarkovian2012}%
  \BibitemOpen
  \bibfield  {author} {\bibinfo {author} {\bibfnamefont {A.~W.}\ \bibnamefont {Chin}}, \bibinfo {author} {\bibfnamefont {S.~F.}\ \bibnamefont {Huelga}},\ and\ \bibinfo {author} {\bibfnamefont {M.~B.}\ \bibnamefont {Plenio}},\ }\href {https://doi.org/10.1103/PhysRevLett.109.233601} {\bibfield  {journal} {\bibinfo  {journal} {Physical Review Letters}\ }\textbf {\bibinfo {volume} {109}},\ \bibinfo {pages} {233601} (\bibinfo {year} {2012})}\BibitemShut {NoStop}%
\bibitem [{\citenamefont {Ahnefeld}\ \emph {et~al.}(2025)\citenamefont {Ahnefeld}, \citenamefont {Theurer},\ and\ \citenamefont {Plenio}}]{ahnefeldCoherenceResourcePhase2025}%
  \BibitemOpen
  \bibfield  {author} {\bibinfo {author} {\bibfnamefont {F.}~\bibnamefont {Ahnefeld}}, \bibinfo {author} {\bibfnamefont {T.}~\bibnamefont {Theurer}},\ and\ \bibinfo {author} {\bibfnamefont {M.~B.}\ \bibnamefont {Plenio}},\ }\href {https://doi.org/10.48550/arXiv.2505.18544} {\bibinfo {title} {Coherence as a resource for phase estimation}} (\bibinfo {year} {2025}),\ \Eprint {https://arxiv.org/abs/2505.18544} {arXiv:2505.18544 [quant-ph]} \BibitemShut {NoStop}%
\bibitem [{\citenamefont {Chitambar}\ and\ \citenamefont {Gour}(2019)}]{chitambarQuantumResourceTheories2019}%
  \BibitemOpen
  \bibfield  {author} {\bibinfo {author} {\bibfnamefont {E.}~\bibnamefont {Chitambar}}\ and\ \bibinfo {author} {\bibfnamefont {G.}~\bibnamefont {Gour}},\ }\href {https://doi.org/10.1103/RevModPhys.91.025001} {\bibfield  {journal} {\bibinfo  {journal} {Reviews of Modern Physics}\ }\textbf {\bibinfo {volume} {91}},\ \bibinfo {pages} {025001} (\bibinfo {year} {2019})}\BibitemShut {NoStop}%
\bibitem [{\citenamefont {Gour}(2024)}]{gourResourcesQuantumWorld2024}%
  \BibitemOpen
  \bibfield  {author} {\bibinfo {author} {\bibfnamefont {G.}~\bibnamefont {Gour}},\ }\href@noop {} {\bibinfo {title} {Resources of the {{Quantum World}}}} (\bibinfo {year} {2024}),\ \Eprint {https://arxiv.org/abs/2402.05474} {arXiv:2402.05474 [math-ph, physics:quant-ph]} \BibitemShut {NoStop}%
\bibitem [{\citenamefont {Aberg}(2014)}]{abergCatalyticCoherence2014}%
  \BibitemOpen
  \bibfield  {author} {\bibinfo {author} {\bibfnamefont {J.}~\bibnamefont {Aberg}},\ }\href {https://doi.org/10.1103/PhysRevLett.113.150402} {\bibfield  {journal} {\bibinfo  {journal} {Physical Review Letters}\ }\textbf {\bibinfo {volume} {113}},\ \bibinfo {pages} {150402} (\bibinfo {year} {2014})},\ \Eprint {https://arxiv.org/abs/1304.1060} {arXiv:1304.1060 [quant-ph]} \BibitemShut {NoStop}%
\bibitem [{\citenamefont {Braunstein}\ and\ \citenamefont {Caves}(1994)}]{braunsteinStatisticalDistanceGeometry1994}%
  \BibitemOpen
  \bibfield  {author} {\bibinfo {author} {\bibfnamefont {S.~L.}\ \bibnamefont {Braunstein}}\ and\ \bibinfo {author} {\bibfnamefont {C.~M.}\ \bibnamefont {Caves}},\ }\href {https://doi.org/10.1103/PhysRevLett.72.3439} {\bibfield  {journal} {\bibinfo  {journal} {Physical Review Letters}\ }\textbf {\bibinfo {volume} {72}},\ \bibinfo {pages} {3439} (\bibinfo {year} {1994})}\BibitemShut {NoStop}%
\bibitem [{\citenamefont {Van~Trees}(2001)}]{vantreesDetectionEstimationModulation2001}%
  \BibitemOpen
  \bibfield  {author} {\bibinfo {author} {\bibfnamefont {H.~L.}\ \bibnamefont {Van~Trees}},\ }\href {https://doi.org/10.1002/0471221082} {\emph {\bibinfo {title} {Detection, {{Estimation}}, and {{Modulation Theory}}, {{Part I}}}}}\ (\bibinfo  {publisher} {John Wiley \& Sons, Inc.},\ \bibinfo {address} {New York, USA},\ \bibinfo {year} {2001})\BibitemShut {NoStop}%
\bibitem [{\citenamefont {Gill}\ and\ \citenamefont {Levit}(1995)}]{gillApplicationsVanTrees1995}%
  \BibitemOpen
  \bibfield  {author} {\bibinfo {author} {\bibfnamefont {R.~D.}\ \bibnamefont {Gill}}\ and\ \bibinfo {author} {\bibfnamefont {B.~Y.}\ \bibnamefont {Levit}},\ }\href {https://doi.org/10.2307/3318681} {\bibfield  {journal} {\bibinfo  {journal} {Bernoulli}\ }\textbf {\bibinfo {volume} {1}},\ \bibinfo {pages} {59} (\bibinfo {year} {1995})},\ \Eprint {https://arxiv.org/abs/3318681} {3318681} \BibitemShut {NoStop}%
\bibitem [{\citenamefont {Elisabeth}\ and\ \citenamefont {Gilles}(2024)}]{elisabethVanTreesInequality2024}%
  \BibitemOpen
  \bibfield  {author} {\bibinfo {author} {\bibfnamefont {G.}~\bibnamefont {Elisabeth}}\ and\ \bibinfo {author} {\bibfnamefont {S.}~\bibnamefont {Gilles}},\ }\href {https://doi.org/10.48550/arXiv.2402.06431} {\bibinfo {title} {The van {{Trees}} inequality in the spirit of {{Hajek}} and {{Le Cam}}}} (\bibinfo {year} {2024}),\ \Eprint {https://arxiv.org/abs/2402.06431} {arXiv:2402.06431 [math]} \BibitemShut {NoStop}%
\bibitem [{\citenamefont {Stinespring}(1955)}]{stinespringPositiveFunctionsAlgebras1955}%
  \BibitemOpen
  \bibfield  {author} {\bibinfo {author} {\bibfnamefont {W.~F.}\ \bibnamefont {Stinespring}},\ }\href {https://doi.org/10.2307/2032342} {\bibfield  {journal} {\bibinfo  {journal} {Proceedings of the American Mathematical Society}\ }\textbf {\bibinfo {volume} {6}},\ \bibinfo {pages} {211} (\bibinfo {year} {1955})},\ \Eprint {https://arxiv.org/abs/2032342} {2032342} \BibitemShut {NoStop}%
\bibitem [{\citenamefont {Purves}\ and\ \citenamefont {Short}(2021)}]{purvesQuantumTheoryCannot2021}%
  \BibitemOpen
  \bibfield  {author} {\bibinfo {author} {\bibfnamefont {T.}~\bibnamefont {Purves}}\ and\ \bibinfo {author} {\bibfnamefont {A.~J.}\ \bibnamefont {Short}},\ }\href {https://doi.org/10.1103/PhysRevLett.127.110402} {\bibfield  {journal} {\bibinfo  {journal} {Physical Review Letters}\ }\textbf {\bibinfo {volume} {127}},\ \bibinfo {pages} {110402} (\bibinfo {year} {2021})},\ \Eprint {https://arxiv.org/abs/2101.09107} {arXiv:2101.09107 [quant-ph]} \BibitemShut {NoStop}%
\bibitem [{\citenamefont {Watson}(1922)}]{watson1922treatise}%
  \BibitemOpen
  \bibfield  {author} {\bibinfo {author} {\bibfnamefont {G.~N.}\ \bibnamefont {Watson}},\ }\href@noop {} {\emph {\bibinfo {title} {A Treatise on the Theory of {{Bessel}} Functions}}},\ Vol.~\bibinfo {volume} {3}\ (\bibinfo  {publisher} {The University Press},\ \bibinfo {year} {1922})\BibitemShut {NoStop}%
\bibitem [{\citenamefont {Ifantis}\ and\ \citenamefont {Siafarikas}(1992)}]{ifantisDifferentialInequalityPositive1992}%
  \BibitemOpen
  \bibfield  {author} {\bibinfo {author} {\bibfnamefont {E.}~\bibnamefont {Ifantis}}\ and\ \bibinfo {author} {\bibfnamefont {P.}~\bibnamefont {Siafarikas}},\ }\href {https://doi.org/10.1016/0377-0427(92)90055-3} {\bibfield  {journal} {\bibinfo  {journal} {Journal of Computational and Applied Mathematics}\ }\textbf {\bibinfo {volume} {44}},\ \bibinfo {pages} {115} (\bibinfo {year} {1992})}\BibitemShut {NoStop}%
\bibitem [{\citenamefont {Qu}\ and\ \citenamefont {Wong}(1999)}]{quBestPossibleUpper1999}%
  \BibitemOpen
  \bibfield  {author} {\bibinfo {author} {\bibfnamefont {C.}~\bibnamefont {Qu}}\ and\ \bibinfo {author} {\bibfnamefont {R.}~\bibnamefont {Wong}},\ }\href {https://doi.org/10.1090/S0002-9947-99-02165-0} {\bibfield  {journal} {\bibinfo  {journal} {Transactions of the American Mathematical Society}\ }\textbf {\bibinfo {volume} {351}},\ \bibinfo {pages} {2833} (\bibinfo {year} {1999})}\BibitemShut {NoStop}%
\end{thebibliography}%

\onecolumngrid

\appendix

\section{Introduction to Quantum Combs} \label{sec:quantum_comb}
Here, we provide a brief introduction to the quantum comb framework. 

\subsection{Quantum Combs with Definite Causal Orders}
The quantum comb framework is initially to represent multi-step quantum processes with a fixed causal structure~\cite{chiribellaQuantumCircuitsArchitecture2008,chiribellaTheoreticalFrameworkQuantum2009}. It generalizes the concept of a quantum channel to a circuit with multiple open slots (inputs and outputs), which can be plugged by other quantum operations. Graphically, a comb is often depicted as a circuit with a series of ``teeth''. A general $N$-step comb $C$ can be represented as shown in Fig.~\ref{fig:comb}, where each tooth corresponds to a step. The spaces between the teeth are the slots to plug other operations. The set of such combs with the given input and output spaces is denoted by $\Comb[(\HH_0,\HH_1),\dots,(\HH_{2N-2},\HH_{2N-1})]$.

\begin{figure}[h]
    \centering
    \includegraphics[width=0.4\linewidth]{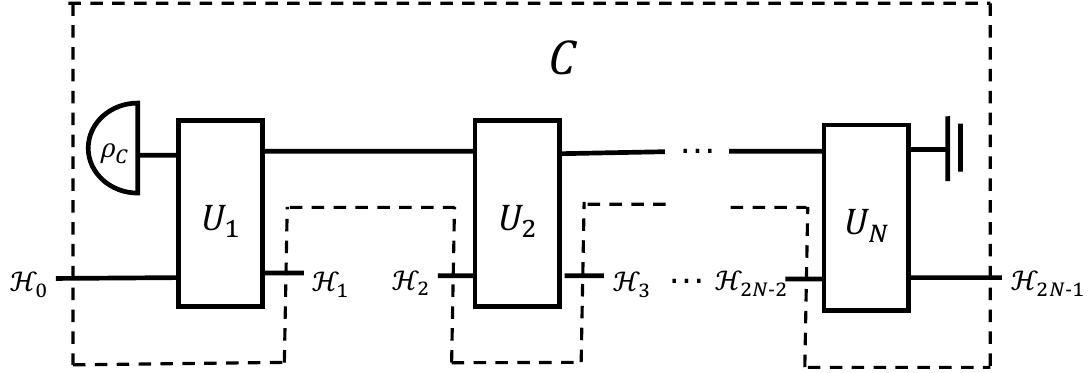}
    \caption{ A schematic representation of an $N$-step quantum comb $C$.}
    \label{fig:comb}
\end{figure}

Mathematically, the quantum comb $C$ is the Choi operator of the process, which is a positive semidefinite operator on the tensor product of all input and output Hilbert spaces. It is characterized by the recursive relation 
\begin{equation}\label{eq:causality_constraint}
    \begin{cases}
        \Tr_{2n-1}[C^{(n)}] = C^{(n-1)}\otimes I_{2n-2}, & n=2,\dots,N, \\
        \Tr_{1}[C^{(1)}] = I_{0},
    \end{cases}
\end{equation}
with $C^{(N)} = C$~\cite{chiribellaQuantumCircuitsArchitecture2008,chiribellaTheoreticalFrameworkQuantum2009}. Here, $\Tr_n[\cdot]$ and $I_n$ denote the partial trace and the identity operator on $\HH_n$, respectively. These causality constraints guarantee that it represents a quantum process with a definite causal order, as tracing out the output space of a Choi operator should yield an identity on the input space. This ensures the process is independent of the future at each step. 

By dropping the causality constraints in \eqref{eq:causality_constraint}, the quantum comb framework can also be extended to the representation of indefinite causal order processes. This will be discussed later in the next subsection. 

Operationally, tracing out the output of a step corresponds to discarding that output. The recursive structure of the comb ensures that by tracing out the outputs of the last $N-n$ steps, one obtains the comb $C^{(n)}$ describing the sub-process of the first $n$ steps. This allows any sub-comb within the process to be determined by sequentially discarding future outputs.

Furthermore, the causality condition implies the neutralization property of combs, that is, for an $N$-step comb, the state obtained by discarding the outputs of the steps $n+1,\dots,N$ from the joint output state of the entire process is identical to the output state of the first $n$ steps. Formally, for an $N$-step comb $C$ with a joint input $\rho$ on $\bigotimes_{m=1}^N \HH_{2m-2}$, if the output spaces $\{\HH_{2m-1}\}_{m=n+1}^N$ are traced out in the joint output $C*\rho$, the obtained state on the partial output system $\bigotimes_{m=1}^n \HH_{2m-1}$ equals
\begin{equation}
\begin{aligned}
    \Tr_{2n+1,2n+3,\dots,2N-1}[C*\rho] =& \Tr_{2n+1,2n+3,\dots,2N-1}[\Tr_{0,2\dots,2N-2}[C(\rho^T \otimes I_{1,3,\dots,2N-1})]]
    \\ =&  \Tr_{0,2\dots,2N-2}[\Tr_{2n+1,2n+3,\dots,2N-1}[C](\rho^T \otimes I_{1,3,\dots,2n-1})].
\end{aligned}
\end{equation}
By the causality constraints in \eqref{eq:causality_constraint}, we have $\Tr_{2n+1,2n+3,\dots,2N-1}[C] = C^{(n)} \otimes I_{2n,2n+2,\dots,2N-2}$. Substituting this into the above expression yields:
\begin{equation}
    \begin{aligned}
        \Tr_{2n+1,2n+3,\dots,2N-1}[C*\rho] =& \Tr_{0,2\dots,2N-2}[(C^{(n)} \otimes I_{2n,2n+2,\dots,2N-2})(\rho^T \otimes I_{1,3,\dots,2n-1})]
        \\ =& \Tr_{0,2\dots,2n-2}[C^{(n)}(\Tr_{2n,2n+2,\dots,2N-2}[\rho]^T \otimes I_{1,3,\dots,2n-1})]
        \\ =& C^{(n)} * \Tr_{2n,2n+2,\dots,2N-2}[\rho].
    \end{aligned}
\end{equation}
Since $\Tr_{2n,2n+2,\dots,2N-2}[\rho]$ constitutes the input state to the first $n$ steps on $\bigotimes_{m=1}^n \HH_{2m-2}$, this final result is precisely the output state of the first $n$ steps.

The Stinespring dilation theorem~\cite{stinespringPositiveFunctionsAlgebras1955}, stating that any quantum channel can be implemented by a unitary on a larger system with an ancilla, is generalized to quantum combs for the circuit implementation~\cite{chiribellaTheoreticalFrameworkQuantum2009}. To be specific, an $N$-step quantum comb can be implemented by a sequence of unitary operations $\{U_n\}_{n=1}^N$, with a control ancilla initialized in some state $\rho_C$ representing the quantum memory between the steps, where $U_n$ is the operation at the $n$-th step, as depicted in Fig.~\ref{fig:comb}. 

When operations are plugged into the slots of a quantum comb $C_1$, the entire set of plugged operations can itself be viewed as a multi-step process described by another quantum comb $C_2$. Given two quantum combs $C_1$ and $C_2$ defined on the Hilbert spaces $\bigotimes_{a \in A} \HH_a$ and $\bigotimes_{b \in B} \HH_b$ respectively, the overall process resulting from their sequential composition is then also represented by a new quantum comb, which is formally determined by their \emph{link product} defined as:
\begin{equation}
C_1 * C_2 := \Tr_{A \cap B} \left[ \left( C_1^{T_{A \cap B}} \otimes I_{B \backslash A} \right) \left( I_{A \backslash B} \otimes C_2 \right) \right],
\end{equation}
where $C^{T_{A \cap B}}$ denotes the partial transpose over the subsystems in the intersection $A \cap B$~\cite{chiribellaQuantumCircuitsArchitecture2008,chiribellaTheoreticalFrameworkQuantum2009}. The result is a valid quantum comb acting on the Hilbert space $\bigotimes_{c \in (A \cup B) \backslash (A \cap B)} \HH_c$.

\subsection{Quantum Combs Extended to Indefinite Causal Order Processes}\label{sec:ICO_comb}
The quantum comb, as introduced before, describes the processes with a definite causal order. By relaxing the causality constraints [cf. Eq.~\eqref{eq:causality_constraint}], it can be extended to describe the processes where the causal order between the operations is not definite, known as indefinite causal order (ICO) processes~\cite{hardyQuantumGravityComputers2009,chiribellaQuantumComputationsDefinite2013,wechsQuantumCircuitsClassical2021}.

An ICO process, represented by its Choi operator $C$, can execute a sequence of operations $\{C_n\}_{n=1}^{N-1}$ (where $C_n$ is the Choi operator of an operation that maps from $\HH_{2n-1}$ to $\HH_{2n}$) in an indefinite causal orders. In the most general formulation, relaxing the causality constraints in \eqref{eq:causality_constraint}, we only require that for any valid set of input operations $\{C_n\}$ and an initial input state $\rho$ on $\HH_0$, the final output $C * \left(\rho \otimes (\bigotimes_{n} C_n)\right)$ on $\HH_{2N-1}$ constitutes a valid quantum state. It is crucial to note, however, that not every Choi operator $C$ satisfying this condition corresponds to a physically implementable process~\cite{wechsQuantumCircuitsClassical2021,purvesQuantumTheoryCannot2021}.

Therefore, this work focuses on \emph{causal superposition} processes, where a set of operations is executed in a quantum superposition of multiple causal orders~\cite{chiribellaQuantumComputationsDefinite2013,hardyQuantumGravityComputers2009}.

Such a process can be physically realized by quantum control with an order control system $\HH_{\text{oc}}$ of dimension $(N-1)!$~\cite{wechsQuantumCircuitsClassical2021}. Let $\tau \in S_{N-1}$ be a permutation specifying an order of applying the operations $\{C_n\}_{n=1}^{N-1}$, where $S_{N-1}$ denotes the symmetric group of degree $N-1$. When the control system is in the state $\ket{\tau}$, the operations are applied in the sequence $C_{\tau(1)} \to C_{\tau(2)} \to \cdots \to C_{\tau(N-1)}$. A causal superposition is then achieved by preparing $\HH_{\text{oc}}$ in a superposition of these permutation states. 

In this way, if we discard the system $\HH_{\text{oc}}$, the causal superposition will become a classical probabilistic mixture of the causal orders. Thus, for the Choi operator $C$ of a causal superposition process on $(\bigotimes_{n=0}^{2N-1} \HH_{n}) \otimes \HH_{\text{oc}}$, it can be characterized by tracing out $\HH_{\text{oc}}$ yields a classical probabilistic mixture of quantum combs in different causal orders:
\begin{equation}
    \Tr_{\HH_{\text{oc}}}[C] = \sum_{\tau \in S_{N-1}} p_\tau C^\tau,
\end{equation}
with the causality quantum comb 
\begin{equation}
    C^\tau \in \Comb[(\HH_0, \HH_{2\tau(1)-1}),\dots,(\HH_{2\tau(n)-1},\HH_{2\tau(n+1)}),\dots,(\HH_{2\tau(N-1)},\HH_{2N-1})]
\end{equation} 
corresponding to the causal order $\tau$ and the proportion $p_\tau \geq 0$ such that $\sum_{\tau} p_\tau = 1$~\cite{wechsQuantumCircuitsClassical2021,liuOptimalStrategiesQuantum2023}.

\section{Proof of Theorem~\ref{thm:energy_constrained_comb}} \label{sec:proof_energy_constrained_comb}
First, by the definition, it is straightforward that $C \in \Comb^{\leq E}_{\text{gl/loc}}$ if and only if $E_{\text{gl/loc}}(C) \leq E$. Therefore, it suffices to determine the value of $E_{\text{gl/loc}}(C)$.

For the global battery case, to support the energy changes during the process, the initial battery energy must be no less than the sum of the energy changes in any first $n_0$ steps with $n_0 = 1,\dots,N$. 
By the energy change in the step $n$, which can be determined as [cf.~Eq.~\eqref{eq:energy_change} in the main text]
\begin{equation} \label{eq:energy_change_appendix}
     \Tr[(H_{2n-1} \otimes I) (C*\rho)] - \Tr[(H_{2n-2} \otimes I) \rho]  =
     \Tr\left[ \left( \Tr_{1,3,\dots,2N-1}[(H_{2n-1} \otimes I) C] - (H_{2n-2})^T \otimes I \right) \rho^T \right],
\end{equation} 
we have
\begin{equation}\label{eq:global_energy_consumption}
    E_{\text{gl}}(C) = \max_{1\leq n_0 \leq N}\max_{\rho} \sum_{n=1}^{n_0} \Tr\left[ \left( \Tr_{1,3,\dots,2N-1}[(H_{2n-1} \otimes I) C] - (H_{2n-2})^T \otimes I \right) \rho^T \right].
\end{equation}
Since $H_{n_1,\dots,n_k} := \sum_{j=1}^k H_{n_j} \otimes I$, this is equivalent to $E_\text{gl}(C)$ being the minimal value satisfying the matrix inequalities:
\begin{equation}
    E_{\text{gl}}(C) I \succeq  \Tr_{1,3,\dots,2N-1}[(H_{1,3,\dots,2n-1} \otimes I) C] - (H_{0,2,\dots,2n-2})^T \otimes I,
    \quad n=1,\dots,N.
\end{equation}
This precisely yields the expression in terms of the largest eigenvalue stated in the theorem.

For the local battery case, the energy $E_n$ in the $n$-th battery must be at least the maximum possible energy change in the $n$-th step. By \eqref{eq:energy_change_appendix}, this implies
\begin{equation}
E_n(C) = \max\left\{\max_{\rho} \Tr\left[ \left( \Tr_{1,3,\dots,2N-1}[(H_{2n-1} \otimes I) C] - (H_{2n-2})^T \otimes I \right) \rho^T \right],0\right\}.
\end{equation}
Since $\Tr[\rho^T] = 1$, $E_n(C)$ is the minimal nonnegative value satisfying the operator inequality
\begin{equation}
    E_n(C) I \succeq \Tr_{1,3,\dots,2N-1}[(H_{2n-1} \otimes I) C] - (H_{2n-2})^T \otimes I,
\end{equation}
which is equivalent to
\begin{equation}\label{eq:local_energy_consumption}
E_n(C) = \max\left\{\lambda_{\max} \left( \Tr_{1,3,\dots,2N-1}[(H_{2n-1} \otimes I) C] - (H_{2n-2})^T \otimes I \right),0\right\}.
\end{equation}
Therefore, $E_{\text{loc}}(C) = \sum_{n=1}^N E_n(C)$ gives the expression stated in the theorem.

\section{Proof of Theorem~\ref{thm:comb_implementation}} \label{sec:proof_comb_implementation}
 
For an $N$-step energy-constrained quantum comb $C \in \Comb^{\leq E}_{\text{gl/loc}}[(\HH_0,\HH_1),\dots,(\HH_{2N-2},\HH_{2N-1})]$ with $\HH_{2n-2} \simeq \HH_{2n-1}$ ($n=1,\dots,N$), Thm.~\ref{thm:comb_implementation} contains three key statements:
\begin{itemize}
    \item The comb $C$ can be implemented using a sequence of operations $\{U_n\}_{n=1}^N$ with a control ancilla $\HH_C$ initialized in a pure state $\ket{\phi_0}$, where $U_n$ is the operation at the $n$-th step acting on the joint system of the comb $n$-th input and the control ancilla.
    \item For the sequence of operations $\{U_n\}_{n=1}^N$, the energy of the control ancilla remains constant throughout the process.
    \item The implementation constructed from the above operations and control ancilla can be realized with a global battery (or $N$ local batteries) charged with energy $E$.
\end{itemize} 

Here, the control ancilla system $\HH_C = \HH_A^{\otimes 3}$ is composed of three identical auxiliary systems $\HH_A$ with a sufficiently large dimension $d_A := \dim \HH_A \geq \prod_{n=0}^{2N-1} \dim \HH_n$.

The third point for the battery energy requirement is a direct corollary of the first two points. Specifically, the energy-preserved control ancilla guarantees it will not contribute to the energy change in each step. As a result, the energy change due to the unitary operations is equivalent to the energy difference between the inputs and outputs of the comb. This aligns with the characterization of the energy-constrained comb set $\Comb^{\leq E}_{\text{gl/loc}}$ [cf.~Eq.~\eqref{eq:energy_change}]. 
Therefore, we only need to prove the first two points, which are equivalent to the following lemma:
\begin{lemma}\label{lem:comb_implementation_eq}
     Let $C$ be an $N$-step quantum comb with $\HH_{2n-2} \simeq \HH_{2n-1}$. There always exists an operation sequence $\{U_n\}_{n=1}^N$ to implement $C$, where the isometry $U_n: \HH_{2n-2} \otimes \HH_C \to \HH_{2n-1} \otimes \HH_C$ is the operation at the $n$-th step, such that the energy of the control ancilla $\HH_C$, initialized in a pure state $\ket{\phi_0}$, remains constant throughout the process.
\end{lemma}

Notably, since the input space $\HH_{2n-2}$ and the output space $\HH_{2n-1}$ at each step are isomorphic, the isometry $U_n$ can be regarded as a unitary operation. We prove this lemma in the remainder of this section.

For the first $n$ steps of the comb ($1 \leq n \leq N$), the comb $C^{(n)}$ of this sub-process can be obtained by the recursive relation $\Tr_{2n+1}[C^{(n+1)}] = C^{(n)} \otimes I_{2n}$ for $n<N$ with $C^{(N)} = C$ (see Appendix.~\ref{sec:quantum_comb} for a detailed introduction).  
Consider the spectral decomposition
\begin{equation}
C^{(n)} = \sum_{j=0}^{M-1} |(K^{(n)}_j)^T\kk \bb (K^{(n)}_j)^T|,
\end{equation}
where $|\cdot\kk$ means the vectorization of an operator and $K^{(n)}_j: \HH_{0,2,\dots,2n-2} \to \HH_{1,3,\dots,2n-1}$. Without loss of generality, we assume $M=\prod_{n=0}^{2N-1} \dim \HH_n$, as any extra indices can be filled with $K^{(n)}_j=0$. Since $C^{(n)}$ is a Choi operator, tracing out the output space $\HH_{1,3,\dots,2n-1}$ yields an identity on the input space $\HH_{0,2,\dots,2n-2}$, which leads to $\sum_j {K^{(n)}_j}^\dagger K^{(n)}_j = I_{0,2,\dots,2n-2}$. For an operation sequence $\{U_n\}_{n=1}^N$, if the composition of the first $n$ operations $\hat{U}_n : \HH_{0,2,\dots,2n-2} \otimes \HH_C \rightarrow \HH_{1,3,\dots,2n-1} \otimes \HH_C$ given by 
\begin{equation}
\begin{cases}
    \hat{U}_1 = U_1, \\
    \hat{U}_n = (U_n \otimes I_{1,3,\dots,2n-3}) (\hat{U}_{n-1} \otimes I_{2n-2}),\, n=2,\dots,N,
\end{cases}
\end{equation}
such that
\begin{equation}\label{eq:hat_Un}
    \hat{U}_n (I_{0,2,\dots,2n-2} \otimes \ket{\phi_0}) = \sum_j K_j^{(n)} \otimes \ket{\phi_j}
\end{equation}
with an orthonormal set $\{\ket{\phi_j}\}$ on $\HH_C$ for all $n=1,\dots, N$, then these operations implement the quantum comb $C$. This is because the Choi operator of this process is
\begin{equation}
\Tr_{\HH_C}[(I \otimes \hat{U}_n) (|I\kk \otimes \ket{\phi_0}) (\bb I| \otimes \bra{\phi_0}) (I \otimes \hat{U}_n^\dagger)] = \sum_j (I \otimes  K_j^{(n)}) |I\kk \bb I| (I \otimes K_j^{(n)})^\dagger,
\end{equation}
which equals to $C^{(n)}$ as $(I \otimes  K_j^{(n)}) |I\kk = |(K^{(n)}_j)^T\kk$. Here, since $\sum_j {K^{(n)}_j}^\dagger K^{(n)}_j = I_{0,2,\dots,2n-2}$, Eq.~\eqref{eq:hat_Un} is always consistent with the conditions for unitary operators.

Before we proceed to the construction of $\{U_n\}_{n=1}^N$, we first show in Lem.~\ref{lem:energy_constant_cond} that, with a specific choice of the states $\{\ket{\phi_j}\}$, if we can construct $\{U_n\}_{n=1}^N$ satisfying \eqref{eq:hat_Un} for $n=1,\dots,N$ then the energy of the control ancilla will be guaranteed to be constant throughout the process.

Suppose the Hamiltonian of the system  $\HH_A$ is $H_A = \sum_{j=0}^{d_A-1} E_j^A\ket{e_j}\bra{e_j}$ with $d_A := \dim \HH_A$. Choose 
\begin{equation}\label{eq:phi_val}
    \ket{\phi_j} = \frac{1}{\sqrt{M}} \sum_{k=0}^{M-1} \omega_M^{jk} F\ket{e_j} \otimes F\ket{e_k} \otimes F\ket{e_k} \in \HH_C
\end{equation}
with $F$ being the quantum Fourier transform on $\{\ket{e_j}\}_{j=0}^{d_A-1}$ and $\omega_M$ being the $M$-th root of unity.
\begin{lemma}\label{lem:energy_constant_cond}
    Let $\{U_n\}_{n=1}^N$ be a sequence of operations implementing the comb $C$ with $U_n : \HH_{2n-2} \otimes \HH_C \to \HH_{2n-1} \otimes \HH_C$. If the composition $\hat{U}_n : \HH_{0,2,\dots,2n-2} \otimes \HH_C \rightarrow \HH_{1,3,\dots,2n-1} \otimes \HH_C$ of the first $n$ operations $U_1,\dots,U_n$ satisfies \eqref{eq:hat_Un} for $n=1,\dots,N$, then the state on $\HH_C$, initialized as $\ket{\phi_0}$, has a constant energy $3\bar{E}$ with $\bar{E} := \sum_j E_j^A / d_A$ before and after each operation. 
\end{lemma}
\begin{proof}
    First, it can be checked that the energy of the initial state $\ket{\phi_0}$ is $3\bar{E}$ by direct calculation. Next, we show that the energy on $\HH_C$ is also $3\bar{E}$ after each operation.

    Suppose the process of the comb $C$ is sequentially composed with another process represented by a comb $C_0 \in \Comb[(\emptyset, \HH_0),(\HH_1, \HH_2), \dots, (\HH_{2N-3},\HH_{2N-2})]$, where $\emptyset$ refers to the trivial space. 
    
    Let $C_0^{(n)}$ be the comb of the first $n$ steps of $C_0$. The comb of the first $n$ steps of $C$ without discarding $\HH_C$ is 
    \begin{equation}
    (I \otimes \hat{U}_n) (|I\kk \otimes \ket{\phi_0}) (\bb I| \otimes \bra{\phi_0}) (I \otimes \hat{U}_n^\dagger) = \sum_{j,j'} |(K^{(n)}_j)^T\kk \bb (K^{(n)}_{j'})^T| \otimes \ket{\phi_j}\bra{\phi_{j'}}.
    \end{equation}
    The state on $\HH_C$ after the $n$-th operation $U_n$ can be determined by the link product between these two combs as
    \begin{equation}
    \begin{aligned}
        \rho_C^{(n)} := &\Tr_{0,\dots,2n-1} \left[ \left(\sum_{j,j'} |(K^{(n)}_j)^T\kk \bb (K^{(n)}_{j'})^T| \otimes \ket{\phi_j}\bra{\phi_{j'}}\right) \left((C^{(n)}_0)^T \otimes I_{\HH_{2n-1}, \HH_C} \right) \right] \\
        =& \sum_{j,j'} \Tr\left[ |(K^{(n)}_j)^T\kk \bb (K^{(n)}_{j'})^T| \left((C^{(n)}_0)^T \otimes I_{2n-1} \right) \right] \ket{\phi_j}\bra{\phi_{j'}} \\
        =& \sum_{j,j'} \varrho_{j,j'} \ket{\phi_j}\bra{\phi_{j'}}.
    \end{aligned}
    \end{equation}
    The energy of $\rho_C^{(n)}$ should be
    \begin{equation}\label{eq:energy_rho_n}
    \Tr[(H_A \otimes I\otimes I + I\otimes H_A\otimes I + I\otimes I\otimes H_A)\rho_C^{(n)}].
    \end{equation}
    With $H_A = \sum_{j=0}^{d_A-1} E_j^A\ket{e_j}\bra{e_j}$ and the property of the $M$-th root of unity
        \begin{equation}
        \sum_{k=0}^{M-1}\omega_M^{(j-j')k}=\begin{cases}
            0, & j\neq j', \\
            M, & j=j',
        \end{cases}
        \end{equation}
    the first term of \eqref{eq:energy_rho_n} can be determined:
    \begin{equation}
        \begin{aligned}
            \Tr[(H_A \otimes I\otimes I)\rho_C^{(n)}]
            &= \frac{1}{M} \sum_{j,j'=0}^{M-1} \sum_{k=0}^{M-1} \sum_{l=0}^{d_A-1} \varrho_{j,j'} E_l^A  \omega_M^{(j-j')k} \bra{e_l}F\ket{e_j}\bra{e_{j'}}F^\dagger\ket{e_l} \\
            &= \sum_{j=0}^{M-1} \sum_{l=0}^{d_A - 1} \varrho_{j,j} E_l^A \bra{e_l}F\ket{e_j}\bra{e_j}F^\dagger\ket{e_l} \\
            &= \sum_{j=0}^{M-1} \sum_{l=0}^{d_A - 1} \varrho_{j,j} E_l^A  \frac{1}{d_A} \\
            &= \sum_{j=0}^{M-1} \varrho_{j,j} \bar{E}.
        \end{aligned}
    \end{equation}
    Noticing that $\sum_j \varrho_{j,j} = \Tr\left[\sum_j |(K^{(n)}_j)^T\kk \bb (K^{(n)}_{j})^T| \left((C^{(n)}_0)^T \otimes I_{2n-1} \right) \right] = \Tr[C^{(n)} * C_0^{(n)}] = 1$ is the trace of the output state on $\HH_{2n-1}$, we can obtain $\Tr[(H_A \otimes I\otimes I)\rho_C^{(n)}] = \bar{E}$.
    Additionally, by $\sum_j \varrho_{j,j}=1$, the second term and the third term of \eqref{eq:energy_rho_n} can be derived similarly:
    \begin{equation}
        \begin{aligned}
            \Tr[(I\otimes H_A \otimes I)\rho_C^{(n)}]=\Tr[(I\otimes I\otimes H_A)\rho_C^{(n)}] 
            &=\frac{1}{M}\sum_{j=0}^{M-1}\sum_{k=0}^{M-1} \sum_{l=0}^{d_A - 1} \varrho_{j,j} E_{l}^A  \left|\bra{e_l}F\ket{e_k}\right|^2 \\
            &= \frac{1}{M}\sum_{k=0}^{M-1} \sum_{l=0}^{d_A - 1} E_{l}^A |\bra{e_l}F\ket{e_k}|^2 \\
            &= \frac{1}{M}\sum_{k=0}^{M-1} \sum_{l=0}^{d_A - 1} E_{l}^A \frac{1}{d_A} \\
            &= \frac{1}{M}\sum_{k=0}^{M-1} \bar{E} \\
            &= \bar{E}.
        \end{aligned}
    \end{equation}
    Hence, the energy of $\rho^{(n)}_C$ is always $3\bar{E}$ regardless of the composed comb $C_0$.

    Since the above derivation holds for arbitrary $1\leq n \leq N$, the energy on $\HH_C$ remains as $3\bar{E}$ throughout the process.
\end{proof}

Finally, we conclude the proof by giving a construction of $\{U_n\}_{n=1}^N$ satisfying \eqref{eq:hat_Un}.

\begin{lemma}
    For an $N$-step quantum comb $C$, there exists a sequence of isometries $\{U_n\}_{n=1}^N$ satisfying \eqref{eq:hat_Un} with $\{\ket{\phi_j}\}$ given in Lem.~\ref{lem:energy_constant_cond}.
\end{lemma}
\begin{proof}
    For $N=1$, we can just pick $U_1 = \hat{U}_1$. For $N>1$, with $U_1 = \hat{U}_1$, we construct $U_2,\dots,U_N$ in the following part. 
    
    In fact, if we input a quantum state on the input space $\HH_{0,2,\dots,2n-2}$ instead of interacting with a general quantum comb, the process $C^{(n)}$ of the first $n$ steps can be regarded as a quantum channel $\mathcal{E}_n$ from the input space $\HH_{0,2,\dots,2n-2}$ to the output space $\HH_{1,3,\dots,2n-1}$. Then the operators $\{K_j^{(n)}\}_{j=0}^{M-1}$ become the Kraus operators of $\mathcal{E}_n$.

    
    For $n>1$, by the relation $\Tr_{2n-1}[C^{(n)}] = C^{(n-1)} \otimes I_{2n-2}$ of the quantum comb, we have $\Tr_{2n-1}[\mathcal{E}_n(\cdot)]=\mathcal{E}_{n-1}(\Tr_{2n-2}[\cdot])$, which implies both sides are the same channel.  Thus, with the energy eigenbasis $\{\ket{l}_{2n-1}\}_l$ for $\HH_{2n-1}$ and the energy eigenbasis $\{\ket{l'}_{2n-2}\}_{l'}$ for $\HH_{2n-2}$, their Kraus operators $\{(I_{1,3,\dots,2n-1} \otimes \bra{l}_{2n-1})K_j^{(n)}\}_{l,j}$ and $\{K_k^{(n-1)}\otimes \bra{l'}_{2n-2}\}_{l',k}$ are equivalent up to some isometry. That is, there exists an isometry $V=\sum_{l,l',j,k} V_{l j,l' k} \ket{l}_{2n-1} \bra{l'}_{2n-2}\otimes \ket{\varphi_j}_A \bra{\varphi_k}_A$, such that
    	\begin{equation}
    	    (I_{1,3,\dots,2n-3}\otimes\bra{l}_{2n-1})K_j^{(n)}=\sum_{l',k}V_{l j,l' k}K_k^{(n-1)}\otimes \bra{l'}_{2n-2}.
    	\end{equation}
    	Here, $\{\ket{\varphi_j}_A\}_j$ is an orthonormal basis of $\HH_A$ to be determined later. Then, by $\sum_l (I_{1,3,\dots,2n-3}\otimes \ket{l}_{2n-1}\bra{l}_{2n-1}) = I_{1,3,\dots,2n-1}$,
        \begin{equation}
    	\begin{aligned}
    		K_j^{(n)}
            &= \sum_{l,l',k} V_{l j,l' k} K_k^{(n-1)} \otimes  \ket{l'}_{2n-1} \bra{l}_{2n-2} \\
            &= \sum_k K_k^{(n-1)} \otimes \left[(I_{2n-1} \otimes \bra{\varphi_j}_A) V (I_{2n-2} \otimes \ket{\varphi_k}_A) \right] \\
    		&= (I_{1,3,\dots,2n-1} \otimes \bra{\varphi_j}_A) (V \otimes  I_{1,3,\dots,2n-3}) \sum_{k}K_k^{(n-1)} \otimes I_{2n-2} \otimes \ket{\varphi_k}_A,
    	\end{aligned}
        \end{equation}
    	or equivalently,
    	\begin{equation}
        \sum_{j} K_j^{(n)} \otimes \ket{\varphi_j}_A = (V\otimes I_{1,3\dots,2n-3})\sum_{j} K_j^{(n-1)} \otimes I_{2n-2} \otimes \ket{\varphi_j}_A.
        \end{equation}
    	Let $\ket{\varphi_j}_A=\omega_M^{jk} F\ket{e_j}$ with $k=0,\dots,N-1$, it follows that
    \begin{equation}\label{eq:Kraus_relation}
        \sum_{j} K_j^{(n)} \otimes  \omega_M^{jk}F\ket{e_j} = (U^{(n)}_k \otimes I_{1,3\dots,2n-3}) \sum_{j} K_j^{(n-1)}\otimes I_{2n-2} \otimes \omega_M^{jk}F\ket{e_j}
    \end{equation}
    	for some isometry $U_k^{(n)} : \HH_{2n-2} \otimes \HH_A \rightarrow \HH_{2n-1} \otimes \HH_A$. Since $\HH_C \simeq \HH_A^{\otimes 3}$, we can construct $U_n: \HH_{2n-2} \otimes \HH_C \rightarrow \HH_{2n-1} \otimes \HH_C$ as 
    	\begin{equation}
        U_n=\sum_{k=0}^{M-1} U_k^{(n)} \otimes F\ket{e_k}\bra{e_k}F^\dagger \otimes I_{\HH_A}.
        \end{equation}

    We then show \eqref{eq:hat_Un} is satisfied with this construction by induction. Since we choose $U_1 = \hat{U}_1$, it holds for $n=1$. If it holds for $n-1$ with $n>1$, we have
    \begin{equation}
    \hat{U}_{n-1} (I_{1,3,\dots,2n-3} \otimes \ket{\phi_0}) = \sum_j K_j^{(n-1)} \otimes \ket{\phi_j}.
    \end{equation}
    Using the relation $\hat{U}_n = (U_n \otimes I_{1,3,\dots,2n-3}) (\hat{U}_{n-1} \otimes I_{2n-2})$ and \eqref{eq:Kraus_relation}, we can obtain
    \begin{equation}
    \begin{aligned}
        \hat{U}_n (I_{0,2,\dots,2n-2} \otimes \ket{\phi_0})
        =& (U_n \otimes I_{1,3,\dots,2n-3}) (\hat{U}_{n-1} \otimes I_{2n-2}) (I_{0,2,\dots,2n-2} \otimes \ket{\phi_0}) \\
        =& (U_n \otimes I_{1,3,\dots,2n-3}) \sum_{j} K_j^{(n-1)} \otimes I_{2n-2} \otimes \ket{\phi_j}.
    \end{aligned}
    \end{equation}
    Now substitute the expression for $\{\ket{\phi_j}\}$ from \eqref{eq:phi_val} into the above equation:
    \begin{equation}
    \begin{aligned}
        \hat{U}_n (I_{0,2,\dots,2n-2} \otimes \ket{\phi_0})
        =& (U_n \otimes I_{1,3,\dots,2n-3}) \sum_{j} K_j^{(n-1)} \otimes I_{2n-2} \otimes \left( \frac{1}{\sqrt{M}} \sum_{k} \omega_M^{jk} F\ket{e_j} \otimes F\ket{e_k} \otimes F\ket{e_k} \right) \\
        =& \frac{1}{\sqrt{M}} \sum_{k} \left[ (U_k^{(n)} \otimes I_{1,3,\dots,2n-3}) \sum_j K_j^{(n-1)} \otimes I_{2n-2} \otimes \omega_{M}^{jk} F\ket{e_j} \right] \otimes F\ket{e_k} \otimes F\ket{e_k} .
    \end{aligned}
    \end{equation}
    Applying \eqref{eq:Kraus_relation} to the above equation finally yields \eqref{eq:hat_Un} for $n$:
    \begin{equation}
    \begin{aligned}
        \hat{U}_n (I_{0,2,\dots,2n-2} \otimes \ket{\phi_0})
        =& \frac{1}{\sqrt{M}} \sum_{k} \left[ \sum_j K_j^{(n)} \otimes \omega_{M}^{jk} F\ket{e_j} \right] \otimes F\ket{e_k} \otimes F\ket{e_k} \\
        =& \sum_j K_j^{(n)} \otimes \ket{\phi_j} .
    \end{aligned}
    \end{equation}
\end{proof}

\section{Proof of Theorem~\ref{thm:qubit_optimal_estimator}} \label{sec:proof_qubit_optimal_estimator}
To demonstrate that the range of an optimal estimator can be restricted to $\{0,\pi\}$ when $\{C_\theta\}$ is covariant and $E_{\text{gl/loc}}(C_\theta) \leq 0$, we begin by establishing an auxiliary lemma concerning the energy consumption of a sequentially composed comb.
\begin{lemma}
\label{lem:compose_energy}
    Let $C_1 \in \Comb[(\HH_1,\HH_2'), \dots, (\HH_{2N-1}, \HH_{2N}')]$ and $C_2 \in \Comb[(\HH_2',\HH_2), \dots, (\HH_{2N}', \HH_{2N})]$ be $N$-step quantum combs with general spaces $\HH_1,\dots,\HH_{2N}$ and $\HH_2',\dots,\HH'_{2N}$. If we concatenate the output spaces of $C_1$ with the input spaces of $C_2$ to obtain an $N$-step comb $C_1 * C_2$, the energy consumption of the concatenation $C_1 * C_2$ is bounded by the sum of the energy consumptions of these two combs, i.e., 
    \[
    E_{\text{gl/loc}}(C_1 * C_2) \leq E_{\text{gl/loc}}(C_1) + E_{\text{gl/loc}}(C_2).
    \]
\end{lemma}
\begin{proof}
    Denote by $H_{n_1,\dots,n_k}$ ($H_{n_1,\dots,n_k}'$) the Hamiltonian of $\HH_{n_1} \otimes \cdots \otimes \HH_{n_k}$ ($\HH_{n_1}' \otimes \cdots \otimes \HH_{n_k}'$). The composition is
    \begin{equation}
        C := C_1 \circ C_2 \in \Comb[(\HH_1,\HH_2), \dots, (\HH_{2N-1}, \HH_{2N})].
    \end{equation}
    By the expression of the energy consumption given in \eqref{eq:global_energy_consumption} and \eqref{eq:local_energy_consumption}, we have
    \begin{equation}
        E_{\text{gl}}(C) = \max_{n=1,\dots,N} \lambda_{\max} \left( \Tr_{2,4,\dots,2N}[(H_{2,4,\dots,2n} \otimes I) C] - (H_{1,3,\dots,2n-1})^T \otimes I  \right)
    \end{equation}
    and
    \begin{equation}
        E_{\text{loc}}(C) = \sum_{n=1}^N \max\left\{\lambda_{\max} \left( \Tr_{2,4,\dots,2N}[(H_{2n} \otimes I) C] - (H_{2n-1})^T \otimes I \right),0\right\}.
    \end{equation}
    For $C_1$ and $C_2$, their energy consumption can be similarly obtained by correspondingly replacing $\HH_{2n}$ with $\HH_{2n}'$ and $H_{2n}$ by $H_{2n}'$. 

    For convenience, we denote both the term $H_{2,4,\dots,2n} \otimes I_{2(n+1),\dots,2N}$ ($H_{1,3,\dots,2n-1} \otimes I_{2n+1,\dots,2N-1}$) for the global battery and the term $H_{2n} \otimes I_{2,\dots,2n-2,2n+2,\dots,2N}$ ($H_{2n-1} \otimes I_{1,\dots,2n-3,2n+1,\dots,2N-1}$) for the local battery in the above equations by a unified notation $H^{(n)}_\out$ ($H^{(n)}_\In$). Similarly, $H'^{(n)}$ represents either $H'_{2,4,\dots,2n} \otimes I_{\HH'_{2n+2},\dots,\HH_{2N}}$ or $H'_{2n} \otimes I_{\HH'_{2},\dots,\HH_{2n-2},\HH_{2n+2},\dots,\HH_{2N}}$ correspondingly. With these notations, the energy consumption can be rewritten as
    \begin{equation}
        E_{\text{gl}}(C) = \max_{n=1,\dots,N} \lambda_{\max} \left( \Tr_{2,4,\dots,2N}[(H^{(n)}_\out \otimes I) C] - (H^{(n)}_\In)^T  \right)
    \end{equation}
    for the global battery and
    \begin{equation}
        E_{\text{loc}}(C) = \sum_{n=1}^N \max\left\{\lambda_{\max} \left( \Tr_{2,4,\dots,2N}[(H^{(n)}_\out \otimes I) C] - (H^{(n)}_\In)^T\right),0\right\}.
    \end{equation}
    for the local battery, where the specific value of $H^{(n)}$ depends on the type of battery model we are considering. Therefore, to prove $E_{\text{gl/loc}}(C_1 * C_2) \leq E_{\text{gl/loc}}(C_1) + E_{\text{gl/loc}}(C_2)$, with the energy consumption of the first $n$ steps given by
    \begin{equation}
    E^{(n)}(C) = \lambda_{\max}\left( \Tr_{2,4,\dots,2N}[(H^{(n)}_\out \otimes I) C] - (H^{(n)}_\In)^T \right),
    \end{equation}
    it suffices to show $E^{(n)}(C) \leq E^{(n)}(C_1) + E^{(n)}(C_2)$. Here, $E^{(n)}(C_1)$ and $E^{(n)}(C_2)$ represent the corresponding values determined by the input and output spaces associated with the combs $C_1$ and $C_2$, respectively.
    
    To see this, we begin with
    \begin{equation}
        \begin{aligned}
             \Tr_{2,4,\dots,2N}[(H^{(n)}_\out \otimes I) C] 
            = \Tr_{2,4,\dots,2N}[(I_{1,3,\dots,2N-1} \otimes H^{(n)}_\out) (C_1 \circ C_2)].
        \end{aligned}
    \end{equation}
    Let $\HH' := \HH_2' \otimes \HH_4' \otimes \cdots \otimes \HH_{2N}'$. By the link product, the concatenation can be determined as
    \begin{equation}
    C_1 \circ C_2 = \Tr_{\HH'}[(C_1^{T_{\HH'}} \otimes I_{2,4,\dots,2N})(I_{1,3,\dots,2N-1} \otimes C_2)].
    \end{equation}
    Then, we have
    \begin{equation}
        \begin{aligned}
            & \Tr_{2,4,\dots,2N}[(H^{(n)}_\out \otimes I) C] \\
            =& \Tr_{2,4,\dots,2N}\left[(I_{1,3,\dots,2N-1} \otimes H^{(n)}_\out) \Tr_{\HH'}[(C_1^{T_{\HH'}} \otimes I_{2,4,\dots,2N})(I_{1,3,\dots,2N-1} \otimes C_2)] \right] \\
            =&  \Tr_{\HH'}\left[C_1^{T_{\HH'}}\left(I_{1,3,\dots,2N-1} \otimes  \Tr_{2,4,\dots,2N}[(I_{\HH'} \otimes H^{(n)}_\out) C_2)] \right) \right] \\
            =& \Tr_{\HH'}\left[C_1\left(I_{1,3,\dots,2N-1} \otimes  \Tr_{2,4,\dots,2N}[(I_{\HH'} \otimes H^{(n)}_\out) C_2)] \right)^T \right].
        \end{aligned}
    \end{equation}
    Observing that $E^{(n)}(C_2) = \lambda_{\max}\left( \Tr_{2,4,\dots,2N}[(I_{\HH'} \otimes H^{(n)}_\out) C_2)] - (H'^{(n)})^T \right)$, we obtain the operator inequality $\Tr_{2,4,\dots,2N}[(I_{\HH'} \otimes H^{(n)}_\out) C_2)] \preceq (H'^{(n)})^T + E^{(n)}(C_2) I_{\HH'}$. This leads to
    \begin{equation}
        \begin{aligned}
            & \Tr_{2,4,\dots,2N}[(H^{(n)}_\out \otimes I) C] \\
            \preceq& \Tr_{\HH'}\left[C_1\left(I_{1,3,\dots,2N-1} \otimes \left((H'^{(n)})^T + E^{(n)}(C_2) I_{\HH'}\right)  \right)^T \right] \\
            =&  \Tr_{\HH'}\left[(I_{1,3,\dots,2N-1} \otimes  H'^{(n)})C_1 \right] + E^{(n)}(C_2) \Tr_{\HH'}[C_1]. 
        \end{aligned}
    \end{equation}
    For the first term, we have $\Tr_{\HH'}\left[(I_{1,3,\dots,2N-1} \otimes  H'^{(n)})C_1 \right] \preceq (H^{(n)}_\In)^T + E^{(n)}(C_1) I_{1,3,\dots,2N-1}$. For the second term, since tracing out the output space of a Choi operator leads to the identity on the input space, we have $\Tr_{\HH'}[C_1] = I_{1,3,\dots,2N-1}$. It follows that
    \begin{equation}
        \begin{aligned}
            \Tr_{2,4,\dots,2N}[(H^{(n)}_\out \otimes I) C] \preceq (H^{(n)}_\In)^T + E^{(n)}(C_1) I_{1,3,\dots,2N-1} + E^{(n)}(C_2) I_{1,3,\dots,2N-1}.
        \end{aligned}
    \end{equation}
    Thus,
    \begin{equation}
    E^{(n)}(C) := \lambda_{\max}\left( \Tr_{2,4,\dots,2N}[(H^{(n)}_\out \otimes I) C] - (H^{(n)}_\In)^T \right) \leq E^{(n)}(C_1) + E^{(n)}(C_2).
    \end{equation}
\end{proof}

Then, we derive a result implying that we can always shift the values of the estimator without any precision loss. 
\begin{lemma}
\label{lem:phase_shift}
Let $\theta_0 \in [0,2\pi]$ be an arbitrarily given shift value. Assume that all input and output spaces are isomorphic and share the same Hamiltonian. Suppose $\{C_\theta\}$ is a covariant comb with non-positive energy consumption, satisfying $E_{\text{gl/loc}}(C_\theta) \leq 0$. For any probe instrument $\{T_x\}$ with an optimal estimator $\hat{\theta}$, there always exists another probe instrument $\{T_x'\}$ with an estimator $\hat{\theta}'$ defined by $\hat{\theta}'(x) = (\hat{\theta}(x) - \theta_0) \mod 2\pi$, which achieves the same average cost, i.e., $\bar{c}(T, \hat{\theta}) = \bar{c}(T', \hat{\theta}')$.
\end{lemma}
\begin{proof}
First, by extending the process of $C_\theta$ with an identity map on the final output space $\HH_{2N-1}$ of the probe $T$, we can ensure that the combs $C_\theta$ and $T$ both are $N$-step combs. Therefore, we can concatenate the output of $T$ with the input of $C_\theta$ to obtain $T' = T \circ C_{\theta_0}$, which is also an $N$-step probe comb with an isomorphism. And, Lem.~\ref{lem:compose_energy} implies that the energy consumption of $T'$ is at most the same as that of $T$, so it is also a feasible probe. Meanwhile, note that $T'*C_\theta = (T \circ C_{\theta_0}) * C_\theta = T * (C_{\theta_0} \circ C_\theta) = T * C_{\theta+\theta_0}$ holds by the associative law for the link product~\cite{chiribellaTheoreticalFrameworkQuantum2009}, as the concatenation $\circ$ is a special form of the link product $*$. This results in the following relationship:
\begin{equation}
\begin{aligned}
    \Tr[T_x' \JJ]
    &= \frac{1}{\int_\theta \dd\theta} \int_\theta \ee^{-\ii\theta} \Tr[T_x' C_\theta^T] \dd\theta \\
    &= \frac{1}{\int_\theta \dd\theta} \int_\theta \ee^{-\ii\theta} \Tr[T_x C_{\theta+\theta_0}^T] \dd\theta \\
    &= \frac{1}{\int_\theta \dd\theta} \ee^{\ii\theta_0} \int_\theta \ee^{-\ii\theta} \Tr[T_x C_\theta^T] \dd\theta \\
    &= \ee^{\ii\theta_0} \Tr[T_x \JJ].
\end{aligned}
\end{equation}
The last equation is from the definition of the task operator $\JJ := \frac{1}{\int_\theta \dd\theta} \int_\theta \mathrm{e}^{-\mathrm{i}\theta}C_\theta^T\mathrm{d}\theta$. Since the optimal estimator for $T$ can be determined as $\hat{\theta}(x) = -\arg(\Tr[T_x \JJ])$, the optimal estimator for $T'$ is $\hat{\theta}'(x) = -\arg(\Tr[T_x' \JJ]) = (\hat{\theta}(x) - \theta_0) \mod 2\pi$.  
\end{proof}
This conclusion implies that we can assume $\hat{\theta}(0) =0 $ without loss of generality. In the qubit case, the value of $\hat{\theta}(1)$ can then also be fixed.
\begin{proof}[Proof of Thm.~\ref{thm:qubit_optimal_estimator}]
    By Lem.~\ref{lem:phase_shift}, we assume that $\hat{\theta}(0)=0$ is optimal without loss of generality. For any probe comb $T$ with $T=\sum_x T_x \otimes \ket{x}\bra{x}$, by the definition of the link product, we have
    \begin{equation}
        \begin{aligned}
            \sum_x \Tr[T_x \JJ]
            =& \int_\theta \ee^{-\ii\theta} \sum_x \Tr[T_x C_\theta^T] \dd\theta \\
            =& \int_\theta \ee^{-\ii\theta} \Tr[T * C_\theta] \dd\theta \\
            =& \int_\theta \ee^{-\ii\theta} \dd\theta \\
            =& 0,
        \end{aligned}
    \end{equation}
where $\Tr[T * C_\theta]=1$ holds because $T*C_\theta$ is the final output state on $\HH_{2N-1}$.
For a qubit, it holds that $x \in \{0,1\}$, which leads to $\Tr[T_1\JJ] = -\Tr[T_0\JJ]$. Since the optimal value of $\hat{\theta}(0)$ is $- \arg(\Tr[T_0\JJ])$, based on our assumption $\hat{\theta}(0)=0$, we can obtain $\Tr[T_0\JJ] \in \R$. Therefore, the optimal value of $\hat{\theta}(1)$ is $- \arg(\Tr[T_1\JJ]) = -\arg(-\Tr[T_0\JJ])$, which should be either $0$ or $\pi$ modulo $2\pi$.
\end{proof}

\section{Energy-constrained Qubit Phase Estimation} \label{sec:proof_phase_estimation}

The main result of this section is summarized in the following theorem, whose proof is provided in the remaining part of the section.

\begin{theorem}[Optimal qubit phase estimation]\label{thm:optimal_qubit_phase}
    When the probe is $E$-globally energy-constrained, the optimal average cost for the qubit phase channel is
    \[
       \bar{c}_{\text{phase}}(E,2) = \begin{cases}
            2 - 4(1 - \sqrt{1-E}) \sqrt{1-E}, & E \leq 3/4, \\
            1, & E > 3/4.
        \end{cases}
    \]
For $E \leq 3/4$, the quantum circuit corresponding to the optimal strategy is shown in Fig.~\ref{fig:phase_strategy_circuit}. With $\kappa := 1- \sqrt{1-E}$, in Fig.~\ref{fig:phase_strategy_circuit}, 
\[
\rho_0 = \sqrt{1 - \kappa} \ket{0}\ket{\phi_0} + \sqrt{\kappa} \ket{1}\ket{\phi_1},
\]
and $U$ is defined as
\begin{gather*}
    U \ket{0}\ket{\phi_0} = \sqrt{1 - \kappa} \ket{0}\ket{\phi_0} + \sqrt{\kappa} \ket{1}\ket{\phi_1},  \\
      U \ket{1}\ket{\phi_1} = \sqrt{\kappa} \ket{0}\ket{\phi_0} - \sqrt{1 - \kappa} \ket{1}\ket{\phi_1},  \\
    U \ket{1}\ket{\phi_0} = \ket{1}\ket{\phi_0},\, 
    U \ket{0}\ket{\phi_1} = \ket{0}\ket{\phi_1},
\end{gather*}
where $\ket{\phi_0}$ and $\ket{\phi_1}$ are two orthonormal states of the control ancilla given in \eqref{eq:phi_val}.
\end{theorem} 
\begin{figure}[htbp]
    \centering
    \includegraphics[width=0.35\linewidth]{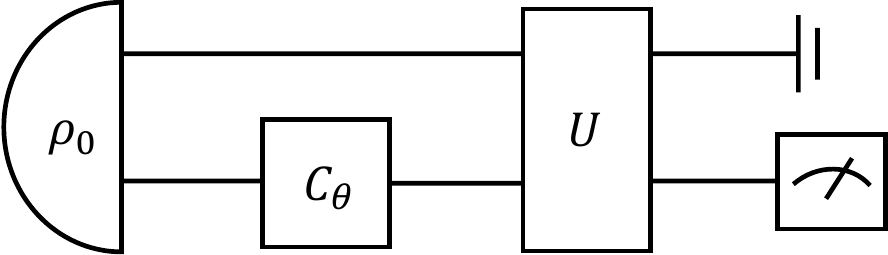}
    \caption{Optimal strategy for the qubit phase channel}
    \label{fig:phase_strategy_circuit}
\end{figure}

Our construction of $\ket{\phi_0}$ and $\ket{\phi_1}$ in \eqref{eq:phi_val} ensures the energy of the control ancilla remains unchanged throughout the process, as stated in Theorem~\ref{thm:comb_implementation}.
Consequently, the ancilla consumes no energy after initialization, and the energy it carries can be recycled at the end of the process. 

We then give the proof of the above result. For a qubit phase channel diagonal in the energy eigenbasis with Hamiltonian $H=\ket{1}\bra{1}$, we have $C_\theta = \sum_{i,j=0}^1 \ket{i}\bra{j} \otimes \ee^{\ii(i-j)\theta} \ket{i}\bra{j}$, and thus by the definition of the operator $\JJ$,   
\begin{equation}
    \JJ := \frac{1}{2\pi}\int_{\theta \in [0,2\pi]} \ee^{-\ii \theta} C_{\theta}^T \dd\theta = \ket{0}\bra{1} \otimes \ket{0}\bra{1}.
\end{equation}
By Thm.~\ref{thm:qubit_optimal_estimator} 
in the main text, we can assume $\hat{\theta}(0)=0, \hat{\theta}(1)=\pi$. Therefore, the average cost is 
\begin{equation}\label{eq:average_cost_phase}
\bar{c}(T, \hat{\theta}) := 2 - 2 \sum_{x=0}^{1} \Re[\ee^{\ii\hat{\theta}(x)} \Tr[T_x \JJ]] = 2 - 2 \Re[\Tr[T_0 \JJ]] + 2 \Re[\Tr[T_1 \JJ]].
\end{equation}
Since 
\begin{equation}
    \Tr[T_0 \JJ] + \Tr[T_1 \JJ] = \int_\theta \ee^{-\ii\theta} \Tr[T * C_\theta] \dd\theta = \int_\theta \ee^{-\ii\theta} \dd\theta = 0,
\end{equation}
the average cost determined in \eqref{eq:average_cost_phase} can be further simplified as
\begin{equation}
\bar{c}(T, \hat{\theta}) = 2 - 4 \Re[\Tr[T_0 \JJ]].
\end{equation}
Hence, our goal is to minimize the above objective function with $T \in \Comb_{\text{gl}}^{\leq E}$. Expanding the constraints for $\Comb_{\text{gl}}^{\leq E}$ in an explicit form by Thm.~\ref{thm:energy_constrained_comb} 
and the characterization of quantum combs
\begin{equation}
    \begin{cases}
        \Tr_{3}[T] = T^{(1)} \otimes I_{2}, \\
        \Tr[T^{(1)}] = 1,
    \end{cases}
\end{equation}
the problem can be formulated as an optimization:
\begin{align}
		\bar{c}_{\text{phase}}(E,2) = \min_{T_x} \quad & 2 - 4\Re[ \Tr[T_0 \JJ]] \\
		\mathrm{s.t.} \quad & T_0 + T_1 = T^{(1)} \otimes I \label{eq:sum_Tx} \\
        & \Tr[T^{(1)}] = 1, \\
        & \Tr[HT^{(1)}] \leq E, \label{eq:energy_rho} \\
        & \Tr_{1,3}[(H \otimes I \otimes I + I \otimes I \otimes H) T] \preceq H^T + EI, \label{eq:energy_T} \\
		& T = \sum_{x=0}^1 T_x \otimes \ket{x}\bra{x}, \label{eq:T} \\
		& T \succeq 0 .
\end{align}
We divide the proof into two parts: $E \geq 3/4$ and $E < 3/4$.
\begin{lemma}
    \label{lem:phase_bound_general}
    It holds that $\bar{c}_{\text{phase}}(E,2) \geq 1$ for any $E \geq 0$. Moreover, this lower bound can be saturated when $E \geq 3/4$.
\end{lemma}
\begin{proof}
    We omit the energy constraints to conclude an optimal probe, and then back to check its energy consumption. 
     Let $T_x = \sum_{i,j,k,l=0}^1 T_x^{ijkl} \ket{i}\bra{j} \otimes \ket{k}\bra{l}$ and $T^{(1)} = \sum_{i,j=0}^1 T^{(1)}_{ij} \ket{i}\bra{j}$. Since $\Re[\Tr[T_0 \JJ]] = \Re[T_0^{1010}]$ as $\JJ = \ket{0}\bra{1} \otimes \ket{0}\bra{1}$, it suffices to obtain an upper bound for $\Re[T_0^{1010}]$.  
    
    For convenience, we denote $T_0^{1010}$ by $t$. The objective becomes 
    \begin{equation}
        \bar{c}_{\text{phase}}(E,2) = 2 - 4 \max_{T_0} \Re[t].
    \end{equation}
Additionally, suppose the diagonal elements of $T_0$ are $a,b,c,d$.

    With $T_0,T_1 \succeq 0$, \eqref{eq:sum_Tx} leads to $0 \preceq T_0 \preceq T^{(1)} \otimes I$. Thus, the diagonal elements should satisfy
\begin{equation} \label{eq:phase_positive_constraint}
    \begin{cases}
        0 \leq a,b \leq T^{(1)}_{00} = 1-T^{(1)}_{11}, \\
        0 \leq c,d \leq T^{(1)}_{11}.
    \end{cases}
\end{equation}
Meanwhile, consider the $2\times 2$ submatrices formed by the elements at the four corners of the matrices $T_0$ and $T^{(1)} \otimes I - T_0 \succeq 0$, both of which should also be positive semidefinite. It follows that
\begin{equation}
    \begin{cases}
        |t|^2 \leq ad, \\
        |t|^2 \leq (T^{(1)}_{00} - a)(T^{(1)}_{11} - d) = (1 - T^{(1)}_{11} - a)(T^{(1)}_{11} - d).
    \end{cases}
\end{equation}
This implies 
\begin{equation}\label{eq:max_t}
|t|^2 \leq \min \{ ad, (1 - T^{(1)}_{11} - a)(T^{(1)}_{11} - d)\}.
\end{equation}
To maximize the right-hand side under the constraint \eqref{eq:phase_positive_constraint}, by the monotonicity, the maximization will always be achieved when
\begin{equation}
    ad = (1-T^{(1)}_{11} - a)(T^{(1)}_{11} - d).
\end{equation}
With $T^{(1)}_{11} \leq \Tr[T^{(1)}] \leq 1$, this results in
\begin{equation}
    \min \{ ad, (1 - T^{(1)}_{11} - a)(T^{(1)}_{11} - d)\} = T^{(1)}_{11}a - \frac{T^{(1)}_{11}}{1-T^{(1)}_{11}} a^2 \leq \frac{T^{(1)}_{11} (1-T^{(1)}_{11})}{4} \leq \frac{1}{16} .
\end{equation}
Hence, we can obtain
\begin{equation}
    \bar{c}_{\text{phase}}(E,2) = 2 - 4 \max \Re[t] \geq 2 - 4 \max |t| \geq 2 - 4\sqrt{1/16} = 1 .
\end{equation}

As for the specific solution to attain this lower bound when $E\geq 3/4$, choose 
\begin{equation}
T_0 = \begin{pmatrix}
    \frac{1}{4} & & & \frac{1}{4} \\
    & 0 & & \\
    & & \frac{1}{2} & \\
    \frac{1}{4} & & & \frac{1}{4}
\end{pmatrix}, \,
T_1 = \begin{pmatrix}
    \frac{1}{4} & & & -\frac{1}{4} \\
    & \frac{1}{2} & & \\
    & & 0 & \\
    -\frac{1}{4} & & & \frac{1}{4}
\end{pmatrix}.
\end{equation}
For this construction, we have $t := T_0^{1010} = 1/4$, which leads to $\bar{c}_{\text{phase}}(E,2) = 2 - 4 \max \Re[t] = 1$. Additionally, we can obtain $T^{(1)} = I/2$ by \eqref{eq:sum_Tx}. This leads to $\Tr[HT^{(1)}] = 1/2 \leq 3/4$ and
\begin{equation}
\Tr_{1,3}[(H \otimes I \otimes I + I \otimes I \otimes H) T] = \frac{1}{2}H^T + \frac{3}{4}I \preceq H^T+ \frac{3}{4}I
\end{equation}
by \eqref{eq:T}. Therefore, the minimal $E$ to satisfy the energy constraints is $3/4$. This aligns with the quantum comb of the circuit in Fig.~\ref{fig:phase_strategy_circuit} when $E=3/4$.
    
\end{proof}

\begin{lemma}
    We have
    \[
       \bar{c}_{\text{phase}}(E,2) = \begin{cases}
            2 - 4\kappa (1-\kappa), & E \leq 3/4, \\
            1, & E > 3/4,
        \end{cases}
    \]
with $\kappa = 1-\sqrt{1-E}$.
\end{lemma}
\begin{proof}
    The part for $E>3/4$ is just the content of Lem.~\ref{lem:phase_bound_general}. Based on the proof of Lem.~\ref{lem:phase_bound_general}, we further consider the energy constraints for $E \leq 3/4$. Similarly, \eqref{eq:phase_positive_constraint} and \eqref{eq:max_t} still hold. Our goal is still to maximize the right hand side of \eqref{eq:max_t}.

    With $H=\ket{1}\bra{1}$, by \eqref{eq:sum_Tx} and \eqref{eq:T}, the left hand side of \eqref{eq:energy_T} can be simplified as
    \begin{equation}
        \Tr_1 [(H \otimes I)T^{(1)}] + \Tr_1[T_1] = T^{(1)}_{11} I + \Tr_1[T^{(1)} \otimes I - T_0] = (T^{(1)}_{11} + 1) I - \Tr_1[T_0].
    \end{equation}
Therefore, \eqref{eq:energy_T} is equivalent to
\begin{equation}
    \Tr_1[T_0] \succeq (T^{(1)}_{11} - E + 1)I - H^T.
\end{equation}
In this way, the diagonal elements on the two sides should satisfy
\begin{equation}
\begin{cases}
    a + c \geq T^{(1)}_{11} - E + 1, \\
    b + d \geq T^{(1)}_{11} - E .
\end{cases}
\end{equation}
By \eqref{eq:phase_positive_constraint}, we have $c \leq T^{(1)}_{11}$, which then leads to
\begin{equation}
        a \geq T^{(1)}_{11} - c - E + 1 \geq 1 - E .
\end{equation}
Meanwhile, \eqref{eq:energy_rho} requires $T^{(1)}_{11} \leq E$. For summarization, together with \eqref{eq:phase_positive_constraint}, the following conditions should hold:
\begin{equation}
\begin{cases}
    0 \leq T^{(1)}_{11} \leq E, \\
    1 - E \leq a \leq 1 - T^{(1)}_{11}, \\
    0 \leq d \leq T^{(1)}_{11}.
\end{cases}
\end{equation}

To maximize $\min\{ ad, (1 - T^{(1)}_{11} - a)(T^{(1)}_{11} - d)\}$ (the right hand side of \eqref{eq:max_t}) within the range determined by the above conditions, the monotonicity can still be utilized to pick
\begin{equation}
ad = (1 - T^{(1)}_{11} - a)(T^{(1)}_{11} - d)
\end{equation}
to achieve the maximization. In other words, we have
\begin{equation}\label{eq:substitution_eq}
\min \{ ad, (1 - T^{(1)}_{11} - a)(T^{(1)}_{11} - d)\} = T^{(1)}_{11}a - \frac{T^{(1)}_{11}}{1-T^{(1)}_{11}} a^2.
\end{equation}
Fixing $T^{(1)}_{11}$, with the condition $1-E \leq a \leq 1-T^{(1)}_{11}$, the maximization is attained by
\begin{equation}
a_{\max} = \begin{cases}
    1 - E, &    \frac{1 - T^{(1)}_{11}}{2} \leq 1 - E, \\
    \frac{1 - T^{(1)}_{11}}{2}, &     \frac{1 - T^{(1)}_{11}}{2} > 1 - E.
\end{cases}
\end{equation}

When $\frac{1 - T^{(1)}_{11}}{2} \leq 1 - E$, i.e., $T^{(1)}_{11} \geq 2E - 1$, substituting the value of $a$ into \eqref{eq:substitution_eq} yields
\begin{equation}
T^{(1)}_{11}a_{\max} - \frac{T^{(1)}_{11}}{1-T^{(1)}_{11}} a_{\max}^2 = (1 - E) \frac{E T^{(1)}_{11} - (T^{(1)}_{11})^2}{1 - T^{(1)}_{11}}.
\end{equation}
By its derivative, we can obtain the maximization is attained by $T^{(1)}_{11} = 1 - \sqrt{1-E}$. Since we are considering the case of $E \leq 3/4$, it can be checked that $T^{(1)}_{11} = 1 - \sqrt{1-E} \geq 2E-1$ satisfies the requirement of this case. Thus, we have 
\begin{equation}\label{eq:min_bound1}
\min \{ ad, (1 - T^{(1)}_{11} - a)(T^{(1)}_{11} - d)\} \leq (1-E) (1-\sqrt{1-E})^2
\end{equation}
for $\frac{1 - T^{(1)}_{11}}{2} \leq 1 - E$.

When $\frac{1 - T^{(1)}_{11}}{2} > 1 - E$, i.e., $T^{(1)}_{11} < 2E - 1$, correspondingly, we have
\begin{equation}
T^{(1)}_{11}a_{\max} - \frac{T^{(1)}_{11}}{1-T^{(1)}_{11}} a_{\max}^2 = \frac{T^{(1)}_{11} (1-T^{(1)}_{11})}{4}.
\end{equation}
As $E \leq 3/4$, we can only choose $T^{(1)}_{11} < 2E - 1 \leq 1/2$. Therefore,
\begin{equation}\label{eq:min_bound2}
\min \{ ad, (1 - T^{(1)}_{11} - a)(T^{(1)}_{11} - d)\} \leq \frac{(1-E)(2E-1)}{4}
\end{equation}
for $\frac{1 - T^{(1)}_{11}}{2} > 1 - E$.

To summarize the results of \eqref{eq:min_bound1} and \eqref{eq:min_bound2}, we notice that 
\begin{equation}
    \frac{(1-E)(2E-1)}{4} \leq (1-E) (1-\sqrt{1-E})^2
\end{equation}
always holds for $E \leq 1$. This means that when the maximization also takes over $T^{(1)}$, we have
\begin{equation}
|t|^2 \leq \min \{ ad, (1 - T^{(1)}_{11} - a)(T^{(1)}_{11} - d)\} \leq (1-E) (1-\sqrt{1-E})^2
\end{equation}
in general. Hence, for $E \leq 3/4$, with $\bar{c}_{\text{phase}}(E,2) = \min (2 - 4\Re[t])$, we can obtain
\begin{equation}
    \bar{c}_{\text{phase}}(E,2) \geq \min (2 - 4|t|) \geq 2 - 4\sqrt{1-E} (1-\sqrt{1-E})=2 - 4\kappa (1-\kappa)
\end{equation}
with $\kappa := 1 - \sqrt{1-E}$.

To attain this lower bound, the optimal solution can be constructed as
\begin{equation}
T_0 = \begin{pmatrix}
    (1-\kappa)^2 &&& \kappa(1-\kappa) \\
    & 0 && \\
    && \kappa & \\
    \kappa(1-\kappa) &&& \kappa^2
\end{pmatrix}, \, 
T_1 = \begin{pmatrix}
    \kappa(1-\kappa) &&& -\kappa(1-\kappa) \\
    & 1-\kappa && \\
    && 0 & \\
    -\kappa(1-\kappa) &&& \kappa(1-\kappa)
\end{pmatrix}.
\end{equation}
For this construction, we have $t := T_0^{1010} = \kappa(1-\kappa)$. Thus, with $\bar{c}_{\text{phase}}(E,2) = 2 - 4 \max \Re[t] = 2 - 4\kappa (1-\kappa)$, the lower bound is attained. Additionally, we can obtain $T^{(1)} = \begin{pmatrix}
    1-\kappa & \\
    & \kappa
\end{pmatrix}$ by \eqref{eq:sum_Tx}. It follows that $\Tr[HT^{(1)}] = \kappa \leq E$ and
\begin{equation}
\Tr_{1,3}[(H \otimes I \otimes I + I \otimes I \otimes H) T] = (1-\kappa)H^T+ (2\kappa-\kappa^2)I = \sqrt{1-E}H^T + EI \preceq H^T + EI
\end{equation}
by \eqref{eq:T}. Therefore, this solution satisfies the energy constraints \eqref{eq:energy_rho} and \eqref{eq:energy_T}.

\end{proof}

\section{Phase Estimation Precision Scaling with Energy Constraints for Large Dimensions} \label{sec:proof_high_dimension}

The result of this section can be formally stated as follows:

\begin{theorem}[Precision scaling with respect to energy for infinite dimension]\label{thm:high_dimension}
    For any $E \geq 0$, we have
    \[
        \bar{c}_{\text{phase}}(E,\infty) \geq \frac{E+1}{\left(\frac{3E+3}{2a}\right)^3 + \frac{3E+3}{2}} \sim \frac{8a^3}{27} \frac{1}{E^2}
    \]
    with a constant $a = -a_1/2^{1/3} \approx 1.85576$, where $a_1$ is the first negative root of the Airy function $\operatorname{Ai}(x)$. Moreover, the scaling of $1/E^2$ is attainable.
\end{theorem}

To prove the theorem, we need a lemma for the finite dimension case:
\begin{lemma}[Lower bound on the average cost for finite dimension]\label{lem:cost_lower_bound}
    For any $E \geq 0$ and a generic dimension $d \geq 2$, we have 
    $\bar{c}_{\text{phase}}(E,d) \geq 2 - \min_{\gamma \geq 0}  [E\gamma - \lambda_{\min}( H_d\gamma + 2A_d)]$,
    where $H_d = \sum_{n=0}^{d-1} n \ket{n}\bra{n}$, $A_d = \sum_{n=0}^{d-2} (\ket{n}\bra{n+1} + \ket{n+1}\bra{n})/2$, and  $\lambda_{\min}(\cdot)$ denotes the smallest eigenvalue.  
\end{lemma}

Its proof is in Section.~\ref{sec:proof_lower_bound}, which is derived by considering the dual problem of the cost optimization. With this lemma, we give the proof of Thm.~\ref{thm:high_dimension}, which is divided into two parts. The first part (Section.~\ref{sec:proof_lower_bound_inf_dim}) is to prove the lower bound for the average cost in the statement. The second part (Section.~\ref{sec:proof_attainability}) is to show the attainability of the scaling $1/E^2$.

\subsection{Proof of Lemma~\ref{lem:cost_lower_bound}} \label{sec:proof_lower_bound}
Recall that the problem is to minimize the average cost
\begin{equation}
\bar{c}(T, \hat{\theta}) = 2 - 2 \sum_{x=0}^{d-1} \Re[\ee^{\ii\hat{\theta}(x)}\Tr[T_x \JJ]]
\end{equation}
for $T \in \Comb^{\leq E}_{\text{gl}}$. And, $\bar{c}_{\text{phase}}(E,d)$ represents the optimal value of this minimization. For convenience, it is equivalent to consider maximize the second term $\Re[\ee^{\ii\hat{\theta}(x)}\Tr[T_x \JJ]]$ in the average cost. 

Fixing the estimator $\hat{\theta}$, let  
\begin{equation} \label{eq:program_norm}
	\begin{aligned}
		v_{\text{primal}}(E,\hat{\theta}) = \max_{T} \quad & \sum_x \Re[\ee^{\ii\hat{\theta}(x)} \Tr[T_x \JJ]] \\
		\mathrm{s.t.} \quad	& T = \sum_x T_x \otimes \ket{x}\bra{x}, \\
        & T \in \Comb^{\leq E}_{\text{gl}}.
	\end{aligned}
\end{equation}
The optimal average cost can be written as
\begin{equation} \label{eq:cost_to_primal}
    \bar{c}_{\text{phase}}(E,d) = 2 - 2 \max_{\hat{\theta}} v_{\text{primal}}(E,\hat{\theta}).
\end{equation}

Let 
\begin{equation}\label{eq:tilde_J}
    \hat{\JJ}(\hat{\theta}) := \sum_x \frac{\ee^{\ii\hat{\theta}(x)} \JJ + \ee^{-\ii\hat{\theta}(x)} \JJ^\dagger}{2} \otimes \ket{x}\bra{x}.
\end{equation}
By the characterization of quantum combs
\begin{equation}
    \begin{cases}
        \Tr_{3}[T] = T^{(1)} \otimes I_{2}, \\
        \Tr[T^{(1)}] = 1,
    \end{cases}
\end{equation}
and Thm.~\ref{thm:energy_constrained_comb}, 
as $\HH_0 = \emptyset$ and $H_0=0$, the constraints in \eqref{eq:program_norm} can be explicitly formulated as
\begin{equation}
	\begin{aligned}
		v_{\text{primal}}(E,\hat{\theta}) = \max_{T} \quad & \sum_x \Tr[T \tilde{\JJ}(\hat{\theta})] \\
		\mathrm{s.t.} \quad	& \Tr_{3}[T] = T^{(1)} \otimes I_2, \\
        & \Tr[T^{(1)}] = 1, \\
        & \Tr[H T^{(1)}] \leq E, \\
        & \Tr_{1,3}[(H_{1,3} \otimes I_{0,2}) T] \preceq (H_{0,2})^T + EI_{0,2}.
	\end{aligned}
\end{equation}
With variables
\begin{gather}
    Y^{(1)}_c \in \operatorname{Herm}(\HH_0) = \mathbb{R},\, Y_c^{(2)} \in \operatorname{Herm}\left( \HH_{0} \otimes \HH_{1} \otimes \HH_{2} \right) = \operatorname{Herm}\left(\HH_{1} \otimes \HH_{2} \right), \\
    Y^{(1)}_e \in \operatorname{Herm}(\HH_0) = \mathbb{R},\, Y_e^{(2)} \in \operatorname{Herm}\left( \HH_0 \otimes \HH_2 \right) = \operatorname{Herm}(\HH_2),    
\end{gather}
where $\operatorname{Herm}(\HH)$ means the set of Hermitian operators on $\HH$, its dual problem can be written as
\begin{equation} \label{eq:dual}
    \begin{aligned}
        v_{\text{dual}}(E,\hat{\theta}) = \min_{Y_c^{(n)}, Y_e^{(n)}} \quad & Y_c^{(1)} + EY_e^{(1)} + \Tr\left[ Y_e^{(2)} \left( H^T + E I \right) \right] \\
        \text{s.t.} \quad & Y_c^{(1)} I - \Tr_{2}[Y_c^{(2)}] + Y_e^{(1)} H \succeq 0 \\
        & Y_c^{(2)} \otimes I_{3} + Y_e^{(2)} \otimes H_{1,3} \succeq \hat{\JJ}(\hat{\theta}) \\
        & Y_e^{(n)} \succeq 0, \, n=1,2
    \end{aligned} 
\end{equation}
By applying the duality to \eqref{eq:cost_to_primal}, we have
\begin{equation} \label{eq:cost_to_dual}
    \bar{c}_{\text{phase}}(E,d) = 2 - 2 \max_{\hat{\theta}} v_{\text{primal}}(E,\hat{\theta}) \geq 2 - 2 \max_{\hat{\theta}} v_{\text{dual}}(E,\hat{\theta}).
\end{equation}
This enables us to derive a lower bound on the average cost by considering the dual problem.

The dual problem \eqref{eq:dual} can be converted to
\begin{equation} \label{eq:dual_channel}
        \begin{aligned}
        v_{\text{dual}}(E,\hat{\theta}) = \min_{\alpha,\beta,V,Z} \quad & \alpha E + \beta + \alpha \Tr[H Z] \\
        \text{s.t.} \quad & V \succeq \frac{\ee^{\ii\hat{\theta}(x)} \JJ + \ee^{-\ii\hat{\theta}(x)} \JJ^\dagger}{2} - \alpha E_x I \otimes Z,\, \forall x, \\
        & \alpha H + \beta I \succeq \Tr_{2}[V],  \\
        & \Tr[Z] \leq 1, \\
        & Z \succeq  0, \alpha \geq 0,
    \end{aligned}
    \end{equation}
with $\alpha,\beta \in \mathbb{R}$ and $Z \in \operatorname{Herm}(\HH_2)$, as shown in the following lemma. Here, $E_x$ are energy levels of the Hamiltonian $H = \sum_x E_x \ket{x}\bra{x}$. With our setting for the phase estimation problem, we have $E_x = x$.

\begin{lemma}
    The dual problem \eqref{eq:dual} is equivalent to \eqref{eq:dual_channel}.
\end{lemma}
\begin{proof}
    For any fixed $Y_e^{(n)}$, we find a corresponding optimal solution $Y_c^{(n)}$. 
    
    Since $Y_c^{(1)} \in \mathbb{R}$, with the first constraint in \eqref{eq:dual}, the optimal choice of $Y_c^{(1)}$ is
\begin{equation}\label{eq:Yc1}
Y_c^{(1)} = \lambda_{\max} \left(  \Tr_{2}[Y_c^{(2)}] - Y_e^{(1)} H \right),
\end{equation}
where $\lambda_{\max}(\cdot)$ refers to the maximal eigenvalue. 

As for $Y_c^{(2)}$, by \eqref{eq:tilde_J} and
\begin{equation}
H_{1,3} = H \otimes I + I \otimes H = \sum_x (H + E_x I) \otimes \ket{x}\bra{x},
\end{equation}
the second constraint in \eqref{eq:dual}
\begin{equation}
Y_c^{(2)} \otimes I_{3} + Y_e^{(2)} \otimes H_{1,3} \succeq \hat{\JJ}(\hat{\theta})
\end{equation}
is equivalent to
\begin{equation}
    Y_c^{(2)} \succeq \frac{\ee^{\ii\hat{\theta}(x)} \JJ + \ee^{-\ii\hat{\theta}(x)} \JJ^\dagger}{2} - (H + E_x I) \otimes Y_e^{(2)}, \; \forall x. 
\end{equation}
Let
\begin{equation}
V = Y_c^{(2)} + H \otimes Y_e^{(2)}.
\end{equation}
The constraint for $V$ is
\begin{equation}\label{eq:V_constraint}
V \succeq \frac{\ee^{\ii\hat{\theta}(x)} \JJ + \ee^{-\ii\hat{\theta}(x)} \JJ^\dagger}{2} - E_x I \otimes Y_e^{(2)},\, \forall x.
\end{equation}
By \eqref{eq:Yc1}, the value of $Y_c^{(1)}$ is also affected by $V$, which is determined by the following relation:
\begin{equation}
\begin{aligned}
    Y_c^{(1)} &= \lambda_{\max} \left(  \Tr_{2}[Y_c^{(2)}] - Y_e^{(1)} H \right) \\
    &= \lambda_{\max} \left(  \Tr_{2}[V] - \left(Y_e^{(1)}  + \Tr[Y_e^{(2)}]  \right) H \right) .
\end{aligned} 
\end{equation}
Suppose $\Tr_{2}[V] = V_0 +  \left(Y_e^{(1)} + \Tr[Y_e^{(2)}] \right) H $. Then $Y_c^{(1)} = \lambda_{\max}(V_0)$. This means that for a fixed $Y_c^{(1)}$, to adjust $V$, we can only adjust $V_0$ while keeping its largest eigenvalue unchanged. Therefore, to satisfy \eqref{eq:V_constraint}, $V_0 = \lambda_{\max}(V_0)I$ will always be the optimal choice, as it is the largest one with the Loewner order for a given largest eigenvalue. Thus, we can assume $\Tr_{2}[V] = \alpha H + \beta I$ for $\alpha = Y_e^{(1)} + \Tr[Y_e^{(2)}] \geq 0$ and $\beta = \lambda_{\max}(V_0) = Y_c^{(1)} \in \R$.
Therefore,
\begin{equation}
\begin{aligned}
    v_{\text{dual}}(E) 
    &= \min \left( Y_c^{(1)} + EY_e^{(1)} + \Tr\left[ Y_e^{(2)} \left( H^T + E I \right) \right] \right)\\
    &= \min \left( \beta + E \alpha + \Tr\left[Y_e^{(2)} H^T \right] \right) .
\end{aligned}
\end{equation}
Notice that $\Tr[Y_e^{(2)}]=0$ if $\alpha=0$ as $\Tr[Y_e^{(1)}],\Tr[Y_e^{(2)}] \geq 0$. Let 
\begin{equation}
Z = \begin{cases}
    (Y_e^{(2)})^T/\alpha , & \alpha >0,\\
    0, & \alpha = 0.
\end{cases}
\end{equation}
We have $\Tr\left[Y_e^{(2)} H^T \right] = \alpha\Tr[HZ]$. We then obtain the form of \eqref{eq:dual_channel}.
\end{proof}

Based on \eqref{eq:cost_to_dual} and \eqref{eq:dual_channel}, we construct a solution to derive a lower bound for the optimal average cost. With $\bra{0}H\ket{0} = 0$, setting $Z = \ket{0}\bra{0}$, we have
\begin{equation}
    v_{\text{dual}}(E,\hat{\theta}) \leq \begin{aligned}[t]
        \min_{\alpha,\beta,V} \quad & \alpha E + \beta \\
        \text{s.t.} \quad & V \succeq \frac{\ee^{\ii\hat{\theta}(x)} \JJ + \ee^{-\ii\hat{\theta}(x)} \JJ^\dagger}{2} - \alpha E_x I \otimes \ket{0}\bra{0},\, \forall x, \\
        & \alpha H + \beta I \succeq \Tr_{2}[V], \\
        & \alpha \geq 0.
    \end{aligned}
\end{equation}
For any fixed $\alpha,V$, the optimal value of $\beta$ will always be $\beta_{\min} = \lambda_{\max}(\Tr_{2}[V] - \alpha H)$. Thus,
\begin{equation} \label{eq:dual_bound}
    v_{\text{dual}}(E,\hat{\theta}) \leq \begin{aligned}[t]
        \min_{\alpha,V} \quad & \alpha (E + E_g) + \lambda_{\max}(\Tr_{2}[V] - \alpha H) \\
        \text{s.t.} \quad & V \succeq \frac{\ee^{\ii\hat{\theta}(x)} \JJ + \ee^{-\ii\hat{\theta}(x)} \JJ^\dagger}{2} - \alpha E_x I \otimes \ket{0}\bra{0},\, \forall x,  \\
        & \alpha \geq 0.
    \end{aligned}
\end{equation}
Here, the feasible set of $(\alpha,V)$ is independent of $E$. However, the optimal solution could depend on the energy value $E$. 

Replace the notations in \eqref{eq:dual_bound} by the specific values for the phase channel. The Choi operator of the phase channel is
\begin{equation}
C_{d,\theta} = \sum_{j,k=0}^{d-1} \ee^{\ii(j-k)\theta} \ket{j}\bra{k} \otimes \ket{j}\bra{k}.
\end{equation}
The task operator is
\begin{equation}
\JJ_d := \frac{1}{2\pi} \int_{\theta \in [0,2\pi]} \ee^{-\ii\theta} C_{d,\theta}^T \dd\theta =  \sum_{j=1}^{d-1} \ket{j-1,j-1}\bra{j,j}.
\end{equation}
The Hamiltonian is $H_d = \sum_{n=0}^{d-1} n\ket{n}\bra{n}$. The upper bound in \eqref{eq:dual_bound} becomes
 \begin{equation} \label{eq:dual_bound_simp}
    v_{\text{dual}}(E,\hat{\theta}) \leq \begin{aligned}[t]
        \min_{\alpha,V_d} \quad & \alpha E + \lambda_{\max}(\Tr_{2}[V_d] - \alpha H_d) \\
        \text{s.t.} \quad & V_d \succeq \frac{\ee^{\ii\hat{\theta}(n)} \JJ_d + \ee^{-\ii\hat{\theta}(n)} \JJ^\dagger_d}{2} - \alpha n I \otimes \ket{0}\bra{0},\, \forall n=0,\dots,d-1,  \\
        & \alpha \geq 0.
    \end{aligned}
\end{equation}

Since $I \otimes \ket{0}\bra{0} \succeq \ket{0}\bra{0} \otimes \ket{0}\bra{0}$, a solution $V_d = \sum_{j,k} V_d^{jk} \ket{j,j}\bra{k,k}$ satisfying
\begin{equation}
V_d \succeq \frac{\ee^{\ii\hat{\theta}(n)} \JJ_d + \ee^{-\ii\hat{\theta}(n)} \JJ^\dagger_d}{2} - \alpha n \ket{0}\bra{0} \otimes \ket{0}\bra{0},\, \forall n=0,\dots,d-1
\end{equation}
will also be a feasible solution to \eqref{eq:dual_bound_simp}. Note that $\ket{j} \mapsto \ket{j,j}$ is an embedding into the tensor product space, this condition is equivalent to
\begin{equation} \label{eq:tilde_V_bound}
\begin{aligned}
    \tilde{V}_d \succeq& \frac{\ee^{\ii\hat{\theta}(n)}\tilde{\JJ}_d + \ee^{-\ii\hat{\theta}(n)} \tilde{\JJ}^\dagger_d}{2} - \alpha n \ket{0}\bra{0} ,\, \forall n=0,\dots,d-1,
\end{aligned}
\end{equation}
where $\tilde{A}$ means the image of the inverse of the embedding map, i.e., 
\begin{equation}
\tilde{V}_d = \sum_{j,k=0}^{d-1} V_d^{jk} \ket{j}\bra{k},\, \tilde{\JJ}_d = \sum_{j=1}^{d-1} \ket{j-1}\bra{j}.
\end{equation}
In the matrix form, the right hand side of \eqref{eq:tilde_V_bound} is
\begin{equation}
\begin{aligned}
\begin{pmatrix}
    -\alpha n & \ee^{\ii\hat{\theta}(n)}/2 &  &  & \\
    \ee^{-\ii\hat{\theta}(n)}/2 & 0 & \ee^{\ii\hat{\theta}(n)}/2 &  & \\
     & \ee^{-\ii\hat{\theta}(n)}/2 & 0  & \ddots & \\
     & & \ddots & \ddots & \ee^{\ii\hat{\theta}(n)}/2 \\
     & & & \ee^{-\ii\hat{\theta}(n)}/2 & 0
\end{pmatrix} 
\in \C^{d \times d} .
\end{aligned}
\end{equation}
For the left hand side of \eqref{eq:tilde_V_bound}, based on our choice $V_d = \sum_{j,k} V_d^{jk} \ket{j,j}\bra{k,k}$, we further make an assumption $\Tr_{2}[V_d] = \sum_j V_d^{jj} \ket{j}\bra{j} = \alpha H_d + \beta I_d$ in \eqref{eq:dual_bound_simp}. Since $H_d = \sum_n n \ket{n}\bra{n}$, we have
\begin{equation}
V_d^{jj} = \alpha  j + \beta.
\end{equation}
Assume $\tilde{V}_d$ is also tridiagonal:
\[
\tilde{V}_d = \sum_{j,k=0}^{d-1} V_d^{jk} \ket{j}\bra{k} = 
\begin{pmatrix}
    V_d^{00} & V_d^{01} &  &  & \\
    V_d^{10} & V_d^{11} & V_d^{12} &  & \\
     & V_d^{21} & V_d^{22}  & \ddots & \\
     & & \ddots & \ddots & V_d^{d-2,d-1} \\
     & & & V_d^{d-1,d-2} & V_d^{d-1,d-1}
\end{pmatrix} .
\]
Then \eqref{eq:tilde_V_bound} becomes
\begin{equation}\label{eq:matrix_form_constraint}
\begin{pmatrix}
    n\alpha+\beta & V_d^{01}-\ee^{\ii\hat{\theta}(n)}/2 &  &  & \\
    V_d^{10} -\ee^{-\ii\hat{\theta}(n)}/2 & \alpha+\beta & \ddots &  & \\
     & \ddots & 2\alpha+\beta  & \ddots & \\
     & & \ddots & \ddots & V_d^{d-2,d-1}-\ee^{\ii\hat{\theta}(n)}/2 \\
     & & & V_d^{d-1,d-2} -\ee^{-\ii\hat{\theta}(n)}/2 & (d-1)\alpha+\beta
\end{pmatrix}  \succeq 0.
\end{equation}
If we choose $V_d^{j,j-1} = V^{j-1,j}_d = 0$, the norm 
\begin{equation}
\left|V_d^{j,j-1} - \frac{\ee^{-\ii\hat{\theta}(n)}}{2}\right| = \left|V_d^{j-1,j} - \frac{\ee^{\ii\hat{\theta}(n)}}{2}\right| = 1/2.
\end{equation}
Notice that the positivity of the left hand side of \eqref{eq:matrix_form_constraint} is independent of the phase of the terms $V_d^{j,j-1} - \ee^{-\ii\hat{\theta}(n)} / 2$ and $V_d^{j-1,j} - \ee^{\ii\hat{\theta}(n)} / 2$. It suffices to consider $\tilde{V}_d$ such that
\begin{equation}
\begin{pmatrix}
    n\alpha+\beta & -1/2 &  &  & \\
    -1/2 & \alpha+\beta & -1/2 &  & \\
     & -1/2 & 2\alpha+\beta  & \ddots & \\
     & & \ddots & \ddots & -1/2 \\
     & & & -1/2 & (d-1)\alpha+\beta
\end{pmatrix}  \succeq 0
\end{equation}
for $n=0,\dots,d-1$. In this way, by constructing the feasible solutions to \eqref{eq:dual_bound_simp}, we obtain another upper bound for $v_{\text{dual}}(E,\hat{\theta})$:
\begin{equation}
v_{\text{dual}}(E,\hat{\theta}) \leq \begin{aligned}[t]
        \min_{\alpha,\beta \geq 0} \quad & \alpha E + \beta \\
        \text{s.t.} \quad & \begin{pmatrix}
    \beta & -1/2 &  &  & \\
    -1/2 & \alpha+\beta & -1/2 &  & \\
     & -1/2 & 2\alpha+\beta  & \ddots & \\
     & & \ddots & \ddots & -1/2 \\
     & & & -1/2 & (d-1)\alpha+\beta
\end{pmatrix}  \succeq 0.
    \end{aligned}
\end{equation}
Let 
\begin{equation}
A_d = \begin{pmatrix}
    0 & 1/2 &  &  & \\
    1/2 & 0 & 1/2 &  & \\
     & 1/2 & 0  & \ddots & \\
     & & \ddots & \ddots & 1/2 \\
     & & & 1/2 & 0
\end{pmatrix} \in \C^{d \times d}.
\end{equation}
We can rewrite the upper bound as
\begin{equation}
v_{\text{dual}}(E,\hat{\theta}) \leq \min_{\alpha \geq 0} \lambda_{\max}((EI_d - H_d)\alpha - A_d).
\end{equation}
By \eqref{eq:cost_to_dual}, the average cost is then lower bounded by
\begin{equation}
\bar{c}_{\text{phase}}(E,d) \geq 2 - 2 \max_{\hat{\theta}} v_{\text{dual}}(E,\hat{\theta}) \geq 2 - 2 \min_{\alpha \geq 0} \lambda_{\max}((EI_d - H_d)\alpha - A_d).
\end{equation}
For convenience, we rewrite the bound by
\begin{equation}
    2 \min_{\alpha \geq 0} \lambda_{\max}((EI_d - H_d)\alpha - A_d) = \min_{\gamma \geq 0}  [E\gamma - \lambda_{\min}( H_d\gamma + 2A_d)]
\end{equation}
with $\gamma := 2 \alpha$.

\subsection{Lower bound for the average cost of the infinite-dimensional phase channel} \label{sec:proof_lower_bound_inf_dim}
To analyze the behavior of the lower bound 
\begin{equation}\label{eq:finite_d_lower_bound}
    \bar{c}_{\text{phase}}(E,d) \geq 2 - \min_{\gamma \geq 0}  [E\gamma - \lambda_{\min}( H_d\gamma + 2A_d)]
\end{equation}
for the average cost given in Lem.~\ref{lem:cost_lower_bound} with a high dimension for
\begin{equation}
A_d = \begin{pmatrix}
    0 & 1/2 &  &  & \\
    1/2 & 0 & 1/2 &  & \\
     & 1/2 & 0  & \ddots & \\
     & & \ddots & \ddots & 1/2 \\
     & & & 1/2 & 0
\end{pmatrix} \in \C^{d \times d},
\end{equation}
we determine the value of $\lambda_{\min}( H_d\gamma + 2A_d)$ when the dimension $d$ tends to infinity.

First, we claim that $\lambda_{\min}( H_d\gamma + 2A_d)$ will converge to a value $\hat{\lambda}_{\min}^\gamma$ in the range of $[-2,0)$. 

\begin{lemma}\label{lem:eigenvalue_convergence}
    For any given parameter $\gamma \geq 0$, as the dimension $d\rightarrow \infty$, the smallest eigenvalue $\lambda_{\min}( H_d\gamma + 2A_d) \in [-2,0] $ converges to a value $\hat{\lambda}_{\min}^\gamma \in [-2,0)$ dependent on $\gamma$.
\end{lemma}
\begin{proof}
    When $\gamma = 0$, the eigenvalues of $2A_d$ are
    \begin{equation}
    2\cos\left(\frac{j\pi}{d+1}\right), \; j =1,\dots,d.
    \end{equation}
    Therefore, 
    \begin{equation}
    \lambda_{\min}(2A_d) = 2\cos\left(\frac{d\pi}{d+1}\right) \longrightarrow \hat{\lambda}_{\min}^0 = -2.
    \end{equation}

    When $\gamma > 0$, with $H_d \succeq 0$, we have 
    \begin{equation}
        \lambda_{\min}(H_d \gamma + 2A_d) \geq \lambda_{\min}(2A_d) > -2.
    \end{equation}
    Meanwhile, by Cauchy interlacing theorem, the value of $\lambda_{\min}(H_d \gamma + 2A_d)$ is non-increasing with $d$. Thus, it will converge to $\hat{\lambda}_{\min}^\gamma \geq -2$. Furthermore, since $\lambda_{\min}(H_d \gamma + 2A_d)$ is negative for $d=2$, the minimal eigenvalue will always be negative, and then $\hat{\lambda}_{\min}^\gamma < 0$.
\end{proof}

With the guarantee of the convergence, it is possible to determine the limitation of the minimum cost. Taking a limit of the dimension $d\rightarrow \infty$ in \eqref{eq:finite_d_lower_bound}, it turns out to be an optimization problem:
\begin{equation}
    2 - \min_{\gamma \geq 0} (E\gamma - \hat{\lambda}_{\min}^\gamma).
\end{equation}
To simplify the problem, we can only consider a subset for $\gamma$ without loss of generality, which is
\begin{equation}
\mathbf{S} = \left\{\gamma > 0 \middle| \frac{\hat{\lambda}_{\min}^\gamma}{\gamma} \notin \Z \right\}.
\end{equation}
\begin{lemma}
    \label{lem:gamma_not_int}
    The closure $\bar{\mathbf{{S}}}$ of the set $\mathbf{S}$ is $[0,+\infty)$, so we can assume $\hat{\lambda}_{\min}^\gamma/\gamma \notin \Z$ without loss of generality.
\end{lemma}
\begin{proof}
    It suffices to show $(0,+\infty) \subset \bar{\mathbf{{S}}}$. Let
    \begin{equation}
    \mathbf{S}_k = \left\{\gamma > 0 \middle| \frac{\hat{\lambda}_{\min}^\gamma}{\gamma} = k \right\}
    \end{equation}
for $k \in \Z$. By Lem.~\ref{lem:eigenvalue_convergence}, it follows that $\hat{\lambda}_{\min}^{\gamma} < 0$, so that $\mathbf{S}_k = \emptyset$ for $k \geq 0$. Meanwhile, for $k < 0$, if the set $\mathbf{S}_k$ contains more than one point, then 
\begin{equation}
\hat{\lambda}_{\min}^\gamma = k \gamma, \, \gamma \in \mathbf{S}_k
\end{equation}
implies $\hat{\lambda}^\gamma_{\min}$ strictly decreases with $\gamma$ in $\mathbf{S}_k$. However, since $H_d \gamma + 2A_d \succeq H_d \gamma' + 2A_d$ for $\gamma \geq \gamma'$, the value of $\hat{\lambda}_{\min}^\gamma$ should be non-decreasing. Thus, the set $\mathbf{S}_k$ is either an empty set or a single-point set. If it is empty, there is no need for further consideration. If it contains one point, note that $(\bigcup_j \mathbf{S}_j) \bigcup \mathbf{S} = (0,+\infty)$, this means that the point in $\mathbf{S}_k$ should be in the closure 
\begin{equation}
\overline{(\bigcup_{j \neq k} \mathbf{S}_j) \bigcup \mathbf{S}} = \overline{(\bigcup_{j \neq k} \mathbf{S}_j)} \bigcup \bar{\mathbf{S}}.
\end{equation}
If this point is not in $\bar{\mathbf{S}}$, it will be on the boundary of $\bigcup_{j \neq k} \mathbf{S}_j$, which is impossible as $\hat{\lambda}^\gamma_{\min} / \gamma$ should be continuous for $\gamma$. Therefore, any point in $\mathbf{S}_k$ for an arbitrary $k \in \Z$ is also in the closure $\bar{\mathbf{S}}$.
\end{proof}

To obtain the eigenvalue of $H_d \gamma + 2A_d$, we consider its characteristic polynomial:
\begin{equation}
P_d (x) = |xI_d - (H_d\gamma + 2A_d)|= \det\begin{pmatrix}
    x & -1 &  &  & \\
    -1 & x - \gamma & -1 &  & \\
     & -1 & x - 2\gamma  & \ddots & \\
     & & \ddots & \ddots & -1 \\
     & & & -1 & x - (d-1)\gamma
\end{pmatrix}.
\end{equation}
We have
\begin{equation} \label{eq:recurrence_poly}
    P_d(x) = (x-(d-1)\gamma) P_{d-1}(x) - P_{d-2}(x),\, d \geq 2
\end{equation}
with $P_0(x) = 1, P_1(x) = x$. Based on Lem.~\ref{lem:gamma_not_int}, without loss of generality, we also assume $\gamma > 0$ and $x/\gamma \notin \Z$.

With our assumption, the recurrence relation in \eqref{eq:recurrence_poly} has a general solution
\begin{equation}
\begin{aligned}
    P_d(x) =& c_1(x) J_{x/\gamma - d}(2/\gamma) + c_2(x) Y_{x/\gamma -  d}(2/\gamma) \\
    =:& c_1(x) \hat{J}_d(x) + c_2(x) \hat{Y}_d(x),
\end{aligned}
\end{equation}
where
\begin{equation} \label{eq:bessel_J_def}
    J_{\nu}(z) = \sum_{k=0}^{\infty} \frac{(-1)^k}{k!\Gamma(k+\nu+1)}\left(\frac{z}{2}\right)^{2 k+\nu},
\end{equation} 
\begin{equation}\label{eq:bessel_Y_def}
    Y_\nu(z)=\frac{J_\nu(z) \cos (\nu \pi)-J_{-\nu}(z)}{\sin (\nu \pi)}
\end{equation}
are Bessel functions of the first kind and the second kind, respectively~\cite{watson1922treatise}. Both two kinds of Bessel functions satisfy the recurrence relation
\begin{equation}\label{eq:bessel_recurrence}
    \frac{2 \nu}{z} Z_\nu(z)=Z_{\nu-1}(z)+Z_{\nu+1}(z)
\end{equation}
with $Z_\nu(z)$ being either $J_\nu(z)$ or $Y_\nu(z)$. By the initial values, we can obtain
\begin{equation}
\begin{cases}
        c_1(x) = \frac{\hat{Y}_1(x) - x \hat{Y}_0(x)}{\hat{J}_0(x) \hat{Y}_1(x) - \hat{J}_1(x) \hat{Y}_0(x)}, \\
        c_2(x) = \frac{x\hat{J}_0(x) - \hat{J}_1(x)}{\hat{J}_0(x) \hat{Y}_1(x) - \hat{J}_1(x) \hat{Y}_0(x)}.
\end{cases}
\end{equation}
Thus, the minimal eigenvalue $\lambda_{\min}( H_d\gamma + 2A_d)$ can be characterized as the smallest real root of 
\begin{equation}
F_d(x) := (\hat{Y}_1(x) - x \hat{Y}_0(x)) \hat{J}_d(x) + (x \hat{J}_0(x) - \hat{J}_{1}(x)) \hat{Y}_{d}(x).
\end{equation}
This enables us to give a characterization of the limitation $\hat{\lambda}_{\min}^{\gamma}$ by regarding it as a root of some function.
\begin{lemma}
    \label{lem:root}
    For a fixed $\gamma > 0$, the limitation $\hat{\lambda}_{\min}^{\gamma} := \lim_{d\rightarrow \infty} \lambda_{\min}( H_d\gamma + 2A_d)$ is a root of $F^\gamma(x) = J_{-x/\gamma-1}(2/\gamma)$ in $[-2,0)$.
\end{lemma}
\begin{proof}
    First, by Lem.~\ref{lem:eigenvalue_convergence}, the smallest eigenvalues and the limitation $\hat{\lambda}_{\min}^{\gamma}$ are both in the interval $[-2,0]$. Thus, we only need to consider the value of the function $F_d(x)$ for $x \in [-2,0]$.
    
    By Lem.~\ref{lem:gamma_not_int}, we assume $\hat{\lambda}^\gamma_{\min}/\gamma \notin \Z$ without loss of generality. Suppose the distance between $\hat{\lambda}^\gamma_{\min}/\gamma$ and the nearest integer is a positive constant $\delta > 0$. Since we only concern the asymptotic behavior, it is unnecessary to consider the whole interval $[-2,0]$. It suffices to only consider
    \begin{equation}
    x \in \mathcal{I}_\delta := \left(\hat{\lambda}^\gamma_{\min} - \frac{\delta}{2}\gamma, \hat{\lambda}^\gamma_{\min} + \frac{\delta}{2}\gamma\right) \cap [-2,0].
    \end{equation}
For any $x \in \mathcal{I}_\delta$, the distance between $x/\gamma$ and the nearest integer is at least $\delta / 2 > 0$.
    
    For $x \in \mathcal{I}_\delta$, by the definition of Bessel function in \eqref{eq:bessel_Y_def}, we have
\begin{equation}
\begin{aligned}
    \hat{Y}_d(x) := Y_{x/\gamma - d}(2/\gamma) =& \frac{J_{x/\gamma - d}(2/\gamma) \cos((x/\gamma - d)\pi) - J_{-(x/\gamma - d)}(2/\gamma)}{\sin((x/\gamma - d)\pi)} \\
    =& \frac{(-1)^d J_{x/\gamma - d}(2/\gamma) \cos(\pi x/\gamma) - J_{-(x/\gamma - d)}(2/\gamma)}{(-1)^d \sin(\pi x/\gamma)}.
\end{aligned}
\end{equation}
For convenience, denote 
\begin{equation}
\lambda_{\min}^{\gamma,d} := \lambda_{\min}( H_d\gamma + 2A_d).
\end{equation}
As discussed before, we have $\lambda^{\gamma,d}_{\min}$ is a root of $F_d(x)$. By the above equation for $\hat{Y}_d(x)$, it will also be a root of
\begin{equation}\label{eq:root_function}
    \sin(\pi x/\gamma) F_d(x) = G_d(x) + (-1)^{d+1} (x \hat{J}_0(x) - \hat{J}_{1}(x))J_{-(x/\gamma - d)}(2/\gamma),
\end{equation}
where
\begin{equation}
G_d(x) := \left[ (\hat{Y}_1(x) - x \hat{Y}_0(x)) \sin(\pi x/\gamma) + (x \hat{J}_0(x) - \hat{J}_{1}(x)) \cos(\pi x/\gamma) \right] \hat{J}_{d}(x).
\end{equation}
We then consider the asymptotic behavior of the above functions when $d \rightarrow \infty$.

For the second term on the right hand side of \eqref{eq:root_function}, Note that the factor $x\hat{J}_0(x)-\hat{J}_1(x)$ is independent of $d$, so only the scaling of $(-1)^{d+1} J_{-(x/\gamma-d)}(2/\gamma)$ matters. By \eqref{eq:bessel_J_def}, we have
\begin{equation}
\begin{aligned}
    J_{-(x/\gamma-d)} (2/\gamma)
    =& \sum_{k=0}^{\infty} \frac{(-1)^k}{k! \Gamma(k + d - x/\gamma + 1)} \gamma^{x /\gamma - d - 2k} \\
    =& \frac{\gamma^{x/\gamma - d}}{\Gamma(d-x/\gamma+1)} \left[1 + \sum_{k=1}^{\infty} \frac{(-1)^k}{k! \prod_{l=1}^k(d - x/\gamma  + l)} \gamma^{-2k} \right].
\end{aligned}
\end{equation}
Thus, for $x \in \mathcal{I}_\delta$, when $d$ is sufficiently large, with $\gamma_- := \min\{\gamma,1\}$, it follows that
\begin{equation}
    \begin{aligned}
        |J_{-(x/\gamma-d)} (2/\gamma)| 
        \leq& \frac{\gamma_{-}^{-2/\gamma -d}}{\Gamma(d+1)} \left[1 + \sum_{k=1}^{\infty} \frac{1}{k!} \gamma^{-2k} \right] \\
        =& \frac{\gamma_{-}^{-2/\gamma -d}}{\Gamma(d+1)} \ee^{\gamma^{-2}}.
    \end{aligned}
\end{equation}
Then by Stirling's formula, we can obtain that the second term in \eqref{eq:root_function} will uniformly converge to zero for $x\in \mathcal{I}_\delta$ as $d \rightarrow \infty$.

In this way, we can only pay attention to the first term $G_d(x)$ in \eqref{eq:root_function}. Similarly, the only factor in $G_d(x)$ dependent on $d$ is 
\begin{equation}
\begin{aligned}
    \hat{J}_d(x) := J_{x/\gamma - d}(2/\gamma) 
    =& \sum_{k=0}^{\infty} \frac{(-1)^k}{k! \Gamma(k + x/\gamma - d + 1)} \gamma^{d - x /\gamma - 2k} \\
    =& \frac{\gamma^{d-x/\gamma}}{\Gamma(x/\gamma - d + 1)} \left[1 + \sum_{k=1}^{\infty} \frac{(-1)^k}{k! \prod_{l=1}^k(x/\gamma - d + l)} \gamma^{-2k} \right]. 
\end{aligned}
\end{equation}
Recall that the distance between $x/\gamma$ and the integers is at least $\delta_1 := \delta/2 > 0$ based on our assumption. Hence, for $x\in\mathcal{I}_\delta$, with $\gamma_+ := \max\{\gamma,1\}$, we have
\begin{equation}
\begin{aligned}
    |\hat{J}_d(x)| 
    \geq& \frac{\gamma_{+}^{d+2/\gamma}}{|\Gamma(x/\gamma - d + 1)|} \left[1 - \sum_{k=1}^{\infty} \frac{1}{k! \prod_{l=1}^k|x/\gamma - d + l|} \gamma^{-2k} \right] \\
    =& \frac{\gamma_{+}^{d+2/\gamma}}{|\Gamma(x/\gamma - d + 1)|} \left[1 - \frac{1}{|x/\gamma - d + 1|} \sum_{k=1}^{\infty} \frac{1}{k! \prod_{l=2}^k|x/\gamma - d + l|} \gamma^{-2k} \right] \\
    \geq& \frac{\gamma_{+}^{d+2/\gamma}}{|\Gamma(x/\gamma - d + 1)|} \left[1 - \frac{1}{|x/\gamma - d + 1|} \sum_{k=1}^{\infty} \frac{1}{k! \delta_1^{k-1}} \gamma^{-2k} \right] \\
    =& \frac{\gamma_{+}^{d+2/\gamma}}{|\Gamma(x/\gamma - d + 1)|} \left[1 - \frac{\delta_1}{|x/\gamma - d + 1|} (\ee^{\delta_1^{-1}\gamma^{-2}}-1) \right] .
\end{aligned}
\end{equation}
When $d$ is sufficiently large, it holds that $|x/\gamma-d+1| = d - x/\gamma - 1 \geq d - 1$ as $x \leq 0$, then
\begin{equation}
\begin{aligned}
    |\hat{J}_d(x)| 
    \geq \frac{\gamma_{+}^{d+2/\gamma}}{|\Gamma(x/\gamma - d + 1)|} \left[1 - \frac{\delta_1}{d- 1} (\ee^{\delta_1^{-1}\gamma^{-2}}-1) \right] .
\end{aligned}
\end{equation}
By the property $\Gamma(z)\Gamma(1-z) = \frac{\pi}{\sin(\pi z)}$ of Gamma function for $z \notin \Z$, noticing that it can be assumed that $x/\gamma \notin \mathbb{Z}$, we can obtain
\begin{equation}
\begin{aligned}
    |\hat{J}_d(x)| 
    \geq& \frac{|\sin(\pi x/\gamma)|}{\pi} \gamma_{+}^{d+2/\gamma}\Gamma(d - x/\gamma) \left[1 - \frac{\delta_1}{d- 1} (\ee^{\delta_1^{-1}\gamma^{-2}}-1) \right] \\
    \geq& \frac{|\sin(\pi \delta_1)|}{\pi} \gamma_{+}^{d+2/\gamma}\Gamma(d) \left[1 - \frac{\delta_1}{d- 1} (\ee^{\delta_1^{-1}\gamma^{-2}}-1) \right].
\end{aligned}
\end{equation}
By Stirling's formula, it follows that $|\hat{J}_d(x)|$ will uniformly tends to be infinite for $x\in \mathcal{I}_\delta$ as $d \rightarrow \infty$. 

Remind that the left hand side of \eqref{eq:root_function} is zero for $x = \lambda^{\gamma,d}_{\min}$.  By the uniform convergence of the second term in the right hand side, we have
\begin{equation}
G_d(\lambda_{\min}^{\gamma,d}) \rightarrow 0.
\end{equation}
However, since $|\hat{J}_d(x)|$ uniformly tends to infinity, to prevent $G_d(\lambda_{\min}^{\gamma,d})$ from being divergent, the only way is to require the constant factor in front of $\hat{J}_d(x)$ vanishes, that is
\begin{equation}
\lim_{d\rightarrow \infty} \left.\left[ (\hat{Y}_1(x) - x \hat{Y}_0(x)) \sin(\pi x/\gamma) + (x \hat{J}_0(x) - \hat{J}_{1}(x)) \cos(\pi x/\gamma) \right]\right|_{x=\lambda^{\gamma,d}_{\min}} = 0.
\end{equation}
Substituting $\hat{Y}_d(x) := Y_{x/\gamma - d}(2/\gamma)$ and $\hat{J}_d(x) := J_{x/\gamma - d}(2/\gamma)$, then by the definition of $Y_\nu(z)$ in \eqref{eq:bessel_Y_def}, this factor can be simplified as
\begin{equation}
\begin{aligned}
    & (\hat{Y}_1(x) - x \hat{Y}_0(x)) \sin(\pi x/\gamma) + (x \hat{J}_0(x) - \hat{J}_{1}(x)) \cos(\pi x/\gamma) \\
    =& (Y_{x/\gamma-1}(2/\gamma) - x Y_{x/\gamma}(2/\gamma)) \sin(\pi x/\gamma) + (x J_{x/\gamma}(2/\gamma) - J_{x/\gamma-1}(2/\gamma)) \cos(\pi x/\gamma) \\
    =& x J_{-x/\gamma}(2/\gamma) + J_{-x/\gamma+1}(2/\gamma).
\end{aligned}
\end{equation}
Then by \eqref{eq:bessel_recurrence}, we obtain
\begin{equation}
\begin{aligned}
     (\hat{Y}_1(x) - x \hat{Y}_0(x)) \sin(\pi x/\gamma) + (x \hat{J}_0(x) - \hat{J}_{1}(x)) \cos(\pi x/\gamma) 
    = -J_{-x/\gamma-1}(2/\gamma).
\end{aligned}
\end{equation}
By its continuity, the limitation $\hat{\lambda}_{\min}^{\gamma} = \lim _{d\rightarrow \infty} \lambda^{\gamma,d}_{\min}$ must be a root of it.
\end{proof}
If the smallest root of $F^\gamma(x) = J_{-x/\gamma -1}(2/\gamma)$ in $[-2,0)$ is $x_0$, by Lem.~\ref{lem:root}, we have $\hat{\lambda}^\gamma_{\min} \geq x_0$. Then, we can obtain
\begin{equation}
\min_\gamma (E\gamma - \hat{\lambda}^\gamma_{\min}) \leq \min_{\gamma} \max_{x_0} (E\gamma - x_0)
\end{equation}
with maximization over $x_0$ satisfying $F^\gamma(x_0) = 0$.  Let $\nu_0 = -x_0 /\gamma-1 > -1$.  Denote the positive roots of Bessel functions of the first kind $J_{\nu}(z)$ as
\begin{equation}
0 < j_{\nu,1} < j_{\nu,2} < \cdots < j_{\nu,k} < \cdots.
\end{equation}
Since $\gamma > 0$ and $J_{\nu_0}(2/\gamma)=0$, it should hold that $\gamma = 2/j_{\nu_0,k}$ and $x_0 = -2(\nu_0+1) / j_{\nu_0,k}$ for some $k$. Thus,
\begin{equation}
\min_{\gamma} \max_{x_0} (E\gamma - x_0) = \min_\gamma \max_{\nu_0} (E+\nu_0+1)\gamma
\end{equation}
with maximization over $\nu_0 > -1$ satisfying $\gamma = 2/j_{\nu_0,k}$ for some $k$. To attain the maximum with respect to $\nu_0$, the optimal choice of $\nu_0$ is the largest one among all the feasible values. Since the zero $j_{\nu,k}$ of Bessel function strictly increasing with respect to $\nu$ and $k$ for $\nu > -1$~\cite{ifantisDifferentialInequalityPositive1992}, for the optimal $\nu_0$, it follows that $\gamma = 2 / j_{\nu_0,1}$ and such an optimal $\nu_0$ is unique with a given $\gamma$. Moreover, in this way, we can construct a bijection between each $\gamma$ and the corresponding optimal $\nu_0$. Therefore,
\begin{equation}
\min_\gamma \max_{\nu_0} (E+\nu_0+1)\gamma = \min_{\nu > -1}  \frac{2(E+\nu+1)}{j_{\nu,1}}.
\end{equation}
We have~\cite{quBestPossibleUpper1999}
\begin{equation}
j_{\nu,1} \geq \begin{cases}
    \nu-\frac{a_1}{2^{1 / 3}} \nu^{1 / 3}, & \nu \geq 0, \\
    0, & -1 < \nu < 0,
\end{cases}
\end{equation}
where $a_1 \approx -2.33811$ is the first negative root of the Airy function $\operatorname{Ai}(x)$. With this lower bound for the root of the Bessel function, we can obtain
\begin{equation}
    \min_{\nu > -1}  \frac{2(E+\nu+1)}{j_{\nu,1}} \leq \min_{\nu \geq 0}  \frac{2(E+\nu+1)}{\nu-\frac{a_1}{2^{1 / 3}} \nu^{1 / 3}}.
\end{equation}
For convenience, let $a = -a_1/2^{1/3} \approx 1.85576$. The minimum of the right hand side is attained when
\begin{equation}
2a\nu - 3(E+1)\nu^{2/3} - a(E+1) = 0.
\end{equation}
Although the above equation has an analytic solution, its form will be too complex. Instead, we can only pick a suboptimal $\nu$ to derive a bound with a simple form. Since the scaling of the ultimate precision with respect to the energy $E$ is more concerned, by considering the asymptotic behavior of the solution of the above equation as $E \rightarrow \infty$, we can obtain a suboptimal choice
\begin{equation}
\hat{\nu} = \left(\frac{3E+3}{2a}\right)^3
\end{equation}
that approximates the optimal solution when $E$ is large. Then, the average cost will be lower bounded by
\begin{equation}
\begin{aligned}
    2 - \min_{\nu \geq 0} \frac{2(E+\nu+1)}{\nu + a \nu^{1 / 3}} 
    \geq& 2 - \frac{2(E+\hat{\nu}+1)}{\hat{\nu} + a \hat{\nu}^{1 / 3}} \\
    =& 2 \frac{(a \hat{\nu}^{1/3} - E - 1)}{\hat{\nu} + a \hat{\nu}^{1 / 3}} \\
    =& \frac{E+1}{\left(\frac{3E+3}{2a}\right)^3 + \frac{3E+3}{2}} \\
    \sim& \frac{8a^3}{27} \frac{1}{E^2}.
\end{aligned}
\end{equation}

\subsection{Attainability of the precision scaling with respect to energy} \label{sec:proof_attainability}
Having derived a lower bound scaling as $1/E^2$, we show this can be saturated with a constructed strategy in this part. 

For a $d$-dimensional phase channel, we use a probe state
\begin{equation}
\ket{\psi_d} = \sqrt{\frac{2}{d}} \sum_{j=0}^{d-1} \sin\left(\frac{\pi j}{d}\right) \ket{j}.
\end{equation}
After the phase channel, we take a measurement on the Fourier basis. The corresponding quantum comb of this strategy is 
\begin{equation}
T = \ket{\psi_d}\bra{\psi_d} \otimes \left[(I \otimes F)\ket{\Omega}\bra{\Omega}(I \otimes F^\dagger)\right],
\end{equation}
where $F$ is quantum Fourier transform and $\ket{\Omega}\bra{\Omega} = \sum_{j=0}^{d-1} \ket{jj}\bra{jj}$ is the unnormalized maximally entangled state. The quantum instrument associated with measurement outcome $x$ is
\begin{equation}
\begin{aligned}
    T_x =& \ket{\psi_d}\bra{\psi_d} \otimes \left[(I \otimes \bra{x}F)\ket{\Omega}\bra{\Omega}(I \otimes F^\dagger\ket{x})\right] \\
    =& \ket{\psi_d}\bra{\psi_d} \otimes \frac{1}{d} \sum_{j,k=0}^{d-1} \ee^{-\frac{2\pi\ii x(j-k)}{d}} \ket{j}\bra{k}.
\end{aligned}
\end{equation}
With the Choi operator of the phase channel
\begin{equation}
C_{d,\theta} = \sum_{j,k=0}^{d-1} \ee^{\ii(j-k)\theta} \ket{j}\bra{k} \otimes \ket{j}\bra{k},
\end{equation}
we have 
\begin{equation}
\JJ_d := \frac{1}{2\pi} \int_\theta \ee^{-\ii\theta} C_{d,\theta}^T \dd\theta = \sum_{j=1}^{d-1} \ket{j-1,j-1}\bra{j,j}.
\end{equation}
The optimal average cost of this strategy is
\begin{equation}\label{eq:high_dim_optimal_cost}
    \begin{aligned}
        \bar{c}(T) 
        =& 2 - 2\sum_{x=0}^{d-1} |\Tr[T_x \JJ_d]| \\
        =& 2 - 2\sum_{x=0}^{d-1} \left| \sum_{j=1}^{d-1} \frac{2}{d^2} \sin\left(\frac{\pi(j-1)}{d}\right) \sin\left(\frac{\pi j}{d}\right) \ee^{-\frac{2\pi\ii x}{d}} \right| \\
        =& 2 - \frac{4}{d} \left| \sum_{j=1}^{d-1} \sin\left(\frac{\pi(j-1)}{d}\right) \sin\left(\frac{\pi j}{d}\right) \right| \\
        =& 2- 2 \cos(\pi / d) \sim \pi^2 / d^2.
    \end{aligned}
\end{equation}

Meanwhile, we determine the energy consumption of this strategy. The preparation of $\ket{\psi_d}$ needs an energy cost 
\begin{equation}\label{eq:high_dim_state_prepare_energy}
    \Tr[H_d\ket{\psi_d}\bra{\psi_d}] = \frac{2}{d} \sum_{j=0}^{d-1} j \sin^2\left(\frac{\pi j}{d}\right) = d/2.
\end{equation}
To perform the final measurement on the Fourier basis, by Thm.~\ref{thm:energy_constrained_comb}, 
it requires the energy given as the maximal eigenvalue of
\begin{equation}\label{eq:high_dim_measurement_energy}
\begin{aligned}
    & \Tr_2[(I \otimes H_d) (I \otimes F)\ket{\Omega}\bra{\Omega}(I \otimes F^\dagger)] - H_d^T \\
    =& \frac{1}{d}\sum_{j,k=0}^{d-1} \sum_{n=0}^{d-1} n\ee^{\frac{2\pi\ii(j-k)n}{d}} \ket{j}\bra{k} - H_d .
\end{aligned}
\end{equation}
Since $H_d = \sum_{n=0}^{d-1}n\ket{n}\bra{n} \succeq 0$, the maximal eigenvalue of \eqref{eq:high_dim_measurement_energy} will not surpass the maximal eigenvalue of
\begin{equation}
    \frac{1}{d}\sum_{j,k=0}^{d-1} \sum_{n=0}^{d-1} n\ee^{\frac{2\pi\ii(j-k)n}{d}} \ket{j}\bra{k},
\end{equation}
which can be determined as the maximal eigenvalue of
\begin{equation}
\begin{aligned}
    & F\left[ \frac{1}{d}\sum_{j,k=0}^{d-1} \sum_{n=0}^{d-1} n\ee^{\frac{2\pi\ii(j-k)n}{d}} \ket{j}\bra{k} \right] F^\dagger \\
    =& \frac{1}{d^2} \sum_{n,m,m'=0}^{d-1} \sum_{j,k=0}^{d-1} n \ee^{\frac{2\pi\ii(m+n)j}{d}} \ee^{-\frac{2\pi\ii(m'+n)k}{d}} \ket{m}\bra{m'} \\
    =& \sum_{n=1}^{d-1} n \ket{d-n}\bra{d-n}.
\end{aligned}
\end{equation}
Thus, the energy consumption of the final measurement is at most $d-1$. Combining this with the energy consumption of the state preparation given by \eqref{eq:high_dim_state_prepare_energy}, the total energy consumption is at most $3d/2-1 \sim 3d/2$.

Therefore, if we choose $d \sim 2E/3$, by \eqref{eq:high_dim_optimal_cost}, the constructed strategy will attain an average cost scaling as $1/E^2$.

\section{Different Battery Models for Causal Superposition Processes}\label{sec:separation_battery}
\subsection{Spacetime-sharing Battery}
In addition to the spacetime-individual battery model discussed in the main text, one may also consider a model where energy is transferable between spacetime branches.

When the energy flow between spacetime is allowed, the batteries from all branches together can be regarded as a common global (or local) battery with an initial energy $E$, which is then shared across the entire spacetime superposition. The \emph{spacetime-shared battery} model is to support the entire process (i.e., the shadow part in Fig.~\ref{fig:causal_superposition_circuit}) 
with a single battery. For the causal order $\tau$, the component $T^\tau$ related to $\tau$ in the causal superposition is implemented by a comb $\iota_\tau(T^\tau)$ in $\Comb\left[\left(\mathcal{H}_0, \mathcal{H}_1 \right), \dots, \left(\mathcal{H}_{2N-2}, \mathcal{H}_{2N-1}\right)\right]$. The energy constraint is thus
\begin{equation}\label{eq:spacetime_shared_battery}
E_{\text{sh,gl/loc}}(T) := E_{\text{gl/loc}}\left( \sum_{\tau} p_\tau \iota_\tau(T^\tau) \right) \leq E.
\end{equation}
Here $\iota_\tau$ denotes an isomorphic mapping from the comb set $\Comb[(\HH_0, \HH_{2\tau(1)-1}), \dots, (\HH_{2\tau(N-1)},\HH_{2N-1})]$ to the comb set $\Comb\left[\left(\mathcal{H}_0, \mathcal{H}_1 \right), \dots, \left(\mathcal{H}_{2N-2}, \mathcal{H}_{2N-1}\right)\right]$. This isomorphism is induced by the system permutation on the underlying spaces, determined by $\HH_{2\tau(n)-1} \to \HH_{2n-1}$ and $\HH_{2\tau(n)} \to \HH_{2n}$.

\subsection{Separation between Battery Models}
We examine several concrete examples to illustrate the distinctions between the battery models. We first provide examples to show the separation in the energy consumption between the models for a given process. We then demonstrate a precision separation in energy-constrained metrological tasks, which arises from the differences in energy efficiency.

For generality, we primarily consider causal superposition processes here. Although the example provided for the separation between the global battery and the local battery models is based on a definite causal order process, the discussion applies to both definite causal order processes and causal superposition processes, as definite causal orders can be viewed as special cases of causal superpositions that involve only a single causal order. The battery models for causal superposition processes are shown in Fig.~\ref{fig:causal_battery}. 

The separation comes from two factors: whether energy is transferable between steps (global vs. local) and between the branches of spacetime (spacetime-sharing vs. spacetime-individual). Specifically, the local battery model allocates energy across the steps before the process begins and forbids the transfer between the steps during the process. In contrast, the global battery model permits the transfer between the steps during the process, eliminating the need for advance allocation. Additionally, when the spacetime branches, the spacetime-individual battery branches as well, whereas the spacetime-sharing battery does not. 

\begin{figure}[h]
    \centering
    \includegraphics[width=0.65\linewidth]{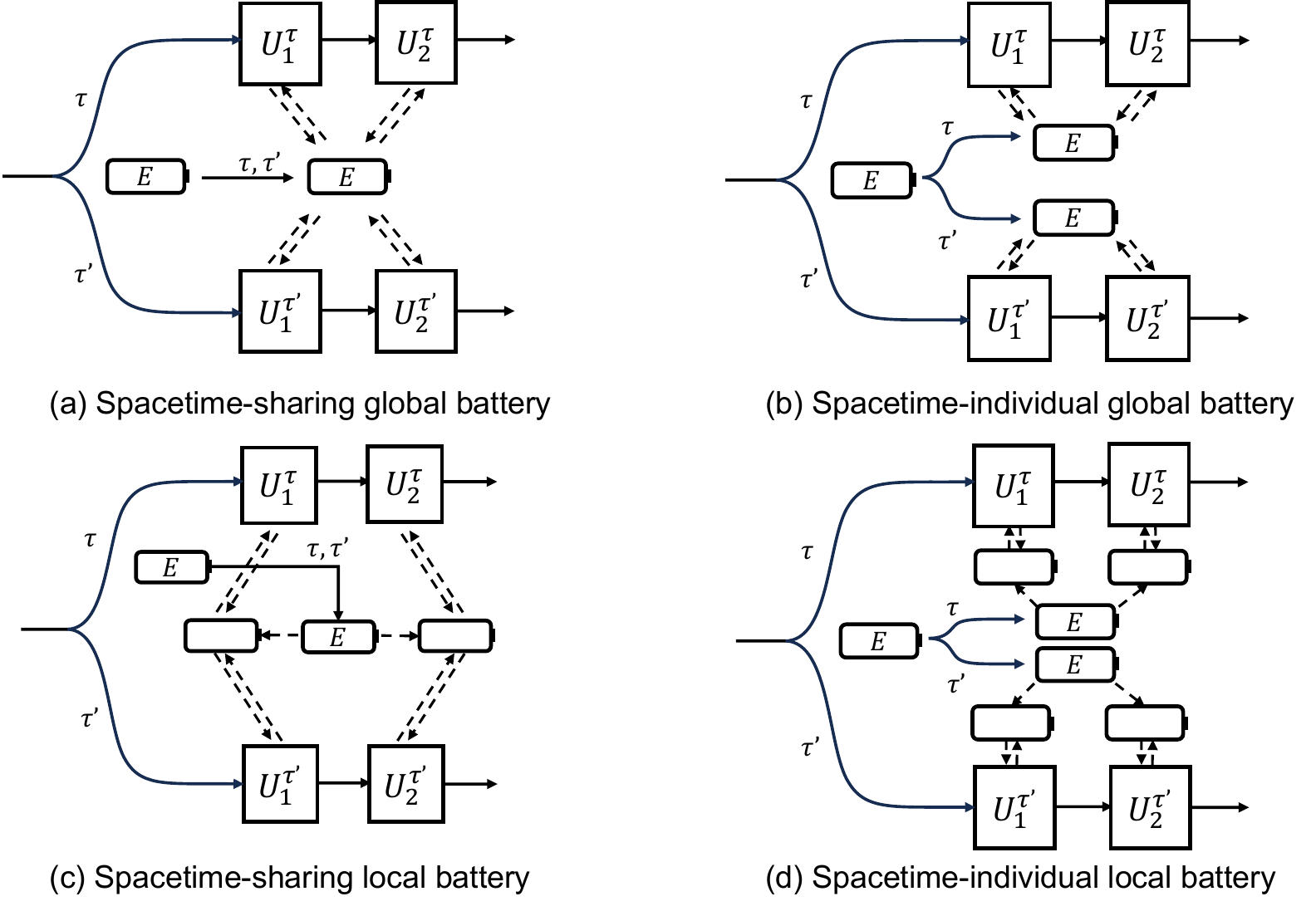}
    \caption{The possible battery models for the causal superposition. The solid arrows represent the time flow in the spacetime corresponding to the causal orders $\tau,\tau'$. The dashed arrows represent the directions of the energy flow.}
    \label{fig:causal_battery}
\end{figure}


\subsubsection{Global vs. Local}
The global battery offers an advantage by enabling adaptive energy distribution between different steps, transferring energy from less consuming steps to those with higher consumption during the process. It achieves energy savings for a process by reducing waste and maximizing the utility of the supplied energy. Since this separation emerges in definite causal order processes, we mainly focus on definite causal order here for convenience.

However, when the maximum energy cost can be attained for all the steps at the same time, the global battery provides no advantage, as there is no extra energy to be transferred. This implies that in a Markovian process, the global battery cannot outperform the local one. This is because, for a Markovian process, each step is independent of the inputs to the other steps, which enables us to choose the most energy-consuming input for each step simultaneously.

Therefore, we consider a two-step non-Markovian process. The process is controlled by a qubit with an arbitrarily given initial state. We suppose the Hamiltonian of the control qubit is degenerate. When the control qubit is $\ket{0}_c$, the first step is a Hadamard gate $\mathbf{H}$ and the second step is an identity operation $I$. When the control qubit is $\ket{1}_c$, the first step is $I$ and the second step is $\mathbf{H}$.

In this process, since the position where the Hadamard gate $\mathbf{H}$ acts cannot be known in advance, the only feasible approach for the local battery is to supply the energy $E_\mathbf{H}$ required by $\mathbf{H}$ in both steps, resulting in a total energy consumption of $2E_\mathbf{H}$. In contrast, if a global battery is available to transfer energy, then only $E_\mathbf{H}$ energy is sufficient, because only one $\mathbf{H}$ is applied in each component of the control ancilla.

This can be formally verified by the characterization of energy consumption in Thm.~\ref{thm:energy_constrained_comb} in the main text. Let $C_\mathbf{H} = (I \otimes \mathbf{H})\ket{\Omega}\bra{\Omega}(I \otimes \mathbf{H}^\dagger)$ be the Choi operator of a single Hadamard gate, where $\ket{\Omega} = \ket{00} + \ket{11}$ refers to the unnormalized maximally entangled state. If the system Hamiltonian is $H=\ket{1}\bra{1}$, the cost of a single Hadamard gate is 
\begin{equation}\label{eq:cost_of_Hadamard}
    E_{\mathbf{H}} = \lambda_{\max}(\Tr_1[(I \otimes H)C_{\mathbf{H}}] - H^T) = \sqrt{2}/2.
\end{equation}
The quantum comb (i.e., the Choi operator) of the constructed process is
\begin{equation}
    C_1 = \sum_{j,k=0}^1 \ket{j}_c\ket{j}_{c'}\bra{k}_c\bra{k}_{c'} \otimes \left[U_j (\ket{\Omega}\bra{\Omega})^{\otimes 2} U_k^\dagger \right],
\end{equation}
with $U_0 = I \otimes \mathbf{H} \otimes I \otimes I$, $U_1 = I \otimes I \otimes I \otimes \mathbf{H}$.
Within the local battery model, the energy consumption of $C_1$ is
\begin{equation}
    \begin{aligned}
        & E_{\text{loc}}(C_1) \\
        =& \max\{ \lambda_{\max}(\Tr_{c',1,3}[(I_{c,c'} \otimes I \otimes H \otimes I \otimes I)C_1] - I_c \otimes H^T \otimes I), 0\}  \\
        &\quad +\max\{ \lambda_{\max}(\Tr_{c',1,3}[(I_{c,c'} \otimes I \otimes I \otimes I \otimes H)C_1] - I_c \otimes I \otimes H^T), 0\} \\
        =& \sqrt{2} \\
        =& 2E_{\mathbf{H}}.
    \end{aligned}
\end{equation}
Within the global battery model, the energy consumption of $C_1$ is
\begin{equation}
    \begin{aligned}
        & E_{\text{gl}}(C_1) \\
        =& \max\{ \lambda_{\max}(\Tr_{c',1,3}[(I_{c,c'} \otimes I \otimes H \otimes I \otimes I)C_1] - I_c \otimes H^T \otimes I),  \\
        &\quad  \lambda_{\max}(\Tr_{c',1,3}[(I_{c,c'} \otimes I \otimes H \otimes I \otimes I + I_{c,c'} \otimes I \otimes I \otimes I \otimes H)C_1] - I_c \otimes H^T \otimes I- I_c \otimes I \otimes H^T)\} \\
        =& \sqrt{2}/2 \\
        =& E_{\mathbf{H}}.
    \end{aligned}
\end{equation}

Applying the models to metrology, we can obtain an energy-efficiency advantage for the global battery. For instance, it can be illustrated by the example presented in Fig.~\ref{fig:hierarchy} in the main text.

\subsubsection{Spacetime-sharing vs. Spacetime-individual}
The spacetime-sharing battery is stronger than the spacetime-individual battery, as it allows the energy to flow between the branches of spacetime. As a concrete example, consider a single-step operation in a superposition of two causal orders, where the state of the order-control system is a uniform superposition qubit $\frac{\ket{0}_{\text{oc}} + \ket{1}_{\text{oc}}}{\sqrt{2}}$. In the first causal order $\tau$, a Hadamard gate $\mathbf{H}$ is applied, while in the second order $\tau'$, the identity operation $I$ is applied. This operation is
\begin{equation}
    \mathbf{H} \otimes \ket{0}\bra{0}_{\text{oc}} + I \otimes \ket{1}\bra{1}_{\text{oc}}.
\end{equation}
As this is a single-step process, there is no difference between the global battery and the local battery. Therefore, we assume the discussed battery is global for convenience. For a spacetime-individual battery, the energy flow between the spacetime branches is prohibited, and thus the energy required for this operation equals that of a single Hadamard gate. This is because the spacetime of order $\tau$ requires the amount of energy $E_\mathbf{H}$ to implement $\mathbf{H}$, and the energy in $\tau'$ remains inaccessible for $\tau$ even if the identity operation consumes no energy. However, for the spacetime-sharing battery, the energy from $\tau'$ can be transferred to $\tau$, reducing the total energy requirement by half to $E_\mathbf{H}/2$.

Formally, with the fixed state $\frac{\ket{0}_{\text{oc}} + \ket{1}_{\text{oc}}}{\sqrt{2}}$ of the order-control system, tracing out the order-control system in the Choi operator $C_2$ of the constructed process, we obtain an ensemble given by $C^\tau_2 = C_{\mathbf{H}}$ and $C^{\tau'}_2 = \ket{\Omega}\bra{\Omega}$ with $p_\tau=p_{\tau'}=1/2$. For a spacetime-individual battery, the energy consumption is
\begin{equation}
    E_{\text{ind,gl}}(C_2) = \max\{E_{\text{gl}}(C^\tau_2), E_{\text{gl}}(C^{\tau'}_2) \} = E_{\mathbf{H}}.
\end{equation}
For a spacetime-sharing battery, the energy consumption is
\begin{equation}
    E_{\text{sh,gl}}(C_2) = \frac{1}{2} E_{\text{gl}}(C^\tau_2) + \frac{1}{2}E_{\text{gl}}(C^{\tau'}_2)  = E_{\mathbf{H}}/2.
\end{equation}

To see their separation in metrology, we consider an example where two channels are $\mathrm{e}^{-\mathrm{i}\theta Z/2}$ and $\mathrm{e}^{\mathrm{i}\theta X/2}$, both preceded by a bit-flip channel $\rho \to \frac{1}{2}\rho + \frac{1}{2}X\rho X$. With the Hamiltonian $H=\ket{1}\bra{1}$, if the different causal orders are performed in a uniform superposition, then the optimal precision with respect to the provided energy $E$ is depicted in Fig.~\ref{fig:ind_vs_sh}.
\begin{figure}
    \centering
    \includegraphics[width=0.6\linewidth]{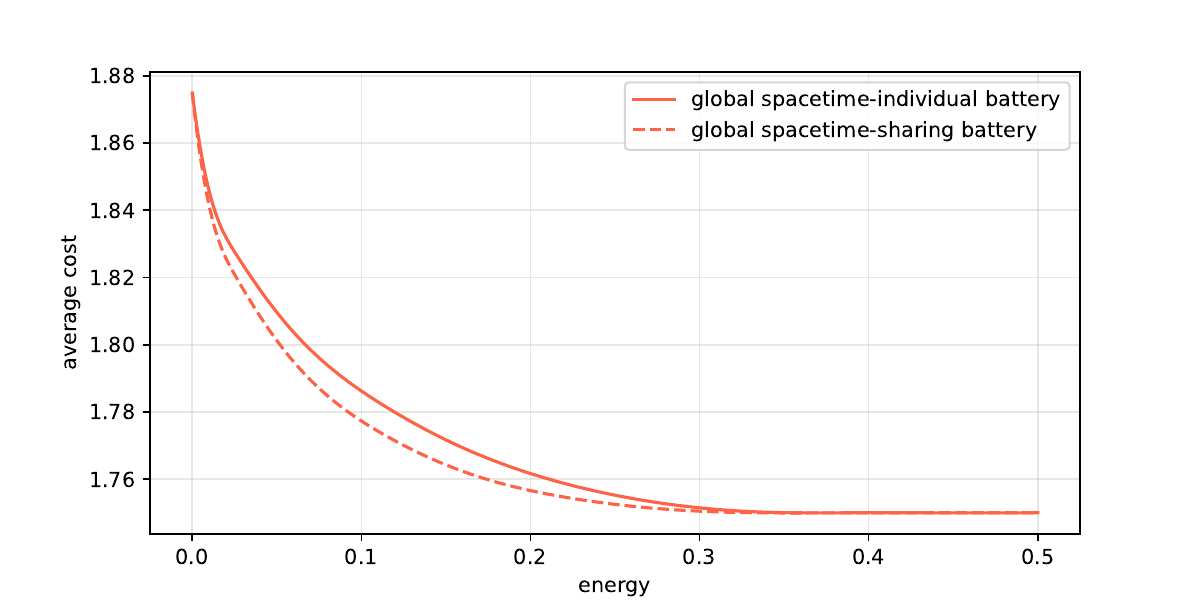}
    \caption{The advantage of the spacetime-sharing battery over the spacetime-individual battery.}
    \label{fig:ind_vs_sh}
\end{figure}

\section{Numerical Simulation for Fisher Information} \label{sec:Fisher_information}
In this section, we carry out a numerical simulation for the energy-constrained Fisher information for the qubit phase channel, demonstrating the degradation of the optimal variance in the low-energy regime.

We first derive a numerical method to optimize the Fisher information under energy constraints in the qubit case. Recall that, in the energy-constrained framework, to constrain the energy consumption of the final measurement as well, we have delegated the energy-consuming operations to the probe process $T$ and fixed the final measurement as a projection in the energy eigenbasis. In other words, involving the final measurement, the probe $T$ should be diagonal on the final output space, i.e., $T = \sum_x T_x \otimes \ket{x}\bra{x}$ with the energy eigenbasis $\{\ket{x}\}_x$ and the quantum instrument $\{T_x\}_x$. If we apply the probe process $T$ to a parameterized process $C_\theta$, the final output state $\rho_\theta := T * C_\theta$ of the combined process is actually the post-measurement state, which is a diagonal state $\rho_\theta = \sum_x p(x|\theta) \ket{x}\bra{x}$ with the measurement outcome distribution $\{p(x|\theta)\}_x$. Therefore, the quantum Fisher information $\mathcal{I}(\rho_\theta)$ of $\rho_\theta$ equals the classical Fisher information of the outcome distribution given by
\begin{equation}
\mathcal{I}_c(\{p(x|\theta)\}_x) = \sum_x \frac{[\partial p(x|\theta)/\partial\theta]^2}{p(x|\theta)}.
\end{equation}
For the qubit case, the measurement outcome $x$ is in $\{0,1\}$. It follows that $p(0|\theta) + p(1|\theta) = 1$ and $\frac{\partial p(0|\theta)}{\partial \theta} + \frac{\partial p(1|\theta)}{\partial \theta} = 0$. Thus,
\begin{equation}
\mathcal{I}_c(\{p(x|\theta)\}_x) = \left( \frac{\partial p(0|\theta)}{\partial \theta} \right)^2 \left( \frac{1}{p(0|\theta)} + \frac{1}{p(1|\theta)} \right).
\end{equation}
With $p(x|\theta) = T_x * C_\theta$, by the linearity of the link product, we have $\frac{\partial p(x|\theta)}{\partial \theta} = T_x * \frac{d C_\theta}{d \theta}$. Defining $p_0 := p(0|\theta)$, the maximization of $\mathcal{I}(\rho_\theta) = \mathcal{I}_c(\{p(x|\theta)\}_x)$ can be reformulated as
\begin{equation}
\begin{aligned}
    \max_{0 < p_0 < 1} \max_{T} \quad &\left( T_0 * \frac{d C_\theta}{d \theta} \right)^2 \left( \frac{1}{p_0} + \frac{1}{1-p_0} \right) \\
    \text{s.t.} \quad & T \in \Comb^{\leq E}_{\text{gl/loc}}, \\
    & T_0 * C_\theta = p_0.
\end{aligned}
\end{equation}
If the value of $p_0$ is fixed, this can be decomposed into two subproblems with objectives $\pm (T_0 * \frac{d C_\theta}{d \theta})$, which can be solved by semidefinite programming. Therefore, we can obtain a (nearly) optimal choice by iterating over different values of $p_0$ and solving the corresponding subproblems efficiently with semidefinite programming.

\begin{figure}[htbp]
    \centering
    \includegraphics[width=0.45\linewidth]{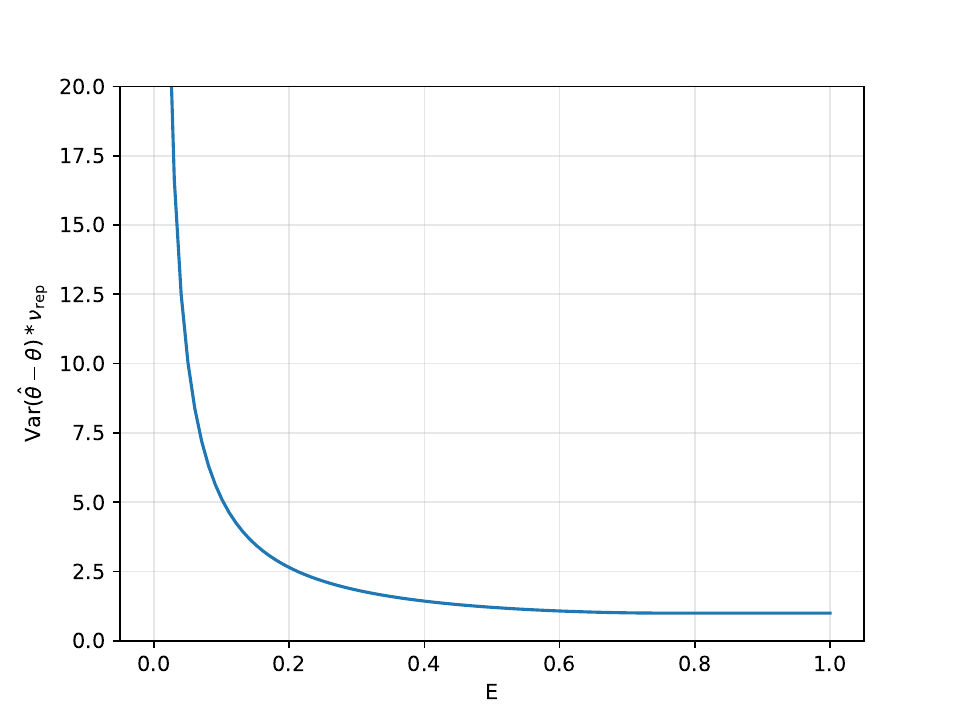}
    \caption{The reciprocal of the Fisher information (equals to $\Var(\hat{\theta}-\theta)\cdot\nu_{\rm rep}$ when the number of repetitions $\nu_{\text{rep}} \gg 1$ according to the single-parameter Cram\`{e}r-Rao bound) for the qubit phase estimation with an energy budget $E$. The probe process $T$ is in the $E$-globally constrained comb set $\Comb^{\leq E}_{\text{gl}}$.}
    \label{fig:fisher_information}
\end{figure}

The numerical simulation results for qubit phase estimation with $\theta = \pi$ are shown in Fig.~\ref{fig:fisher_information}. The optimal variance increases rapidly when the energy $E$ is insufficient, reflecting the limitation of the unbiased estimator under energy-constrained conditions.

However, for a general qudit system, the maximization of the Fisher information is difficult to solve directly via convex optimization methods due to the non-concavity of the Fisher information.


\end{document}